%% file: quasilocalmass.tex
\documentclass[11pt]{article} 

\usepackage[utf8]{inputenc} 

\usepackage[margin=1in]{geometry}
\usepackage{graphicx} 

\usepackage{booktabs} 
\usepackage{array} 
\usepackage{paralist} 
\usepackage{verbatim} 
\usepackage{subfig} 
\usepackage{transparent}
\usepackage{color}
\usepackage{fancyhdr} 
\usepackage{sectsty}
\allsectionsfont{\sffamily\mdseries\upshape} 

\usepackage[nottoc,notlof,notlot]{tocbibind} 
\usepackage[titles,subfigure]{tocloft} 

\usepackage{amsmath,amsthm}
\usepackage{amssymb}
\usepackage{slashed}
\usepackage{accents}
\usepackage{wrapfig}
\newcommand{\ubar}[1]{\underaccent{\bar}{#1}}
\newcommand{\s}{\slashed}

\newtheorem{theorem}{Theorem}[section]
\newtheorem{proposition}[theorem]
{Proposition}
\newtheorem{lemma}{Lemma}[subsection]
\newtheorem{corollary}{Corollary}[theorem]
\newtheorem{definition}{Definition}[section]
\newtheorem{remark}{Remark}[section]

\DeclareMathOperator{\tr}{tr}
\DeclareMathOperator{\II}{II}

\title{Proof of a Null Penrose Conjecture using a new Quasi-local Mass}
\author{Henri Roesch}
\date{\today} 

\begin{document}
\maketitle

\begin{abstract}
We define an explicit quasi-local mass functional which is nondecreasing along \textit{all} foliations (satisfying a convexity assumption) of null cones. We use this new functional to prove the null Penrose conjecture under fairly generic conditions.
\end{abstract}

\tableofcontents

\section{Introduction}
A spacetime $(M,g)$ is defined to be a four dimensional smooth manifold $M$ equipped with a metric $g(\cdot,\cdot)$ (or $\langle\cdot,\cdot\rangle$) of Lorentzian signature $(-,+,+,+)$. We assume that the spacetime is time oriented, i.e. admits a nowhere vanishing timelike vector field, defined to be future-pointing.\\
\indent Throughout this paper, we will denote by $\Sigma$ a spacelike embedding of a sphere in $M$ with induced metric $\gamma$. It is well known that $\Sigma$ has trivial normal bundle $T^{\perp}\Sigma$ with induced metric of signature $(-,+)$. From any choice of null section $\ubar L\in \Gamma(T^{\perp}\Sigma)$, we have a unique null partner section $L\in\Gamma(T^{\perp}\Sigma)$ satisfying $\langle \ubar L,L\rangle = 2$, providing $T^{\perp}\Sigma$ with a null basis $\{L,\ubar L\}$. We also notice that any `boost' $\{\ubar L,L\}\to\{\ubar L_a,L_a\}$ given by:
$$\ubar L_a:= a\ubar L,\,\,\, L_a:= \frac{1}{a}L$$ 
(for $a\in\mathcal{F}(\Sigma)$ a non-vanishing smooth function on $\Sigma$) gives $\langle \ubar L_a,L_a\rangle=\langle \ubar L,L\rangle=2$ as well.\\
\indent Our convention for the second fundamental form II and mean curvature $\vec{H}$ of $\Sigma$ are
$$\text{II}(V,W) = D^{\perp}_VW,\,\,\,\,\vec{H} = \tr_\Sigma\text{II}$$
for $V,W\in \Gamma(T\Sigma)$ and $D$ the Levi-Civita connection of the spacetime. 

\begin{figure}[h]
\centering
\def\svgwidth{400pt} 
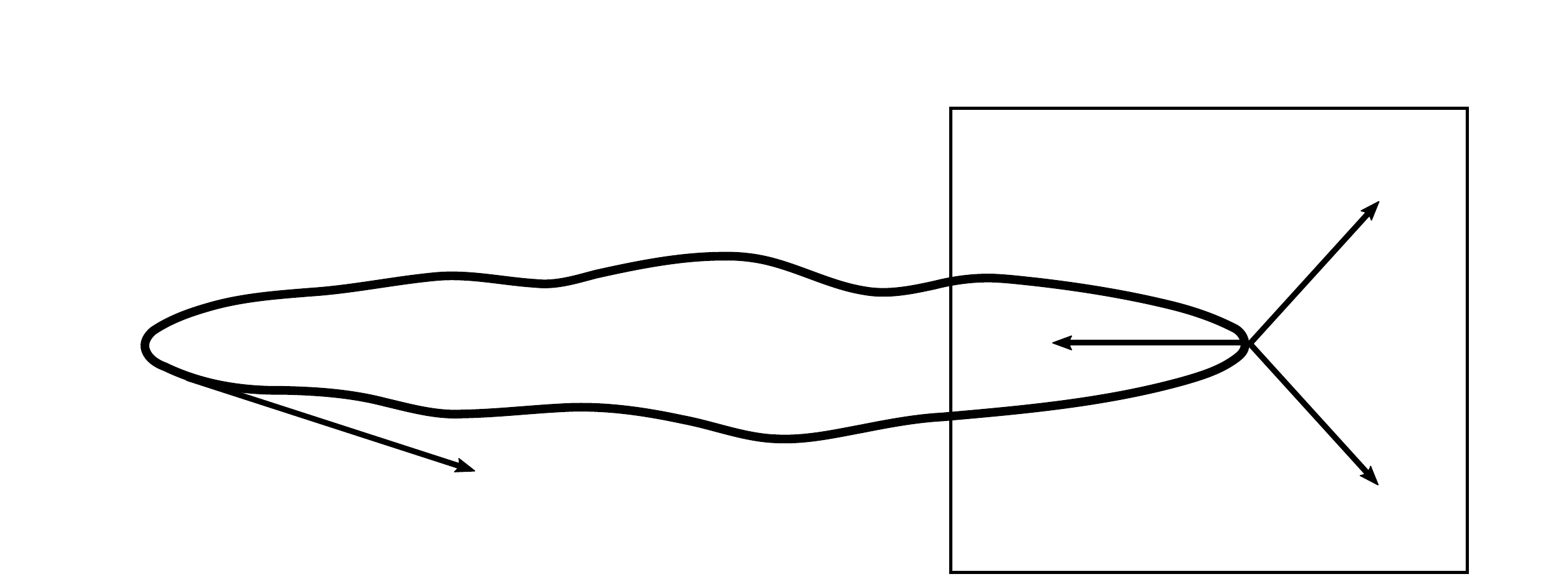
\end{figure}

\begin{definition} 
\label{d1}
Given a choice of null basis $\{\ubar L,L\}$, following the conventions of Sauter \cite{S}, we define the associated symmetric 2-tensors $\ubar\chi,\chi$ and torsion (connection 1-form) $\zeta$ by
\begin{align*}
\ubar\chi(V,W) &:= \langle D_V\ubar L,W\rangle = -\langle \ubar L,\II(V,W)\rangle\\
\chi(V,W) &:=\langle D_VL,W\rangle = -\langle L,\II(V,W)\rangle\\
\zeta(V) &:= \frac12\langle D_V\ubar L,L\rangle=-\frac12\langle D_VL,\ubar L\rangle
\end{align*}
where $V,W\in \Gamma(T\Sigma)$.
\end{definition}
Denoting the exterior derivative on $\Sigma$ by $\s{d}$, any boosted basis $\{\ubar L_a,L_a\}$ produces the associated tensors of Definition \ref{d1}:
\begin{align*}
\ubar\chi_a(V,W)&:=\langle D_V(a\ubar L),W\rangle = a\ubar\chi(V,W)\\
\chi_a(V,W)&:=\langle D_V(\frac{1}{a}L),W\rangle = \frac{1}{a}\chi(V,W)\\
\zeta_a(V)&:=\frac12\langle D_V(a\ubar L),\frac{1}{a}L\rangle = \zeta(V)+V\log|a|=(\zeta+\s{d}\log|a|)(V).
\end{align*}
For a symmetric 2-tensor $T$ on $\Sigma$ its \textit{trace-free} (or \textit{trace-less}) part is given by
$$\hat{T}:=T-\frac12(\tr_\gamma T)\gamma$$
allowing us to decompose $\ubar\chi$ into its \textit{shear} and \textit{expansion} components respectively:
$$\ubar\chi = \hat{\ubar\chi}+\frac12(\tr\ubar\chi)\gamma.$$
\begin{definition}
\label{d2}
We say $\Sigma$ \underline{is expanding along $\ubar L$} for some null section $\ubar L\in\Gamma(T^{\perp}\Sigma)$ provided that,
\begin{equation}
\langle-\vec{H},\ubar L\rangle = \tr\ubar\chi>0\tag{\dag}
\end{equation}
on all of $\Sigma$. 
\end{definition} 
Any infinitesimal flow of $\Sigma$ along $\ubar L$ gives, by first variation of area, $\dot{dA} = \langle-\vec{H},\ubar L\rangle dA= \tr\ubar\chi dA$. So the flow is locally area expanding ($\dot{dA}>0$) only if $\Sigma$ ``is expanding along $\ubar L$": 

\begin{figure}[h]
\centering
\def\svgwidth{350pt} 
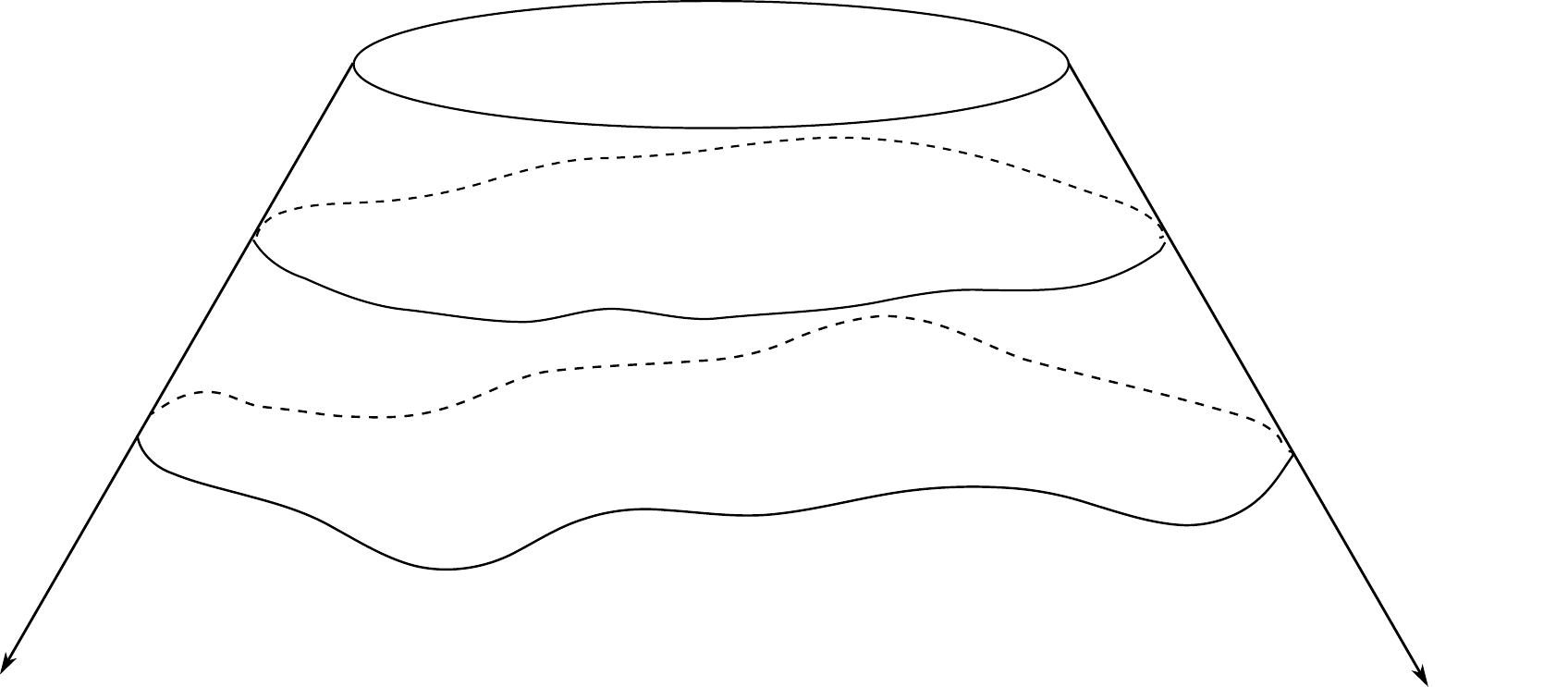
\end{figure}

\begin{remark}\label{r1}
In Section 4 we will show (Lemma \ref{l15}), whenever $\Omega$ is past asymptotically flat inside a spacetime satisfying the null energy condition, a consequence of the famous Raychaudhuri equation ((8), Section 3) is that any cross section $\Sigma\hookrightarrow\Omega$ is expanding along the past pointing null section $\ubar L\in \Gamma(T^\perp\Sigma)\cap\Gamma(T\Omega)|_{\Sigma}$. So inequality ($\dag$) holds for any foliation of $\Omega$ along $\ubar L_a$ where $a>0$ and we have an expanding null cone (as illustrated in the figure above).
\end{remark}
For $\Sigma$ expanding along some $\ubar L\in\Gamma(T^\perp\Sigma)$ we are able to choose a canonical null basis $\{L^-,L^+\}$ by requiring that our flow along $L^- = a\ubar L$ be uniformly area expanding ($\dot{dA} = dA$). From first variation of area, flowing along $a\ubar L$ gives
$$\dot{dA} = -\langle\vec{H},a\ubar L\rangle dA = a\tr\ubar\chi dA.$$
So we achieve a uniformly area expanding null flow when $a = \frac{1}{\tr\ubar\chi}$ giving:
\begin{definition}\label{d3} 
For $\Sigma$ expanding along some $\ubar L\in\Gamma(T^\perp\Sigma)$ we call the associated canonical uniformly area expanding null basis $\{L^-,L^+\}$ given by
$$L^-:=\frac{\ubar L}{\tr\ubar\chi},\,\,\, L^+:=\tr\ubar\chi L$$
the \underline{null inflation basis}.\\
We also define $\chi^{-(+)} := -\langle \II,L^{-(+)}\rangle$. It follows from the comments proceeding Definition \ref{d1} that 
\begin{align*}
\tr\chi^- &= 1\\ 
\tr\chi^+ &= \tr\ubar\chi\tr\chi = \langle\vec{H},\vec{H}\rangle
\end{align*}
and for $V\in\Gamma(T\Sigma)$ the torsion associated to this basis is given by
$$\tau(V)=\frac12\langle D_VL^-,L^+\rangle = (\zeta-\s{d}\log\tr\ubar\chi)(V).$$
\end{definition} 
We will denote the induced covariant derivative on $\Sigma$ by $\s\nabla$.
\begin{definition}\label{d4}
Assuming $\Sigma$ is expanding along $\ubar L$, for some $\ubar L\in\Gamma(T^\perp\Sigma)$, we define the geometric flux function
\begin{equation}
\rho = \mathcal{K}-\frac14\langle\vec{H},\vec{H}\rangle+ \s\nabla\cdot\tau
\end{equation}
where $\mathcal{K}$ represents the Gaussian curvature of $\Sigma$.\\
This allows us to define the associated quasi-local mass
\begin{equation}
m(\Sigma) = \frac12\Big(\frac{1}{4\pi}\int_\Sigma\rho^{\frac23}dA\Big)^{\frac32}.
\end{equation}
\end{definition}
For the induced covariant derivative $\s\nabla$ we denote the associated Laplacian on $\Sigma$ by $\s\Delta$. 
\begin{remark}\label{r2}
Whenever $\tr\chi^+ = \langle\vec{H},\vec{H}\rangle\neq0$, $\Sigma$ has two null inflation bases given by $\{L^-,L^+\}$ and $\{\frac{L^+}{\tr\chi^+},\tr\chi^+L^-\}$. As a result we typically have two distinct flux functions
\begin{align*}
\rho_{-}&=\mathcal{K}-\frac14\langle\vec{H},\vec{H}\rangle+\s\nabla\cdot\tau\\
\rho_{+}&=\mathcal{K}-\frac14\langle\vec{H},\vec{H}\rangle-\s\nabla\cdot\tau-\s\Delta\log|\langle\vec{H},\vec{H}\rangle|
\end{align*}
with associated mass functionals $m_\pm$. For the Bartnik datum $\alpha_H$ (see Definition \ref{d6}), we will see for a past pointing $\ubar L$ that $\rho_--\rho_+ = 2\s\nabla\cdot\alpha_H$ (Lemma \ref{l3}). For $\langle\vec{H},\vec{H}\rangle\neq0$, whenever $\Sigma$ is `time-flat' (i.e. $\s\nabla\cdot\alpha_H=0$) it follows that $\rho_- = \rho_+\implies m_- = m_+$.
\end{remark}
For a normal null flow off of some $\Sigma$ with null flow vector $\ubar L$, technically the flow speed is zero since $\langle \ubar L,\ubar L\rangle = 0$. In the case the $\Sigma$ expands along $\ubar L$ we define the \underline{expansion speed}, $\mathfrak{s}$, according to $\ubar L = \mathfrak{s}L^-$. We notice that $\mathfrak{s} = \tr\ubar\chi$.
We are now ready to state our first result.
\begin{theorem}\label{t1}
Let $\Omega$ be a null hypersurface foliated by spacelike spheres $\{\Sigma_s\}$ expanding along the null flow direction $\ubar L$ such that $|\rho(s)|>0$ for each $s$. Then the mass $m(s):=m(\Sigma_s)$ has rate of change 
$$\frac{dm}{ds} =\frac{(2m)^{\frac13}}{8\pi}\int_{\Sigma_s}\frac{\mathfrak{s}}{\rho^{\frac13}}\Big((|\hat{\chi}^-|^2+G(L^-,L^-))(\frac14\langle\vec{H},\vec{H}\rangle-\s\Delta\log|\rho|^{\frac13})+\frac12|\eta_\rho|^2+G(L^-,N)\Big)dA$$
where
\begin{itemize}
\item $G$ is the Einstein tensor for the ambient metric $g$
\item $\mathfrak{s}=\tr\ubar\chi$ is the expansion speed
\item $\eta_\rho:=2\hat{\chi}^-\cdot\s{d}\log|\rho|^{\frac13}-\tau$
\item $N:=|\s{d}\log|\rho|^{\frac13}|^2L^-+\s\nabla\log|\rho|^{\frac13}-\frac14L^+$
\end{itemize}
\end{theorem}
If we assume therefore that our spacetime $M$ satisfies the null energy condition we can show our mass functional $m(\Sigma_s)$ is non-decreasing for foliations $\{\Sigma_s\}$ satisfying the following convexity condition:
\begin{definition}\label{d5}
Given a foliation of 2-spheres $\{\Sigma_s\}_{s\geq0}$ we say it is a \underline{(P)-foliation} provided: 
\begin{align*}
\rho&>0\\
\frac14\langle\vec{H},\vec{H}\rangle&\geq\s\Delta\log\rho^{\frac13}
\end{align*}
is satisfied on each $\Sigma_s$. We say $\{\Sigma_s\}_{s\geq0}$ is a strict (P)-foliation or \underline{(SP)-foliation} if additionally:
\begin{align*} 
\frac14\langle\vec{H},\vec{H}\rangle&=\s\Delta\log\rho^{\frac13},\,\text{for}\,\,s=0\\
\frac14\langle\vec{H},\vec{H}\rangle&>\s\Delta\log\rho^{\frac13},\,\text{for}\,\,s>0.
\end{align*}
\end{definition}
So for a (P)-foliation the null energy condition ensures the product of the first two terms of the integrand in Theorem \ref{t1} be non-negative. The second is non-negative since each $\Sigma_s$ is spacelike and the last term is non-negative from the null energy condition since $\langle N,N\rangle = 0$ and $\langle N,L^-\rangle = -\frac12<0$ (i.e. $N$ is null and at every point $p\in\Sigma$ lies inside the same connected component of the null cone in $T_pM$ as $L^-$).\\
\indent We will assume in Sections 4 and 5 that $\ubar L$ is past pointing. Adopting the same definitions as Mars and Soria \cite{MS1} (see Section 4.1) we have our second main result:
\begin{theorem}\label{t2}
Let $\Omega$ be a null hypersurface in a spacetime satisfying the null energy condition that extends to past null infinity. Then given the existence of a (P)-foliation $\{\Sigma_s\}$ we have
$$m(0) \leq \lim_{s\to\infty}m(\Sigma_s)=:M$$
(for $M\leq\infty$). If, in addition, $\Omega$ is past asymptotically flat with strong flux decay and $\{\Sigma_s\}$ asymptotically geodesic (see Section 4) then
$$M\leq m_B$$ 
where $m_B$ is the Bondi mass of $\Omega$. Moreover, in the case that $\tr\chi|_{\Sigma_0} = 0$ we have the null Penrose inequality
$$\sqrt{\frac{|\Sigma_0|}{16\pi}}\leq m_B.$$
Furthermore, when equality holds for an (SP)-foliation we conclude that equality holds for all foliations of $\Omega$ and the data ($\gamma$, $\ubar\chi$, $\tr\chi$ and $\zeta$) agree with some foliation of the standard null cone of the Schwarzschild spacetime.
\end{theorem}
\subsection{Background}
An interesting energy functional for a closed spacelike surface $\Sigma$ introduced by Hawking \cite{H} is defined by 
$$E_H(\Sigma) = \sqrt{\frac{|\Sigma|}{16\pi}}\Big(1-\frac{1}{16\pi}\int_\Sigma\langle\vec{H},\vec{H}\rangle dA\Big).$$
Named the \textit{Hawking Energy} this functional provides a measure of the energy content within $\Sigma$. We also notice by the Gauss-Bonnet and Divergence Theorems that
$$\int_\Sigma\rho dA = 8\pi\frac{E_H(\Sigma)}{\sqrt{\frac{|\Sigma|}{4\pi}}}$$
motivating in part why we call $\rho$ a flux function. \\
\indent We find various interesting scenarios where justification for $E_H$ as an energy is given. In Minkowski spacetime if $\Sigma$ is chosen to be any cross-section of the null cone of a point, work of Sauter (\cite{S}, Section 4.5) shows that $E_H(\Sigma)=0$, as expected of a flat vacuum spacetime. For Schwarzschild spacetime, the famous geometry modeling a static isolated black hole of fixed mass $M$, Sauter also shows when $\Sigma$ is any cross-section of the so called `standard null cone' $E_H(\Sigma)\geq M$ with equality if and only if $\Sigma$ is `time-symmetric' (\cite{S}, Lemma 4.4). This is reminiscent of the special relativistic understanding that an energy measurement $E=\sqrt{M^2+|\vec{p}|^2}$ for a particle always over-estimates its mass $M$ except when measured within its rest frame (i.e a frame where $\vec{p}=0$).\\
\indent Although the Hawking Energy enjoys monotonicity and convergence along certain flows, difficulty remains in assigning physical significance to the convergence of $E_H$ due to the lack of control on the asymptotics of such flows \cite{S,MS2}. We expect these difficulties may very well be symptomatic of the fact that an energy functional is particularly susceptible to the plethora of ways boosts can develop along any given flow.

\begin{figure}[h]
\centering
\def\svgwidth{300pt} 
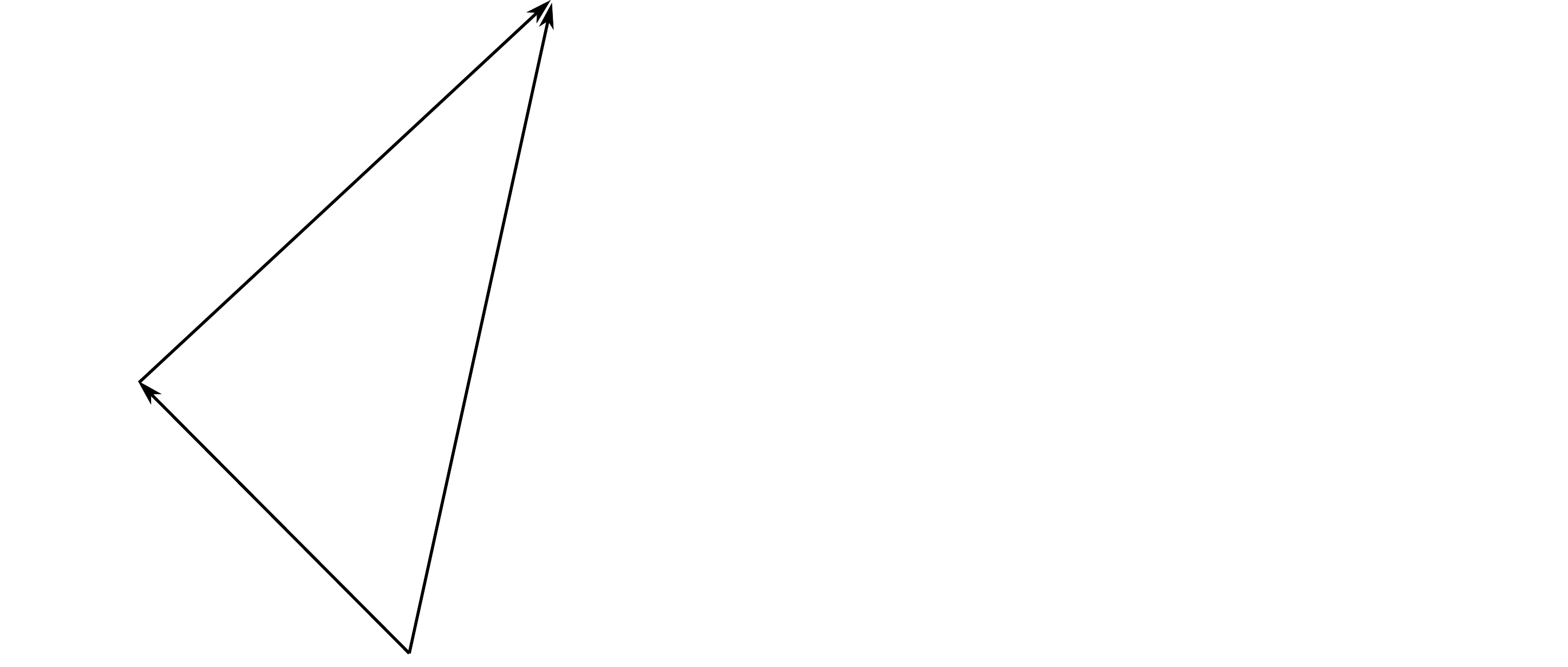
\end{figure}

Analogous to the addition of 4-velocities in special relativity, $P_1+P_2 = P_3\implies E_3 = E_1+E_2$, (as shown below) we expect an infinitesimal null flow of $\Sigma$ within a fixed reference frame to raise energy due to an influx of matter. However, with no a priori knowledge of the flow, we have no way to fix or even identify a reference frame. So it is likely that `phantom energy' will accumulate from infinitesimal boosts along the flow in analogy with special relativistic boosts, $P\to P'$ (i.e energy increases) or $P'\to P$ (i.e. energy decreases) as shown below. Geometrically we expect this to manifest along the flow in a (local) `tilting' of $\Sigma$. One may even expect a net decrease in energy as is evident in Schwarzschild spacetime (recall $E_H(\Sigma)\geq M$). This is not a problem, however, if we appeal instead to mass rather than energy since boosts leave mass invariant, $M^2=E^2-|\vec{p}|^2= (E')^2-|\vec{p}'|^2=(M')^2$. Moreover, by virtue of the Lorentzian triangle inequality (provided all vectors are time-like and either all future or all past pointing), along any given flow the mass should always increase:
$$M_3 = |(E_1+E_2,\vec{p_1}+\vec{p_2})|\geq |(E_1,\vec{p}_1)|+|(E_2,\vec{p}_2)| = M_1+M_2.$$
We hope therefore by appealing instead to a quasi local mass functional a larger class of valid flows and more generic monotonicity should arise. We approach the problem of finding such a mass functional by first finding an optimal choice of flux function for $E_H$.\\
\indent One such flux, first introduced by Christodoulou \cite{C}, is the `mass aspect function'
$$\mu = \mathcal{K}-\frac14\langle\vec{H},\vec{H}\rangle-\s{\nabla}\cdot\zeta$$
associated to an arbitrary null basis $\{\ubar L, L\}\subset \Gamma(T^\perp\Sigma)$. Using $\mu$ in his PhD thesis \cite{S}, Sauter showed the existence of flows on past null cones that render $E_H$ non-decreasing making explicit use of the fact that under a boost this mass aspect function changes via $\zeta$ according to
$$\zeta\to\zeta_a = \zeta+\s{d}\log|a|\implies \mu\to\mu_a = \mathcal{K}-\frac14\langle\vec{H},\vec{H}\rangle-\s\nabla\cdot\zeta-\s\Delta\log|a|.$$
From these observations, the divergence term in (1) (up to a sign) is somewhat motivated by an attempt to find a flux function independent of boosts. In fact, it follows in the case that $0<\langle\vec{H},\vec{H}\rangle=:H^2$ and $\ubar L$ is past pointing, that $\rho$ can be given in terms of the Bartnik data of $\Sigma$ as
$$\rho = \mathcal{K}-\frac14\langle\vec{H},\vec{H}\rangle+\s\nabla\cdot\alpha_H-\s\Delta\log H$$
(we refer the reader to Section 2 for definitions and proof). Moreover, in our two simplest cases, namely spherical cross-sections of the null cone of a point in a space form or the standard null cone of Schwarzchild spacetime, the last two terms cancel identically. Interestingly, work of Wang, Wang and Zhang \cite{CWZ} show deep connection between the 1-form $\alpha_H-\s{d}\log H$ and the underlying null geometry of a closed, co-dimension 2 surface $\Sigma$. For $\Sigma$ satisfying $\alpha_H = \s{d}\log H$ they show for various ambient structures that $\Sigma$ must be constrained to a shear-free ($\hat{\ubar\chi}=0$) null-hypersurface of spherical symmetry. In Section 2 (Proposition \ref{p2}) we show for a connected $\Sigma$ of arbitrary co-dimension inside a space form, if $\Sigma$ is expanding along some null section $\ubar L\in \Gamma(T^{\perp}\Sigma)$ such that $D^{\perp}\ubar L\propto \ubar L$, then it must be constrained to the null cone of point whenever $\hat{\ubar\chi}=0$. Leaning on work by Bray, Jauregui and Mars \cite{BJM} we also find direct motivation for (1) showing that $\rho$ arises naturally from variation of $E_H$ along null flows.\\
\indent In Section 2.2 we motivate our quasi local mass (2) by studying the particularly simple expression that results for $\rho$ on arbitrary cross sections of the standard null cone of Schwarzchild spacetime. In fact, for a standard null cone $\Omega$ of mass $M$ we show for any spherical cross section $\Sigma\subset\Omega$ that $m(\Sigma) = M$.\\
\indent One important application for a quasi local mass would is to use it to prove the Penrose conjecture \cite{P1,P2}:
$$\sqrt{\frac{|\Sigma_{\tiny{BH}}|}{16\pi}}\leq M$$
where $|\Sigma_{BH}|$ is the area of an isolated black-hole and $M$ is the mass of the system. In the appropriate setting this provides not only a strengthened version of the Positive Mass Theorem, but also insight regarding the mathematical validity of the \textit{weak cosmic censorship hypothesis} that Penrose employed in the formulation of his conjecture. As such, it is of great interest to both physicists and mathematicians alike. 
\newpage
\subsection{Outline}
This paper is organized as follows:
\begin{enumerate}
\item Section 1: Introduction
\item Section 2: Motivation\\
We present motivation for $\rho$ and $m$ with an analysis of null cones in space forms as well as the standard null cones of Schwarzschild geometry. We also show how (1) naturally arises when studying the variation of $E_H$ along arbitrary normal null flows.
\item Section 3: Propogation of $\rho$\\
We prove Theorem \ref{t1} by calculating the propagation of $\rho$ along arbitrary null flows in a null hypersurface $\Omega$. We also study the restrictions placed on $\Omega$ in the case that a flow satisfies $\frac{dm}{ds} = 0$, the case of equality for Theorem \ref{t2}.
\item Section 4: Foliation Comparison\\
Given an arbitrary cross section $\Sigma$ within a null hypersurface $\Omega$, we find its flux $\rho$ in terms of the data for a given background foliation. This allows us to prove Theorem \ref{t2} under the necessary decay assumptions.
\item Section 5: Spherical Symmetry\\
For a class of perturbations of the black hole exterior in a spherically symmetric spacetime, we show the existence of asymptotically flat null cones of strong flux decay that allow an (SP)-foliation. As a result, the existence of such perturbations satisfying the null energy condition give rise, via Theorem \ref{t2}, to the null Penrose conjecture.
\end{enumerate}

\newpage
\section{Motivation}
In this section we further develop our motivation for (1) based on an analysis of null cones within space forms. We then provide analysis of standard null cones in Kruskal spacetime to motivate (2) and for comparison with (SP)-foliations satisfying $\frac{dm}{ds} = 0$. We also show how an arbitrary variation of $E_H$ building on work of Bray, Jauregui and Mars \cite{BJM}, points toward $\rho$ being the optimal choice of flux function in the case of null flows. In this paper we will be using the following convention to construct the Riemann curvature tensor:
$$R_{XY}Z:= D_{[X,Y]}Z-[D_X,D_Y]Z.$$
From this will have need of the following versions of the Gauss and Codazzi equations:
\begin{proposition}\label{p1}
Suppose $\Sigma$ is a co-dimension 2 semi-Riemannian submanifold of $M^{n+1}$ that locally admits a normal null basis $\{\ubar L,L\}$ such that $\langle \ubar L,L\rangle = 2$. Then,
\begin{align}
&(n-1)\mathcal{K}-\frac{n-2}{n-1}\langle\vec{H},\vec{H}\rangle + \hat{\ubar\chi}\cdot\hat{\chi} = -R-2G(\ubar L,L)-\frac12\langle R_{\ubar L L}\ubar L,L\rangle\\
&\s\nabla\cdot\hat{\ubar\chi}(V) - \hat{\ubar\chi}(V,\vec\zeta) + \frac{n-2}{n-1}\tr\ubar\chi\zeta(V)-\frac{n-2}{n-1}Vtr\ubar\chi = G(V,\ubar L) - \frac12\langle R_{\ubar L V}L,\ubar L\rangle
\end{align}
for $V\in\Gamma(T\Sigma)$ and $(n-1)\mathcal{K}$ the scalar curvature of $\Sigma$.
\end{proposition}
\begin{proof}
From the Gauss equation (cite) we have,
$$\langle \s{R}_{V,W}U,S\rangle = \langle R_{V,W}U,S\rangle + \langle\vec{\II}(V,U),\vec{\II}(W,S)\rangle - \langle\vec{\II}(V,S),\vec{\II}(W,U)\rangle$$
for $\s{R},\,R$ the Riemann tensors of $\Sigma, M$ respectively and $V,W,U,S\in\Gamma(T\Sigma)$. Restricted to $\Sigma$ the ambient metric has inverse
$$g^{-1}|_{\Sigma}=\frac12\ubar L\otimes L+\frac12 L\otimes \ubar L+\gamma^{-1}$$
so taking a trace over $V,U$ then $W,S$ in $\Sigma$ we have:
\begin{align*}
\langle \s{R}_{VW}U,S\rangle &\xrightarrow{(V,U)} \s{Ric}(W,S)\xrightarrow{(W,S)} (n-1)\mathcal{K}\\
\langle R_{VW}U,S\rangle &\xrightarrow{(V,U)} Ric(W,S) - \frac12(\langle R_{\ubar L W}L,S\rangle+\langle R_{L W}\ubar L,S\rangle)\xrightarrow{(W,S)} R - 2Ric(\ubar L, L) - \frac12\langle R_{\ubar L L}\ubar L,L\rangle.
\end{align*}
Since $\vec{\II} = \hat{\II}+\frac{1}{n-1}\vec{H}\gamma = -\frac12\hat{\ubar\chi}L-\frac12\hat{\chi}\ubar L + \frac{1}{n-1}\vec{H}\gamma$ we have
\begin{align*}
\langle\vec{\II},\vec{\II}\rangle = \frac12(\hat{\ubar\chi}\otimes\hat{\chi}+\hat{\chi}\otimes\hat{\ubar\chi})-\frac{\langle\vec{H},L\rangle}{2(n-1)}(\hat{\ubar\chi}\otimes\gamma+\gamma\otimes\hat{\ubar\chi})-\frac{\langle\vec{H},\ubar L\rangle}{2(n-1)}(\hat{\chi}\otimes\gamma+\gamma\otimes\hat{\chi})+\Big(\frac{1}{n-1}\Big)^2\langle\vec{H},\vec{H}\rangle\gamma\otimes\gamma
\end{align*}
so returning to our trace 
\begin{align*}
&\langle\vec{\II}(V,U),\vec{\II}(W,S)\rangle\xrightarrow{(V,U),(W,S)}\langle\vec{H},\vec{H}\rangle\\
&\langle\vec{\II}(V,S),\vec{\II}(W,U)\rangle\xrightarrow{(V,U),(W,S)}\hat{\ubar\chi}\cdot\hat{\chi}+\frac{1}{n-1}\langle\vec{H},\vec{H}\rangle.
\end{align*}
Equating terms according to the Gauss equation we have
\begin{align*}
(n-1)\mathcal{K} &= R-2Ric(\ubar L,L)-\frac12\langle R_{\ubar L L}\ubar L,L\rangle-\hat{\ubar\chi}\cdot\hat{\chi}+(1-\frac{1}{n-1})\langle\vec{H},\vec{H}\rangle\\
&=-R-2G(\ubar L,L)-\frac12\langle R_{\ubar L L}\ubar L,L\rangle-\hat{\ubar\chi}\cdot\hat{\chi} +\frac{n-2}{n-1}\langle\vec{H},\vec{H}\rangle
\end{align*}
having used $G(\cdot,\cdot)=Ric(\cdot,\cdot)-\frac12 R\langle\cdot,\cdot\rangle$, (3) follows.\\
\indent From the Codazzi equation \cite{O} (pg 115), for any $V,W,U\in \Gamma(T\Sigma)$,
$$R^{\perp}_{VW}U = -(\s\nabla_V\text{II})(W,U) + (\s\nabla_W\text{II})(V,U)$$
where
$$(\s\nabla_V\text{II})(W,U):= D_V^{\perp}(\text{II}(W,U)) - \text{II}(\s\nabla_VW,U) - \text{II}(W,\s\nabla_VU).$$
So given our choice of null normal $\ubar L$ we see that
\begin{align*}
\langle D_V^{\perp}(\text{II}(W,U)),\ubar L\rangle &= -V(\ubar\chi(W,U)) - \langle \text{II}(W,U),D_V\ubar L\rangle\\
&=-V(\ubar\chi(W,U)) - \frac12\langle \text{II}(W,U),\ubar L\rangle\langle L, D_V\ubar L\rangle\\
&= -V(\ubar\chi(W,U)) + \ubar\chi(W,U)\zeta(V)\\
\langle (\s\nabla_V\text{II})(W,U),\ubar L\rangle &=  -V(\ubar\chi(W,U))+\ubar\chi(W,U)\zeta(V) + \ubar\chi(\s\nabla_VW,U)+\ubar\chi(W,\s\nabla_VU)\\
&=\zeta(V)\ubar\chi(W,U)-(\s\nabla_V\ubar\chi)(W,U).
\end{align*}
Therefore,
$$\langle R_{VW}U,\ubar L\rangle = (\s\nabla_V\ubar\chi)(W,U)-(\s\nabla_W\ubar\chi)(V,U) - \zeta(V)\ubar\chi(W,U) + \zeta(W)\ubar\chi(V,U).$$
Taking a trace over $V,U$ we conclude,
\begin{align*}
Ric(W,\ubar L) - \frac12\langle R_{\ubar L,W}L,\ubar L\rangle &= \s\nabla\cdot\ubar\chi(W) - Wtr\ubar\chi - \ubar\chi(W,\vec{\zeta}) + tr\ubar\chi\zeta(W)\\
&=\s\nabla\cdot\hat{\ubar\chi}(W) - \frac{n-2}{n-1}Wtr\ubar\chi - \hat{\ubar\chi}(W,\vec{\zeta}) + \frac{n-2}{n-1}tr\ubar\chi\zeta(W)
\end{align*}
and notice that $G(W,\ubar L) = Ric(W,\ubar L)$ since $\langle \ubar L,W\rangle = 0$.
\end{proof}
\subsection{Null Cone of a point in a Space Form}
In this section we spend some time studying (1) and (2) on cross-sections of our first example of a null hypersurface, the null cone of a point in a space form. We adopt the notation as in \cite{O} where $\mathbb{R}^n_\nu$ corresponds to the manifold $\mathbb{R}^n$ endowed with the standard inner product of index $\nu$. 
\begin{lemma}\label{l1}
Suppose $\Sigma^k\hookrightarrow\mathbb{R}^n_\nu$ ($k\geq 2$) is a connected semi-Riemannian submanifold admitting a non-trivial section $\vec{n}\in\Gamma(T^\perp\Sigma)$ such that $D^\perp\vec{n}=\eta\vec{n}$ for some 1-form $\eta$. Then the following are equivalent
\begin{enumerate}
\item $p\mapsto\exp(-\vec{n}|_p)$ is constant
\item $\eta=0$ and $\langle \II,-\vec{n}\rangle=\gamma$
\item $\langle\II,-\vec{n}\rangle=\gamma$
\end{enumerate}
where $\gamma=\langle\cdot,\cdot\rangle|_\Sigma$ and $\exp:T\mathbb{R}^n_\nu\to\mathbb{R}^n_\nu$ is the exponential map.
\end{lemma}
\begin{proof}
Choosing an origin $\vec{o}$ for $\mathbb{R}^n_\nu$ with associated position vector field $P = x^i\partial_i\in \Gamma(T\mathbb{R}^n_\nu$) it follows that 
$$\exp(-\vec{n}|_{\vec{p}}) = (P-\vec{n})|_{\vec{p}}$$
where, by an abuse of notation, we have omitted the composition of canonical isometries $T_{\vec{p}}\mathbb{R}^n_\nu\to T_{\vec{o}}\mathbb{R}^n_\nu\to \mathbb{R}^n_\nu$ identifying $\vec{p}$ with $P|_{\vec{p}}$. As a result, for any $V\in\Gamma(T\Sigma)$:
\begin{align*}
d(\exp(-\vec{n}))(V) &= D_V(P-\vec{n})\\
&=V-D_V\vec{n}\\
&=(V-D^{\|}_V\vec{n})-\eta(V)\vec{n}
\end{align*}
and we conclude that $\exp\circ(-\vec{n})$ is locally constant (or constant when $\Sigma$ is connected) if and only if both $V = D^{\|}_V\vec{n}$ for any $V\in\Gamma(T\Sigma)$ and $\eta=0$. Since $D^{\|}_V\vec{n} = V$ for any $V\in\Gamma(T\Sigma)$ is equivalent to $-\langle\II(V,W),\vec{n}\rangle (= \langle W,D_V\vec{n}\rangle)=\langle V,W\rangle$ for any $V,W\in\Gamma(T\Sigma)$ we have that 1.$\iff$2.\\
2.$\implies$3. is trivial. To show 3.$\implies$2. we start by taking any $U,V,W\in\Gamma(T\Sigma)$ so that the Codazzi equation gives
$$\langle(\s\nabla_V\II)(W,U),\vec{n}\rangle = \langle(\s\nabla_W\II)(V,U),\vec{n}\rangle$$
where
\begin{align*}
\langle(\s\nabla_V\II)(W,U),\vec{n}\rangle&:=\langle D^\perp_V(\II(W,U))-\II(\s\nabla_VW,U)-\II(W,\s\nabla_VU),\vec{n}\rangle\\
&=-(\s\nabla_V\gamma)(W,U)-\langle\II(W,U),D_V\vec{n}\rangle\\
&=\eta(V)\langle W,U\rangle
\end{align*}
and therefore $\eta(V)\langle W,U\rangle = \eta(W)\langle V,U\rangle$. Taking a trace over $V,U$ we conclude that $\eta(W) = k\eta(W)$ so that $k\geq2$ forces $\eta = 0$ as desired.
\end{proof}
The Hyperquadrics of $\mathbb{R}^n_\nu$ correspond to the complete, totally umbilic hypersurfaces $H_C$ of constant curvature $C$ (provided $C\neq 0$) given by
$$H_C:=\{\vec{v}\in\mathbb{R}^n_\nu|\langle \vec{v},\vec{v}\rangle = C\}$$
where $C$ runs through all values in $\mathbb{R}$. When $C=0$, $\Omega = H_0$ is the collection of all null geodesics emanating from the origin called the \textit{null-cone} centered at the origin. Consequently $\Omega+\vec{p}$ corresponds to the null cone at the point $\vec{p}$. Similarly for a space form $M$ we will define the null cone of a point $p\in M$ as the collection of all null geodesics emanating from $p$.
\begin{proposition}\label{p2}
Suppose $\Sigma^k\hookrightarrow M^{n-1}$ ($k\geq 2$) is a connected semi-Riemannian submanifold of a space form $M$. Suppose that $\Sigma$ is expanding along some null section $\ubar L$ satisfying $D^{\perp}\ubar L=\zeta\ubar L$ for some 1-form (or torsion) $\zeta$. Then the following are equivalent
\begin{enumerate}
\item $p\mapsto\exp(-\frac{k\ubar L}{\tr\ubar\chi}|_p)$ is constant
\item $\tau:=\zeta-\s{d}\log\tr\ubar\chi = 0$ and $\hat{\ubar\chi}=0$
\item $\hat{\ubar\chi}=0$
\end{enumerate}
where $\ubar\chi:=-\langle \ubar L,\II\rangle$.
\end{proposition}
\begin{proof}
If $M$ has constant curvature $C\neq 0$ we find a Hyperquadric $H_C$ of $\mathbb{R}^n_\nu$ (for some $\nu$) of the same dimension and index as M. It's a well known fact that $M$ and $H_C$ have isometric semi-Riemannian coverings (\cite{O}, (pg224) Theorem 17) which we identify and denote by $\mathcal{O}$. As a result, for any $q\in \Sigma\subset M$ we find a $\vec{q}\in H_C$ with isometric neighborhoods. Moreover, we find an open set $U_q\subset\Sigma$ of $q$ which is isometric onto some $V_{\vec{q}}\subset H_C$. Without loss of generality we will also identify $T^\perp(U_q)$ and $T^\perp(V_{\vec{q}})$. Denoting the ambient connection on $\mathbb{R}^n_\nu$ by $\bar{D}$ and the unit normal of $H^{n-1}_C\subset\mathbb{R}_\nu^n$ by $\vec{N}$ we conclude, for the null section $L^- = \frac{\ubar L}{\tr\ubar\chi}$, that $\bar{D}_V^{\perp}L^- = \tau(V)L^-+\langle\vec{N},\vec{N}\rangle\langle \bar{D}_V L^-,\vec{N}\rangle\vec{N}$. So given that all Hyperquadrics are totally umbilic it follows that $\langle \bar{D}_VL^-,\vec{N}\rangle\propto\langle V,L^-\rangle = 0$ and therefore $\bar{D}^\perp_V L^-=\tau(V) L^-$.\\
\indent We wish to show 3.$\implies$1. From the hypothesis we have that $\ubar\chi = \frac{1}{k}\tr\ubar\chi\gamma$ so it follows that $k\ubar\chi^-=-\langle kL^-,\II\rangle = \gamma$ and Lemma \ref{l2} applies for $V_{\vec{q}}\subset\mathbb{R}^n_\nu$. We conclude that $V_{\vec{q}}$ is contained inside the null-cone of a point $\vec{o}\in\mathbb{R}_\nu^n$, where $\vec{o} = \exp_{n,\nu}(-kL^-)|_{V_{\vec{q}}}$, and every $\vec{p}\in V_{\vec{q}}$ is connected to $\vec{o}$ by a null geodesic in $\mathbb{R}_\nu^n$ along $kL^-|_{\vec{p}}=\frac{k\ubar L}{\tr\ubar\chi}\Big|_{\vec{p}}\in T^\perp_{\vec{p}}V_{\vec{q}}\subset T_{\vec{p}}H_C$. Since $H_C$ is complete and totally umbilic these null geodesics must remain within $H_C$. 
Up to a possible shrinkage of $V_{\vec{q}}$ we may lift a neighborhood of the geodesic $\vec{q}\to\vec{o}$ to a neighborhood of some null geodesic $\tilde q\to\tilde o$ in $\mathcal{O}$ concluding that the isometric image $V_{\tilde{q}}$ of $V_{\vec{q}}$ contracts to $\tilde{o}$ along null geodesics. Since $M$ is complete the null geodesic $\tilde{q}\to\tilde{o}$ in turn gives rise to a null geodesic $q\to o$ in $M$ and up to an additional shrinkage we conclude that $U_q$ contracts along null geodesics onto $o$:

\begin{figure}[h]
\centering
\def\svgwidth{400pt} 
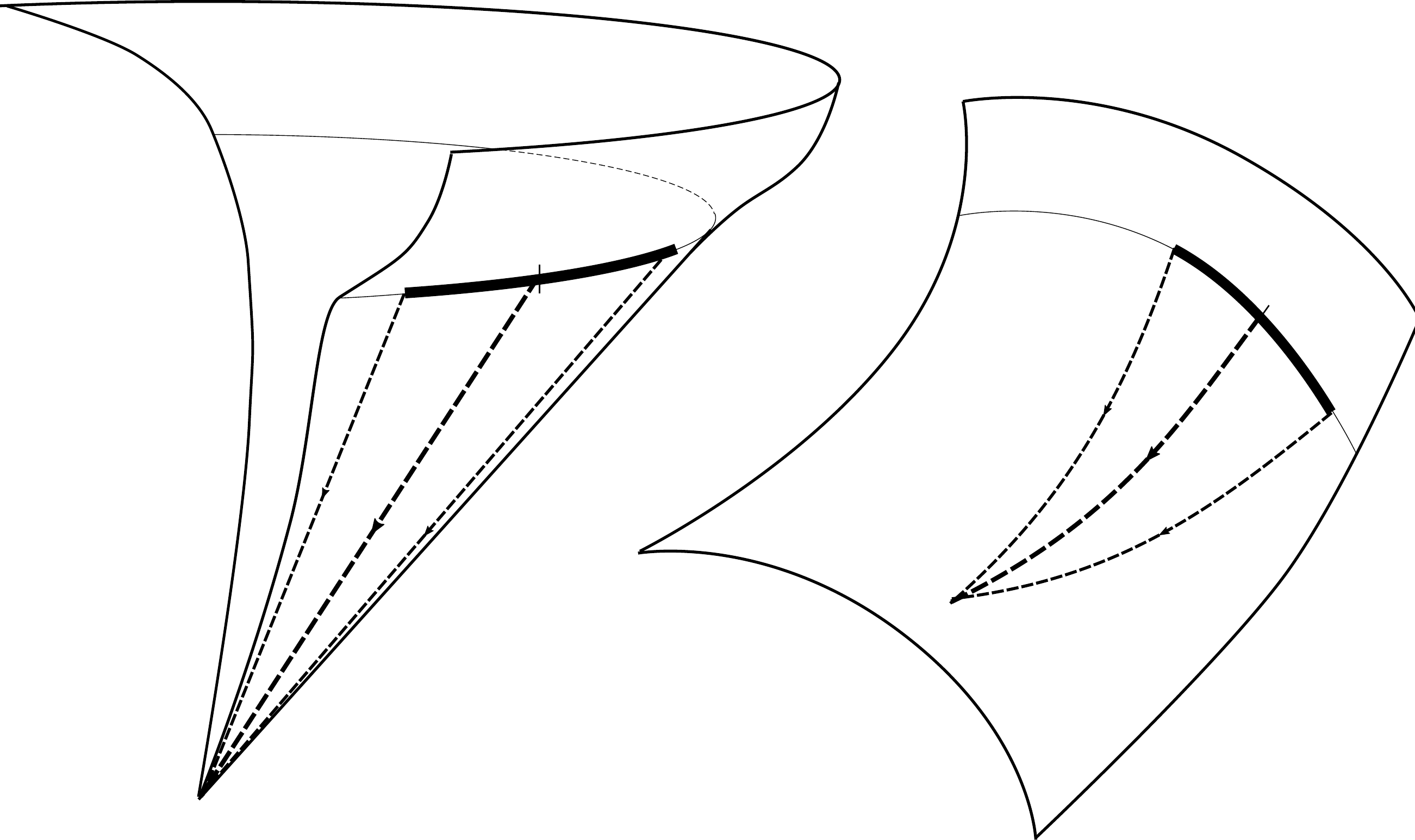
\end{figure}

\noindent In fact our argument shows that the union of all points in $\Sigma$ that get transported to $o$ must form an open subset of $\Sigma$. Conversely, if any point in $\Sigma$ gets transported to a point other than $o$ the same follows for a neighborhood around that point in $\Sigma$. By connectedness, all of $\Sigma$ must be transported to $o$ along null geodesics as desired.\\
\indent For 1.$\implies$3. we take a null geodesic from $q\in\Sigma$ along $kL^- = \frac{k\ubar L}{\tr\ubar\chi}$ to the focal point, at say, $o\in M$. Similarly as before this gives rise to a tubular neighborhood around some null geodesic $\vec{q}\to\vec{o}$ in $H_C$ within which $V_{\vec{q}}$ is contracted along null geodesics onto $\vec{o}$. Since $H_C$ is totally umbilic $V_{\vec{q}}$ is transported to $\vec{o}$ along null geodesics in $\mathbb{R}_\nu^n$ forming part of the null cone at $\vec{o} = \exp_{n,\nu}(-kL^-)|_{V_{\vec{q}}}$. Lemma \ref{l2} applies once again and we conclude that $k\chi^- =\gamma \implies\hat{\ubar\chi}=0$ on $V_{\vec{q}}$ hence on $U_q$ (since they have isometric neighborhoods). Since $q$ was arbitrary chosen the result follows.\\
Once again 2.$\implies$3. is trivial. To show 3.$\implies$2. we have similarly as in Lemma \ref{l2} from the Codazzi equation for $\Sigma\hookrightarrow M$ and $M$ of constant curvature that:
$$\tau(V)\langle W,U\rangle=\tau(W)\langle V,U\rangle$$
so that a trace over $V,U$ yields again $\tau(W) = k\tau(W)$ and therefore $\tau = 0$.
\end{proof}
For any connected, co-dimension 2 surface, from the fact that $\langle D_V\ubar L,\ubar L\rangle = \frac12V\langle\ubar L,\ubar L\rangle = 0$, it necessarily follows that $D^\perp \ubar L=\zeta\ubar L$ for any null section $\ubar L\in \Gamma(T^\perp\Sigma)$ and some associated 1-form $\zeta$. In particular, $\Sigma$ will be contained inside the null cone of a point inside a space form $M$ if we're able to find a null section $\ubar L$ along which $\Sigma$ is expanding and shear-free. Along such $\ubar L$ it follows from Lemma \ref{l1} for $C=0$ and Proposition \ref{p2} for $C\neq 0$ that $\tau=0$. So for Lorentzian space forms of dimension-4 (i.e. `Minkowski spacetime' for $C=0$, `de Sitter spacetime' for $C>0$ and `anti-de Sitter spacetime' for $C<0$) that (3) implies $\Sigma$ has flux $\rho = K - \frac14\langle\vec{H},\vec{H}\rangle=C$. When $\Sigma$ is a 2-sphere, by the Gauss-Bonnet Theorem, we conclude that $$m(\Sigma)=|E_H(\Sigma)|=\frac{|C|}{2}\Big(\frac{|\Sigma|}{4\pi}\Big)^{\frac32}.$$
The reader may be wondering why the need for the divergence term in (1) when it vanishes altogether. We take as our first hint the fact that vanishing $\tau=\zeta-\s d\log\tr\ubar\chi$ is characteristic of spherical cross sections of $\Omega$ which subsequently may obscure it's contribution. In the paper by Wang, Wang and Zhang (\cite{CWZ} Theorem 3.13, Theorem 5.2) the authors prove $\tau=0$ to be sufficient in spacetimes of constant curvature to constrain a closed, co-dimension 2 surface $\Sigma$ to a shear-free null hypersurface of spherical symmetry. Proof follows from the following Lemma and Proposition \ref{p2} when $\Sigma$ is a 2-sphere:
\begin{lemma}\label{l2}
Suppose $\Sigma$ is a spacelike 2-sphere expanding along some $\ubar L$ inside a space form $M$. Suppose also $D^{\perp}\ubar L=\zeta\ubar L$ for some 1-form $\zeta$ then
$$\tau:=\zeta-\s{d}\log \tr\ubar{\chi}=0\implies \hat{\ubar\chi}=0.$$
\end{lemma}
\begin{proof}
As used in Proposition \ref{p2} to prove the implication in the opposite direction, we start with the Codazzi equation. For $L^- = \frac{\ubar L}{\tr\ubar\chi}$ we recall that $D^\perp L^- = \tau L^-=0$ and $\tr\chi^-=1$ so we have:
\begin{align*}
\langle\s\nabla_V\II(W,U),L^-\rangle&=\langle D^\perp_V(\II(W,U))-\II(\s\nabla_VW,U)-\II(W,\s\nabla_VU),L^-\rangle\\
&=-(\s\nabla_V\chi^-)(W,U)-\langle\II(W,U),D^\perp_VL^-\rangle\\
&=-(\s\nabla_V\chi^-)(W,U)\\
&=-(\s\nabla_V\hat{\chi}^-)(W,U)\\
0=\langle R^\perp_{VW}U,L^-\rangle&=\langle-\s\nabla_V\II(W,U)+\s\nabla_W\II(V,U),L^-\rangle\\
&=(\s\nabla_V\hat{\chi}^-)(W,U)-(\s\nabla_W\hat{\chi}^-)(V,U).
\end{align*}
Taking a trace over $V,U$ this implies $\s\nabla\cdot\hat{\chi}^-=0$. Since $\Sigma$ is a topological 2-sphere it's a well known consequence of the Uniformization Theorem (see for example \cite{C}) that the divergence operator on symmetric trace-free 2-tensors is injective so that $\hat{\ubar\chi}=\tr\ubar\chi\hat{\chi}^-=0$.
\end{proof} 
\begin{definition}\label{d6}
We say a 2-sphere $\Sigma$ is \underline{admissible} if 
$$\langle\vec{H},\vec{H}\rangle=H^2>0.$$
\end{definition}
In the case that $\Sigma$ is admissible we're able to construct the orthonormal frame field $$\{e_r=-\frac{\vec{H}}{H},e_t\}$$
for $e_t$ future pointing. The associated connection 1-form is given by
$$\alpha_H(V):=\langle D_Ve_r,e_t\rangle.$$
From the following known Lemma (\cite{CWZ}), Proposition \ref{p2} and Lemma \ref{l2}, a necessary and sufficient condition for an admissible sphere $\Sigma$ to be constrained to the past(future) light-cone of a point in a space form is given by $\alpha_H = \pm\s d\log H$:
\begin{lemma}\label{l3}
For $\Sigma$ admissible
$$\tau = \pm\alpha_H-\s d\log H$$
from which we conclude that
$$\rho_{\mp} = K-\frac14\langle\vec{H},\vec{H}\rangle\pm\s\nabla\cdot\alpha_H-\s \Delta\log H$$
where $+/-$ indicates whether $L^-$ is past/future pointing.
\end{lemma}
\begin{proof}
Since $-\vec{H}=\frac12(\tr\chi \ubar L+\tr\ubar\chi L)$ we see $H^2=\tr\ubar\chi\tr\chi$ so that the inverse mean curvature vector is given by
$$\vec{I}:=-\frac{\vec{H}}{H^2}=\frac12\Big(\frac{\ubar L}{\tr\ubar\chi}+\frac{L}{\tr\chi}\Big).$$
As a result,
\begin{align*}
\alpha_H(V) &=\langle D_Ve_r,e_t\rangle\\
&=\langle D_V \frac{e_r}{H},He_t\rangle\\
&=\langle D_V\frac12\Big(\frac{\ubar L}{\tr\ubar\chi}+\frac{L}{\tr\chi}\Big),\mp\frac{1}{2}(\tr\chi\ubar L-\tr\ubar\chi L)\rangle\\
&=\pm\frac{1}{4}\Big(\langle D_V\frac{\ubar L}{\tr\ubar\chi},\tr\ubar\chi L\rangle - \langle D_V\frac{\tr\ubar\chi L}{H^2},H^2\frac{\ubar L}{\tr\ubar\chi}\rangle\Big)\\
&=\pm\frac14\Big(\langle D_V\frac{\ubar L}{\tr\ubar\chi},\tr\ubar\chi L\rangle +\langle D_V (H^2\frac{\ubar L}{\tr\ubar\chi}),\frac{\tr\ubar\chi L}{H^2}\rangle\Big)\\
&=\pm\frac14\Big(2\langle D_V\frac{\ubar L}{\tr\ubar\chi},\tr\ubar\chi L\rangle +2V\log H^2\Big)\\
&=\pm\Big(\zeta(V)-V\log\tr\ubar\chi+V\log H\Big)
\end{align*}
\end{proof}
Wang, Wang and Zhang (\cite{CWZ} Theorem B') also extend their result to expanding, co-dimension 2 surfaces $\Sigma$ in $n$-dimensional Schwarzschild spacetime ($n\geq 4$). Namely, that any such $\Sigma$ satisfying $\alpha_H = d\log H$ must be inside a shear-free null hypersurfaces of symmetry, or the `standard null-cone' in this geometry. So with the hopes of further illuminating modification of $E_H$ by way of the flux function $\rho$ we move on to this setting in dimension $4$. 
\subsection{Schwarzschild Geometry}
The Schwarzschild spacetime models a static black hole of mass $M$ given by the metric
$$g_S=-\mathfrak{h}dt\otimes dt+\mathfrak{h}^{-1}dr\otimes dr+r^2(d\vartheta\otimes d\vartheta+(\sin\vartheta)^2 d\varphi\otimes d\varphi)$$
where $\mathfrak{h} = 1-\frac{2M}{r}$ for $2M>r>0,\,r> 2M$.
The maximal extension of this geometry is called the Kruskal spacetime $(\mathbb{P}\times_r\mathbb{S}^2,g_K)$ which is given by the warped product of the Kruskal Plane $\mathbb{P}:=\{uv>-2Me^{-1}\}$ and the standard round $\mathbb{S}^2$ with warping function $r=g^{-1}(uv)$ for $g(r) = (r-2M)e^{\frac{r}{2M}-1},\,r>0$. The metric and its inverse is given by:
\begin{align*}
g_K &= F(r)(dv\otimes du+dv\otimes du) + r^2(d\vartheta\otimes d\vartheta+(\sin\vartheta)^2 d\varphi\otimes d\varphi)\\
g_K^{-1}&=\frac{1}{F}(\partial_v\otimes\partial_u +\partial_u\otimes\partial_v)+r^{-2}(\partial_\vartheta\otimes\partial_\vartheta+(\sin\vartheta)^{-2}\partial_\varphi\otimes\partial_\varphi)
\end{align*}
where $F(r) = \frac{8M^2}{r}e^{1-\frac{r}{2M}}$. We recover the Schwarzschild spacetime on $v>0,\,u\neq0$ with the coordinate change $t = 2M\log|\frac{v}{u}|$ (\cite{O}).\\
Each round $\mathbb{S}^2$ has area $4\pi r^2$ so we interpret $r$ as a `radius' function and $F(r)$ gives rise to unbounded curvature at $r=0$ and the `black hole' singularity. A standard past null-cone of Schwarzchild spacetime $\Omega$ is the hypersurface given by fixing the coordinate $v$, say $v=v_0$. Denoting the gradient of a function $f$ by $Df$ we recognize the null vector field $ \frac{\partial_u}{F} = Dv$ restricts to $\Omega$ as both a tangent (since $\partial_u(v) = 0$) and normal (since $Dv\perp T\Omega$) vector field. It follows that $Dv\in T^\perp\Omega\cap T\Omega$ and the induced metric on $\Omega$ degenerate, so $\Omega$ is an example of a \textit{null hypersurface}. From the identity $D_{Df}Df=\frac12D|Df|^2$ we see $\frac{\partial_u}{F}$ is geodesic and $\Omega$ is realized as the past light cone of a section of the event horizon ($r=2M$) as shown below:

\begin{figure}[h]
\centering
\def\svgwidth{500pt} 
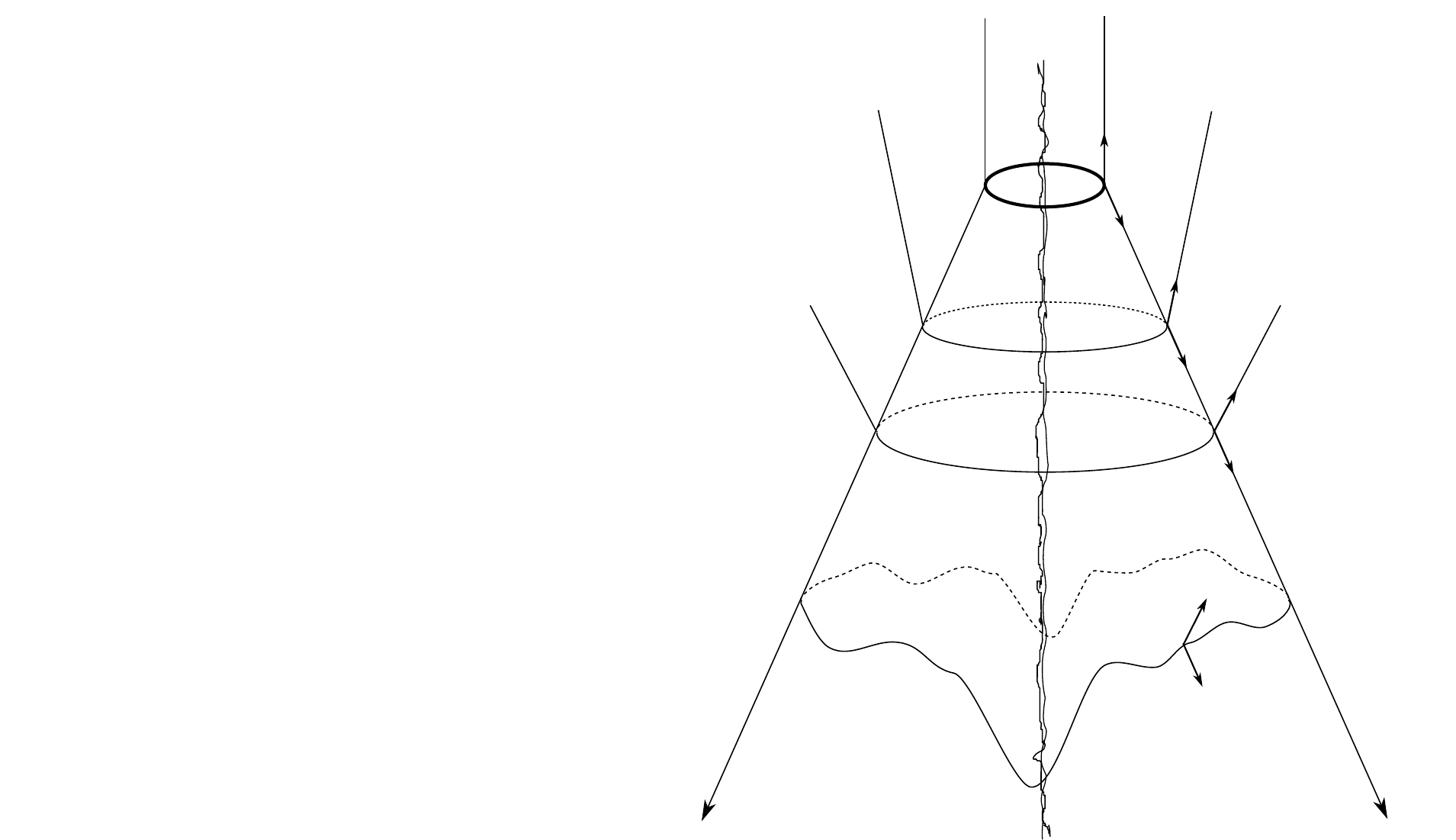
\end{figure}

Setting $\ubar L = D(4M\log v) = \frac{4M}{v}\frac{\partial_u}{F}$ we see $\ubar L(r) = \frac{4M}{v}\frac{r_u}{F} = \frac{4M}{v}\frac{v}{g'(r)F} = \frac{4M}{v}\frac{v}{4M}=1$. We conclude that $r$ restricts to an affine parameter along the geodesics generating $\Omega$ and therefore any cross section $\Sigma$ can be given as a graph over $\mathbb{S}^2$ in $\Omega$ with graph function $\omega = r|_\Sigma$. We extend $\omega$ to the rest of $\Omega$ by assigning $\ubar L(\omega) =0$ and to a neighborhood of $\Omega$ by assigning $\partial_v\omega=0$. From the canonical, homothetic embedding onto the leaves $\mathbb{S}^2\hookrightarrow \mathbb{P}\times_r\mathbb{S}^2$ we obtain the lifted vector fields $V\in \mathcal{L}(\mathbb{S}^2)\subset \Gamma(T(\mathbb{P}\times_r\mathbb{S}^2))$ such that $\langle\partial_u(\partial_v),V\rangle=[\partial_v(\partial_u),V]=0$. It follows that $\mathcal{L}(\mathbb{S}^2)|_{\Sigma_r}=\Gamma(T\Sigma_r)$ (for $\Sigma_r:=\{r=const.,v=v_0\}$) and therefore $\tilde V := V+V\omega\ubar L\in \Gamma(T\Sigma)$ since
$$\tilde V(r-\omega) = -V\omega+V\omega\ubar L(r) = -V\omega+V\omega=0.$$
Since $\Sigma=\{4M\log \frac{v}{v_0} = 0,\,r=\omega\}$ we have $\ubar L, D(r-\omega)\in\Gamma(T^\perp\Sigma)$ are linearly independent so that $L = a\ubar L+bD(r-\omega)$ and we wish to solve for $a,b$. Now $Dr = r_u\frac{\partial_v}{F}+r_v\frac{\partial_u}{F} = \frac{v}{4M}(\partial_v+r_v\ubar L)$ and $D\omega =\nabla\omega$ for $\nabla$ the induced covariant derivative on $\Sigma_r$ giving 
$$L=(a+b\frac{v_0}{4M}r_v)\ubar L+\frac{bv_0}{4M}\partial_v-b\nabla\omega.$$ 
For simplicity we set $A = a+b\frac{v_0}{4M}r_v$ and solve for $A,b$ in $L = A\ubar L+b(\frac{v_0}{4M}\partial_v)-b\nabla\omega$:
\begin{align*}
2 &= \langle L,\ubar L\rangle = b\langle\frac{v_0}{4M}\partial_v,\ubar L\rangle = b\\
0 &= \langle L,L\rangle = 2Ab+b^2|\nabla\omega|^2 = 4(A+|\s\nabla\omega|^2)
\end{align*}
having used $\nabla\omega= \s\nabla\omega - |\s\nabla\omega|^2\ubar L$ in the second equality. We conclude that
\begin{align*}
L &= \frac{v_0}{2M}\partial_v-|\s\nabla\omega|^2\ubar L - 2(\s\nabla \omega-|\s\nabla\omega|^2\ubar L)\\
&=\frac{v_0}{2M}\partial_v+|\s\nabla\omega|^2\ubar L-2\s\nabla\omega.
\end{align*}
\begin{lemma}\label{l4}
Given a cross section $\Sigma:=\{r=\omega\}$ of the standard null cone $\Omega:=\{v=v_0\}$ in Kruskal spacetime we have for the generator $\ubar L$ satisfying $\ubar L(r) = 1$ that: 
\begin{align*}
\langle \tilde V,\tilde W\rangle &= \omega^2(\tilde V,\tilde W)\\
\ubar\chi(\tilde V,\tilde W)&=\frac{1}{\omega}\langle\tilde V,\tilde W\rangle\\
\tr\ubar\chi&=\frac{2}{\omega}\\
\chi(\tilde V,\tilde W) &= \frac{1}{\omega}(1-\frac{2M}{\omega}+|\s\nabla\omega|^2)\langle \tilde V,\tilde W\rangle-2H^\omega(\tilde V,\tilde W)\\
\tr\chi&=\frac{2}{\omega}\Big(1-\frac{2M}{\omega}-\omega^2\s\Delta\log\omega\Big)\\
\zeta (\tilde V) &= -\tilde V\log\omega\\
\rho&=\frac{2M}{\omega^3}
\end{align*}
where $\tilde V,\tilde W\in \Gamma(T\Sigma)$ and $(\cdot,\cdot)$ the round metric on $\mathbb{S}^2$.
\end{lemma}
\begin{proof}
The first identity follows trivially from the metric $g_K$ upon restriction to $\Sigma$. From the Koszul formula and the fact that $\ubar L$ is geodesic it follows that $D_{\tilde V}\ubar L = \frac{\ubar{L}(r)}{r}V|_{\Sigma}=\frac{1}{\omega}V$, denoting the Hessian of $\omega$ on $\Sigma$ by $H^\omega$ we therefore have
\begin{align*}
\ubar\chi(\tilde V,\tilde W)&=\langle D_{\tilde V}\ubar L,\tilde W\rangle\\
&=\frac{1}{\omega}\langle V,W\rangle\\
\chi(\tilde V,\tilde W)&=\frac{v_0}{2M}\langle D_{\tilde V}\partial_v,\tilde W\rangle+|\s\nabla\omega|^2\langle D_V\ubar L,W\rangle-2\langle D_{\tilde V}\s\nabla\omega,\tilde W\rangle\\
&=\frac{v_0}{2M}\Big(\langle D_V\partial_v,W\rangle+V\omega\langle D_{\ubar L}\partial_v,W\rangle+W\omega\langle D_V\partial_v,\ubar L\rangle+V\omega W\omega\langle D_{\ubar L}\partial_v,\ubar L\rangle\Big)\\
&\indent+|\s\nabla\omega|^2\ubar\chi(V,W)-2H^\omega(\tilde V,\tilde W)\\
&=\frac{v_0r_v}{2M\omega}\langle V,W\rangle+\frac{1}{\omega}|\s\nabla\omega|^2\langle V,W\rangle-2H^\omega(\tilde V,\tilde W)
\end{align*}
where in the forth line we use the Koszul formula to evaluate the first term and metric compatibility to show the last three terms vanish. We have
$$\frac{vr_v}{2Mr} = \frac{v}{2Mr}\frac{u}{g'(r)} = \frac{1}{2Mr}\frac{Fg}{4M}=\frac{1}{r}(1-\frac{2M}{r})$$
so the second and forth identities follow upon restriction to $\Sigma$. The third identity is simply a trace over $\Sigma$ of the second. Similarly the fifth follows our taking a trace of the forth and employing the fact that
$$\s\Delta\omega-\frac{1}{\omega}|\s\nabla\omega|^2= \omega\s\Delta\log\omega.$$
For $\zeta$:
\begin{align*}
\zeta(\tilde V) &= \frac12\langle D_{\tilde V}\ubar L,L\rangle\\
&=\frac{1}{2\omega}\langle V,\frac{v_0}{2M}\partial_v+|\s\nabla\omega|^2\ubar L - 2\s\nabla \omega\rangle\\
&=-\frac{1}{\omega}\langle V,\s\nabla\omega\rangle\\
&=-\frac{1}{\omega}\langle \tilde V,\s\nabla\omega\rangle\\
&=-\tilde V\log\omega.
\end{align*}
From the first identity we conclude that $\Sigma$ has Gaussian curvature $\mathcal{K}=\frac{1}{\omega^2}-\s\Delta\log\omega$ and therefore
$$\langle\vec{H},\vec{H}\rangle=\tr\chi\tr\ubar\chi= 4(\mathcal{K}-\frac{2M}{\omega^3}).$$
Since $\zeta - \s{d}\log\tr\ubar\chi = -\s{d}\log\omega -(-\s{d}\log\omega) = 0$ on $\Sigma$ we have 
$$\rho = \mathcal{K}-\frac14\langle\vec{H},\vec{H}\rangle = \frac{2M}{\omega^3}.$$
\end{proof}
It follows, in Schwarzschild spacetime, that all foliations to the past of a section of the event horizon ($r=2M$) inside the standard null cone ($v=v_0$) are (SP)-foliations since $\Sigma=\{r=\omega\}$ satisfies
$$\frac14\langle\vec{H},\vec{H}\rangle-\frac13\s\Delta\log\rho = \frac{1}{\omega^2}(1-\frac{2M}{\omega})>0\iff \omega>2M.$$
Moreover, equality is reached only at the horizon itself indicating physical significance to our property (P). One of the motivating factors for our choice of mass functional (2) comes from our ability, in this special case, to extract the exact mass content within any $\Sigma\subset\Omega$:
$$m(\Sigma) = \frac12\Big(\frac{1}{4\pi}\int_\Sigma(\frac{2M}{\omega^3})^\frac23dA\Big)^\frac32=\frac12\Big(\frac{1}{4\pi}\int_\Sigma\frac{(2M)^\frac23}{\omega^2}\omega^2d\mathbb{S}^2\Big)^\frac32 = M.$$
\begin{lemma}\label{l5}
Suppose $\Sigma$ is a compact Riemannian manifold, then for any $f\in\mathcal{F}(\Sigma)$
$$\Big(\int f^{\frac23}dA\Big)^{\frac32} = \inf_{\psi>0}\Big(\sqrt{\int\psi^2dA}\int\frac{|f|}{\psi}dA\Big)$$
\end{lemma}
\begin{proof}
by choosing $\psi^3_\epsilon = |f|+\epsilon$ for some $\epsilon>0$ it's a simple verification that
$$\Big(\int(|f|+\epsilon)^{\frac23}dA\Big)^{\frac32}\geq \sqrt{\int\psi^2_\epsilon dA}\int\frac{|f|}{\psi_\epsilon}dA\geq\inf_{\psi>0}\sqrt{\int\psi^2dA}\int\frac{|f|}{\psi}dA$$
so by the Dominated Convergence Theorem
$$\Big(\int f^{\frac23}dA\Big)^{\frac32}=\lim_{\epsilon\to0}\Big(\int(|f|+\epsilon)^{\frac23}dA\Big)^{\frac32}\geq\inf_{\psi>0}\sqrt{\int\psi^2dA}\int\frac{|f|}{\psi}dA.$$
We show the inequality holds in the opposite direction from H\"{o}lder's inequality
$$\int f^{\frac23}dA = \int(\frac{f}{\psi})^{\frac23}\psi^{\frac23}dA\leq\Big(\sqrt{\int\psi^2dA}\Big)^{\frac23}\Big(\int\frac{|f|}{\psi}dA\Big)^{\frac23}$$
where the result follows from raising both sides to the $\frac32$ power and taking an infimum over all $\psi>0$.
\end{proof}
\noindent So given any 2-sphere with non-negative flux $\rho\geq0$ in an arbitrary spacetime, defining $E_H^\psi(\Sigma):=\frac{1}{8\pi}\sqrt{\frac{\int\psi^2dA}{4\pi}}\int\frac{\rho}{\psi}dA$, we conclude that
$$m(\Sigma)=\inf_{\psi>0}E_H^\psi(\Sigma)\leq E_H(\Sigma)$$
as desired. Recalling our use of H\"{o}lder's inequality in the proof of Lemma \ref{l5}, we see that $m(\Sigma) = E_H(\Sigma)$ if and only if $\rho$ is constant on $\Sigma$. So for $\Sigma:=\{r = \omega\}\subset\Omega$, where $\Omega$ is the standard null cone in Schwarzschild spacetime, we see that $m(\Sigma)$ underestimates the Hawking energy $E_H(\Sigma)$ with equality only if $\rho$ hence $\omega$ is constant. Namely the round spheres within `time-symmetric' slices given by $t=const>0\iff\frac{v}{u}=const>0$ (so that $v=v_0\implies r=const$) as expected from Sauters work (\cite{S}, Lemma 4.4). Foliating $\Omega$ with time-symmetric spheres is known to correspond asymptotically with round spheres in the `rest-space' of the black hole and should therefore give rise to an energy measurement that matches the total mass $M$. We will show this by taking any foliation $\{\Sigma_s\}$ of $\Omega$ approaching round spheres (asymptotically corresponding to coordinate spheres in a ``boosted" frame) measuring the total energy as $\lim_{s\to\infty}E_H(\Sigma_s)$ (see, for example, \cite{MS1}). Setting $\omega_s:=r|_{\Sigma_s}$ and $\bar{\omega}_s = \frac{1}{4\pi}\int \omega_s d\mathbb{S}^2$ we have (\cite{MS1}, Corollary 3) that:
$$\frac{\bar{\omega}_s}{\omega_s}\xrightarrow{s\to\infty} \sqrt{1+|\vec{a}|^2}+\sum_{i=1}^3a_iY^i$$
for some $\vec{a}\in \mathbb{R}^3$ and $\{Y_i\}$ the $l=1$ spherical harmonics of round $\mathbb{S}^2$.
It follows that
$$E_H(\Sigma_s)=\frac{1}{8\pi}\sqrt{\frac{|\Sigma_s|}{4\pi}}\int_{\Sigma_s}\rho dA =\frac{M}{4\pi}\sqrt{\frac{1}{4\pi}\int\Big(\frac{\omega_s}{\bar{\omega}_s}\Big)^2d\mathbb{S}^2}\int_{\Sigma_s}\frac{\bar{\omega}_s}{\omega_s}d\mathbb{S}^2 \xrightarrow{s\to\infty} M\sqrt{1+|\vec{a}|^2}$$
so that the energy approaches the mass $M$ only if $\vec{a}=0$ i.e. $\frac{\bar{\omega}_s}{\omega}\to1$. Clearly this corresponds asymptotically to the $r=const$ foliation inside $\Omega$ i.e. the time symmetric spheres. Herein it seems the difficulty lies in finding foliations such that $E_H(\Sigma_s)$ \textit{increases} to the Bondi mass. Even in Schwarzschild spacetime, if insistent upon the use of $E_H$, our only choice of foliation increasing to the mass $M$ is to foliate with time symmetric spheres. Not only is this flow highly specialized it dictates strong restrictions on our initial choice of $\Sigma$. This is to be expected of a quasi-local energy as mentioned in Section 1.1 due to its inherent sensitivity to boosts in our abstract reference frame along the flow. We hope that our quasi-local mass functional $m(\Sigma)$ requires less rigidity in our choice of foliation in appealing instead to the intrinsic mass content within $\Sigma$ rather than energy. We direct the reader therefore to Section 3 for an immediate analysis of the propagation of $m(\Sigma)$ along any given foliation of a null-cone in arbitrary spacetimes.

\subsection{Variation of $E_H$}

In this section we spend some time studying arbitrary normal variations of $E_H$ on admissible spheres following work of Bray, Jauregui and Mars (\cite{BJM}). The authors of \cite{BJM} consider `uniform area expanding flows' according to the flow vector $\partial_s = \vec{I}+\beta\vec{I}^\perp$ so we first spend some time extending their Plane Theorem to incorporate arbitrary normal flows $\partial_s = \alpha\vec{I}+\beta\vec{I}^\perp$. Subsequently, we show that an arbitrary null flow is obstructed from monotonicity by a term with direct dependence upon $\rho$ in analogy with the variation found by Christodoulou regarding the mass aspect function $\mu$ (see, for example, \cite{S} Theorem 4.1). We hope that this points towards $\rho$ being potentially closer to an optimal choice of flux for the Hawking Energy $E_H$ in capturing the ambient spacetime.\\
\indent The following proposition is known (see \cite{BHMS}, Lemma 4), we provide proof to complement the Plane Theorem of \cite{BJM} and to establish the result in the notation introduced in Definition \ref{d6}.
\begin{proposition}[Plane Derivation]\label{p3}
Suppose $\Omega\cong I\times\mathbb{S}^2$, for some interval $I\subset \mathbb{R}$, is a hypersurface of $M$ and $\alpha\neq0$ is a smooth function on $\Omega$. Assuming the existence of a foliation of $\Omega$ by admissible spheres $\{\Sigma_s\}$ according to the level set function $s:\Omega\to\mathbb{R}$ whereby $\partial_s|_{\Sigma_s}=\alpha\vec{I}=-\alpha\frac{\vec{H}}{H^2}$ then we have
\begin{align*}
\frac{1}{\sqrt{\frac{|\Sigma_s|}{(16\pi)^3}}}\frac{dE_H}{ds} &=\int_{\Sigma_s}(\bar{\alpha}-\alpha)(2\mathcal{K}_s - \frac12 H^2-2\s{\Delta}\log H)dA\\
&+\int\alpha(2G(e_t,e_t) + |\hat{\text{II}_r}|^2+|\hat{\text{II}_t}|^2 +2|\alpha_H|^2+2|\s{\nabla}\log H|^2 )dA
\end{align*}
for $\text{II}_{r(t)} = \langle\vec{\text{II}},e_{r(t)}\rangle$ where $e_{r(t)}$ is given in Definition \ref{d5}.
\end{proposition}
Before proving Proposition \ref{p3} we will first need to find the second variation of area:
\begin{lemma}\label{l6}
$$\langle \vec{H}, D_{\partial_s}\vec{H}\rangle = -\alpha|\text{II}_r|^2 - \alpha Ric_\Omega(e_r,e_r) -H\s{\Delta}(\frac{\alpha}{H})$$
\end{lemma}
\begin{proof}
In this lemma we temporarily denote the induced covariant derivative on $\Omega$ by $D$ noticing that $\langle\vec{H},D_{\partial_s}\vec{H}\rangle$ calculates the same quantity as if the ambient connection was used. Taking a local basis $\{X_1,X_2\}$ along the foliation we define $\gamma_{ij}:=\langle X_i,X_j\rangle$ giving rise to the inverse metric $\gamma^{ij}$. For any $V\in \Gamma(T\Omega)$ parallel to the leaves of the foliation (i.e. $V|_{\Sigma_s} = \Gamma(T\Sigma_s)$ we have $[\partial_s,V]s = \partial_s(V(s))-V(\partial_s(s)) = 0$ giving $[\partial_s,V]|_{\Sigma_s}\in \Gamma(T\Sigma_s)$. As such
\begin{align*}
\langle \vec{H},D_{\partial_s}\vec{H}\rangle &= \langle \vec{H}, D_{\partial_s}(\gamma^{ij}\vec{\text{II}}_{ij})\rangle\\
&= -\gamma^{ik}\gamma^{jl}(\langle D_{\partial_s}X_k,X_l\rangle + \langle D_{\partial_s}X_l,X_k\rangle)\langle \vec{\text{II}}_{ij},\vec{H}\rangle + \gamma^{ij}\langle D_{{\partial_s}}(D_{X_i}X_j-\s{\nabla}_{X_i}X_j),\vec{H}\rangle\\
&= -2\gamma^{ik}\gamma^{jl}(\langle [{\partial_s},X_k],X_l\rangle - \langle \vec{\text{II}}_{kl},{\partial_s} \rangle) \langle \vec{\text{II}}_{ij},\vec{H}\rangle\\
&\indent+\gamma^{ij}(\langle -R_{{\partial_s} X_i}X_j +D_{X_i}D_{{\partial_s}}X_j +D_{[{\partial_s}, X_i]}X_j - D_{\s{\nabla}_{X_i}X_j}{\partial_s},\vec{H}\rangle 
\end{align*}
where we used the fact that $[{\partial_s},\s\nabla_{X_i}X_j]|_{\Sigma_s}\in T\Sigma_s$ to get the last term.  From the fact that $\partial_s = \alpha\vec{I}$ it follows that
\begin{align*}
2\gamma^{ik}\gamma^{jl}(\langle [{\partial_s},X_k],X_l\rangle - \langle \vec{\text{II}}_{kl},{\partial_s}\rangle)\langle \vec{\text{II}}_{ij},\vec{H}\rangle
&=2\gamma^{ik}\langle D_{[\partial_s,X_k]}X_i,\vec{H}\rangle-2(-\frac{\alpha}{H^2})\gamma^{ik}\gamma^{jl}\langle\vec{\II}_{kl},\vec{H}\rangle\langle\vec{\II}_{ij},\vec{H}\rangle\\
 &= 2\gamma^{ik}\langle D_{[{\partial_s},X_k]}X_i,\vec{H}\rangle + 2\alpha|\text{II}_r|^2\\
\gamma^{ij}\langle D_{X_i}D_{{\partial_s}}X_j + D_{[{\partial_s},X_i]}X_j,\vec{H}\rangle 
&=\gamma^{ij}\langle D_{X_i}[\partial_s,X_j]+D_{X_i}D_{X_j}\partial_s+D_{[\partial_s,X_i]}X_j,\vec{H}\rangle\\
&=\gamma^{ij}\langle [X_i,[\partial_s,X_j]]+D_{[\partial_s,X_j]}X_i+D_{[\partial_s,X_i]}X_j+D_{X_i}D_{X_j}\partial_s,\vec{H}\rangle\\
&= 2\gamma^{ij}\langle D_{[{\partial_s},X_i]}X_j, \vec{H}\rangle +\gamma^{ij}\langle D_{X_i}D_{X_j}{\partial_s},\vec{H}\rangle
\end{align*}
having used the fact that $[X_i,[\partial_s,X_j]]\in\Gamma(T\Sigma_s)$ to get the final equality. This allows us to simplify to
\begin{align*}
\langle \vec{H}, D_{{\partial_s}}\vec{H}\rangle = -2\alpha|\text{II}_r|^2+Ric_\Omega({\partial_s},\vec{H}) + \gamma^{ij}\langle(D_{X_i}D_{X_j}{\partial_s}-D_{\s{\nabla}_{X_i}X_j}{\partial_s}),\vec{H}\rangle.
\end{align*}
Given also that ${\partial_s} = \frac{\alpha}{H}e_r$ we see $\langle\partial_s,D_{X}e_r\rangle = \frac{\alpha}{2H}X\langle e_r,e_r\rangle = 0$ for any $X\in\Gamma(T\Sigma_s)$ so we simplify the last two terms
\begin{align*}
\langle D_{X_i}D_{X_j}\partial_s,\vec{H}\rangle &= -HX_iX_j\langle\partial_s,e_r\rangle+HX_i\langle\partial_s,D_{X_j}e_r\rangle+H\langle D_{X_j}\partial_s,D_{X_i}e_r\rangle\\
&=-HX_iX_j(\frac{\alpha}{H})+\alpha\langle D_{X_j}e_r,D_{X_i}e_r\rangle\\
&=-HX_iX_j(\frac{\alpha}{H})+\alpha\gamma^{kl}\langle D_{X_j}e_r,X_k\rangle\langle D_{X_i}e_r,X_l\rangle\\
&=-HX_iX_j(\frac{\alpha}{H})+\alpha|\II_r|^2\\
\langle D_{\s{\nabla}_{X_i}X_j}{\partial_s},\vec{H}\rangle &=-H\s\nabla_{X_i}X_j\langle\partial_s,e_r\rangle+H\langle\partial_s,D_{\s\nabla_{X_i}X_j}e_r\rangle\\
&=-H\s\nabla_{X_i}X_j(\frac{\alpha}{H})
\end{align*}
and the result follows after we collect all the terms and take a trace over $i,j$.
\end{proof}
\begin{proof}{(Proposition \ref{p3})}
The proof follows in parallel to the Plane Theorem of \cite{BJM} (Theorem 2.1). From the first variation of area formula:
\begin{align*}
\dot{dA_s} &= -\langle \vec{H},{\partial_s}\rangle dA_s = \alpha dA_s\\
\implies \dot{|\Sigma_s|} &= |\Sigma_s|\bar{\alpha}(s).
\end{align*}
So variation of the Hawking Energy gives:
\begin{align*}
\frac{d{E_H}}{ds} &= \frac{d}{ds}\Big(\sqrt{\frac{|\Sigma_s|}{(16\pi)^3}}\Big(16\pi-\int H^2 dA_s\Big)\Big)\\
&= \sqrt{\frac{|\Sigma_s|}{(16\pi)^3}}\Big(\frac12\bar{\alpha}\Big(16\pi -\int H^2 dA_s\Big) -2\int \langle \vec{H},D_{{\partial_s}}\vec{H}\rangle dA_s - \int \alpha H^2dA_s\Big)\\
&= \sqrt{\frac{|\Sigma_s|}{(16\pi)^3}}\Big(\int \bar{\alpha}(2\mathcal{K}_s-\frac12 H^2)dA_s + \int 2\alpha|\text{II}_r|^2 + 2\alpha Ric_\Omega(e_r,e_r) + 2H\s{\Delta}(\frac{\alpha}{H})dA - \int \alpha H^2dA\Big)
\end{align*}
where we used the Gauss-Bonnet Theorem and Proposition \ref{p3} respectively to get the first and second integrands of the last line. As in \cite{BJM} we now trace the Gauss equation for $\Sigma_s$ in $\Omega$ twice over $\Sigma_s$ to get
$$2Ric_\Omega(e_r,e_r) = S - 2\mathcal{K}_s +H^2 - |\text{II}_r|^2$$
for $S$ the scalar curvature of $\Omega$. We then trace the Gauss equation for $\Omega$ in $M$ twice over $\Omega$ to conclude
$$S = 2G(e_t,e_t) + 2|\alpha_H|^2 + |\text{II}_t|^2$$
since $e_t\in \Gamma(T^{\perp}\Omega)$. Substitution into our variation of $E_H$ therefore gives us after some algebraic manipulation that
\begin{align*}
\frac{dE_H}{ds} &= \sqrt{\frac{|\Sigma_s|}{(16\pi)^3}}\Big(\int (\bar{\alpha} -\alpha)(2\mathcal{K}_s-\frac12 H^2)dA_s\\
&+ \int 2G(e_t,e_t) + (|\text{II}_r|^2 - \frac12H^2) +|\text{II}_t|^2 + 2|\alpha_H|^2 + 2H\s{\Delta}(\frac{\alpha}{H})dA_s\Big).
\end{align*} 
First performing an integration by parts on the last term
$$\int H\s{\Delta}(\frac{\alpha}{H}) dA = \int (\s{\Delta}H)\frac{\alpha}{H}dA$$
followed by the identity $\frac{\s{\Delta}H}{H} = \s{\Delta}\log H + |\s{\nabla}\log H|^2$ we obtain the first line of the variation in Proposition \ref{p3}. The second follows from the fact that
\begin{align*}
|\text{II}_r|^2 &= |\hat{\text{II}}_r|^2 + \frac12H^2\\
\text{II}_t &= \hat{\text{II}}_t
\end{align*}
\end{proof}
We refer the reader to \cite{BJM} (Theorem 2.2) for proof of the Cylinder Theorem:
\begin{proposition}\label{p4}
Under the same hypotheses as in Proposition \ref{p3} with $\partial_s=\beta\vec{I}^{\perp}$ for some smooth function $\beta\neq0$ on $\Omega$ and $\vec{I}^\perp = \frac{e_t}{H}$ we have
$$\frac{1}{\sqrt{\frac{|\Sigma_s|}{(16\pi)^3}}}\frac{dE_H}{ds}=\int_{\Sigma_s} \beta(2G(e_t,e_r)+2\langle \hat{\text{II}}_r,\hat{\text{II}}_t\rangle+4\alpha_H(\s{\nabla}\log H)+2\s{\nabla}\cdot\alpha_H) dA_s.$$
\end{proposition}
The full variation of $E_H$ is known from (\cite{BHMS}, Lemma 3), we are now in a position to show it within our context:
\begin{corollary}\label{c1}
Under the same hypotheses as Proposition \ref{p3} and \ref{p4} with $\partial_s = \alpha\vec{I}+\beta\vec{I}^{\perp}$ 
\begin{align*}
\frac{1}{\sqrt{\frac{|\Sigma_s|}{(16\pi)^3}}}\frac{dE_H}{ds}&=\int_{\Sigma_s} (\bar{\alpha}-\alpha)(2\mathcal{K}_s - \frac12 H^2-2\s{\Delta}\log H)dA_s\\
&+\int_{\Sigma_s} \alpha(2G(e_t,e_t) + |\hat{\text{II}_r}|^2+|\hat{\text{II}_t}|^2 +2|\alpha_H|^2+2|\s{\nabla}\log H|^2 )dA_s\\
&+\int_{\Sigma_s} \beta(2G(e_t,e_r)+2\langle \hat{\text{II}}_r,\hat{\text{II}}_t\rangle+4\alpha_H(\s{\nabla}\log H)+2\s{\nabla}\cdot\alpha_H) dA_s
\end{align*}
\end{corollary}
\begin{proof}
As in \cite{BJM} (Theorem 1.13) variation of $E_H$ is achieved by summing the contributions from Propositions 3 and 4 since the variation of the area form and the mean curvature vector are known to be $\mathbb{R}$-linear over the flow vector decomposition.
\end{proof}
Subsequently, we achieve an arbitrary past(future) directed null flow by setting $\alpha=\mp \beta>0$ in Corollary \ref{c1} giving $\partial_s = \alpha(\vec{I}\mp\vec{I}^{\perp})$ and
\begin{align*}
\frac{1}{\sqrt{\frac{|\Sigma_s|}{(16\pi)^3}}}\frac{dE_H}{ds} &= \int_{\Sigma_s}(\bar{\alpha}-\alpha)(2\mathcal{K}_s-\frac12 H^2 \pm 2\s{\nabla}\cdot\alpha_H-2\s{\Delta}\log H)dA_s\\
&+\int_{\Sigma_s}\alpha(2G(e_t,e_t\mp e_r) +|\hat{\text{II}}_r \mp \hat{\text{II}}_t|^2+2|\alpha_H\mp \s{\nabla}\log H|^2)dA_s.
\end{align*}
It follows in an energy dominated spacetime that the only obstruction to a non-decreasing Hawking energy is the integrand
$$(\bar{\alpha}-\alpha)(2\mathcal{K}_s-\frac12 H^2 \pm 2\s{\nabla}\cdot\alpha_H-2\s{\Delta}\log H)=2(\bar{\alpha}-\alpha)\rho_{\mp}.$$
In particular $\frac{dE_H}{ds}\geq 0$ for any foliation where $\rho$ is constant on each $\Sigma_s$, moreover, since $m_{\mp}(\Sigma_s)=E_H(\Sigma_s)$ in this case (provided also $\rho_{\mp}\geq 0$) we have monotonicity of our quasi local mass as well. We extend beyond this case in the next section with the proof of Theorem \ref{t1}.
\newpage
\section{Propagation of $\rho$}
In this section we will work towards proving Theorem \ref{t1} by finding the propagation of our flux function $\rho$ along an arbitrary null flow.
\subsection{Setup}
We adopt the same setup as in \cite{MS1} which we summarize here in order to introduce our notation:\\
Suppose $\Omega$ is a smooth connected, null hypersurface embedded in $(M,\langle\cdot,\cdot\rangle)$. Here we let $\ubar L$ be a smooth, non-vanishing, null vector field of $\Omega$, $\ubar L\in\Gamma(T\Omega)$. It's a well known fact (see, for example, \cite{C}) that the integral curves of $\ubar L$ are pre-geodesic so we're able to find $\kappa\in \mathcal{F}(\Omega)$ such that $D_{\ubar L}\ubar L = \kappa\ubar L$. 

\begin{wrapfigure}{r}{7.5cm}
\centering
\def\svgwidth{300pt} 
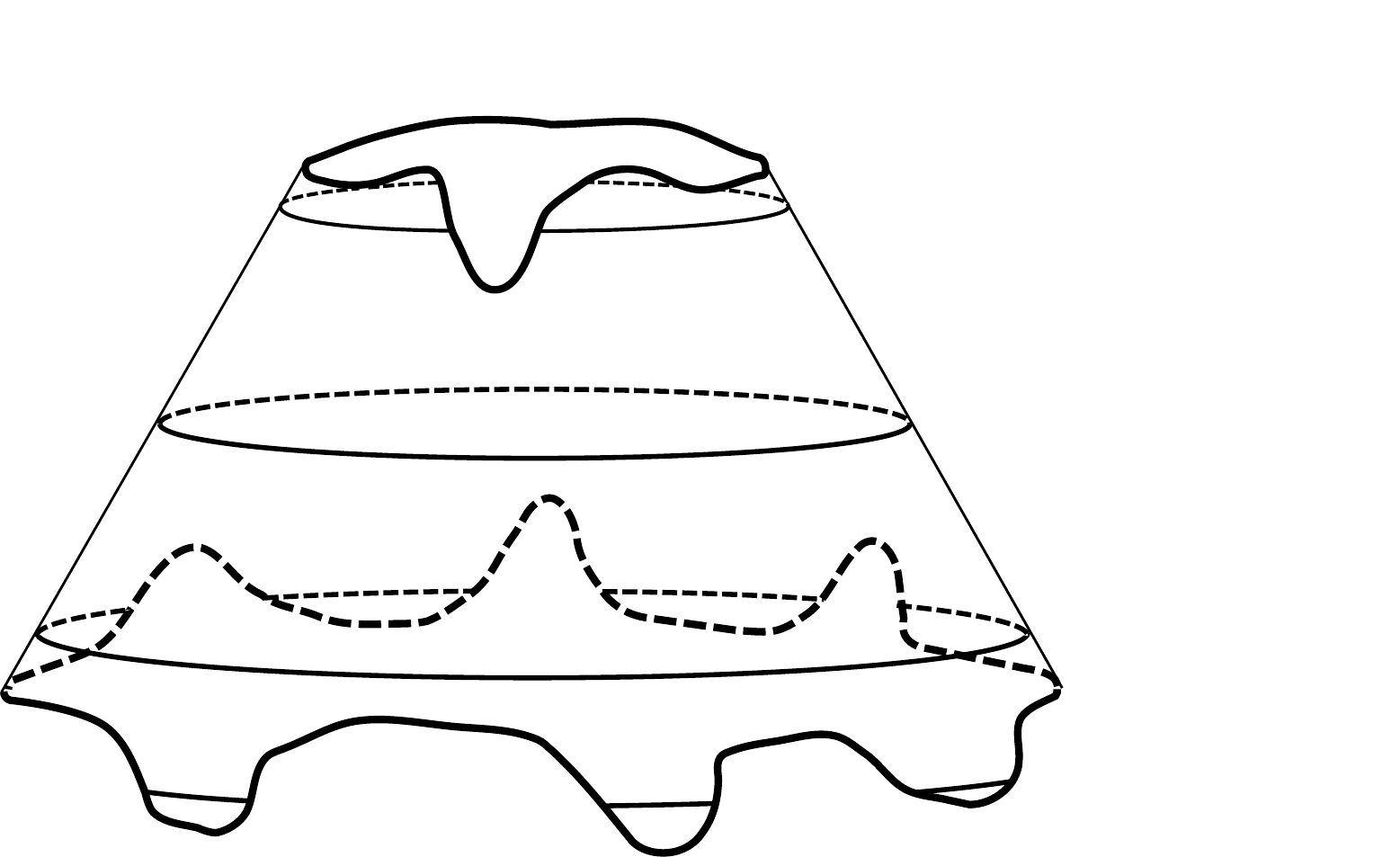
\caption*{$s_3<s_2<S_-<s_0<S_+<s_1$}
\end{wrapfigure}

We assume the existence of an embedded sphere $\Sigma$ in $\Omega$ such that any integral curve of $\ubar L$ intersects $\Sigma$ precisely once. As previously used, we will refer to such $\Sigma$ as \textit{cross sections} of $\Omega$. This gives rise to a natural submersion $\pi:\Omega\to\Sigma$ sending $p\in\Omega$ to the intersection with $\Sigma$ of the integral curve $\gamma_p^{\ubar L}$ of $\ubar L$ for which $\gamma_p^{\ubar L}(0) = p$. Given $\ubar L$ and a constant $s_0$ we may construct a function $s\in\mathcal{F}(\Omega)$ from $\ubar L(s)=1$ and $s|_{\Sigma}= s_0$. For $q\in\Sigma$, if $(s_-(q),s_+(q))$ represents the range of $s$ along $\gamma_q^{\ubar L}$ then letting ${S_-}=\sup_{\Sigma}s_-$ and ${S_+}=\inf_{\Sigma}s_+$ we notice that the interval $(S_-,S_+)$ is non-empty. Given that $\ubar L(s)=1$ the Implicit Function Theorem gives for $t\in(S_-,S_+)$ that $\Sigma_t:= \{p\in\Omega|s(p)=t\}$ is diffeomorphic to $\mathbb{S}^2$ through $\Sigma$. For $s<S_-$ or $s>S_+$, in the case that $\Sigma_s$ is non-empty, although smooth it may no longer be connected. We have that the collection $\{\Sigma_s\}$ gives a foliation of $\Omega$.\\
We construct another null vector field $L$ by assigning at every $p\in\Omega$ $L|p\in T_pM$ be the unique null vector satisfying $\langle \ubar L,L\rangle = 2$ and $\langle L,v\rangle = 0$ for any $v\in T_p\Sigma_{s(p)}$. As before each $\Sigma_s$ is endowed with an induced metric $\gamma_s$, two null second fundamental forms $\ubar\chi = -\langle \vec{\II},\ubar L\rangle$ and $\chi=-\langle \vec{\II},L\rangle$ as well as the connection 1-form (or torsion) $\zeta(V) = \frac12\langle D_V\ubar L,L\rangle$. We will need the following known result (\cite{S}):
\begin{lemma}\label{l7} Given $V\in\Gamma(T\Sigma_s)$,
\begin{itemize}
\item $D_V\ubar L = \vec{\ubar\chi}(V) + \zeta(V)\ubar L$
\item $D_V L = \vec{\chi}(V) -\zeta(V)L$
\item $D_{\ubar L}L = -2\vec{\zeta}-\kappa L$
\end{itemize}
\end{lemma}
where, given $V,W\in\Gamma(T\Sigma)$, the vector fields $\vec{\zeta},\,\vec{\ubar\chi}(V)$ are uniquely determined by $\langle \vec{\zeta},V\rangle = \zeta(V)$ and $\langle \vec{\ubar\chi}(V),W\rangle = \ubar\chi(V,W)$.
\begin{proof}
It suffices to check all identities agree by taking the metric inner product with vectors $\ubar L,L$ and an extension $W$ satisfying $W|_{\Sigma_s}\in\Gamma(T\Sigma_s)$ keeping in mind that $[\ubar L,W]|_{\Sigma_s}\in\Gamma(T\Sigma_s)$. We leave this verification to the reader.
\end{proof}
For any cross section $\Sigma$ of $\Omega$ and $v\in T_q(\Sigma)$ we may extend $v$ along the generator $\gamma_q^{\ubar L}$ according to
\begin{align*}
\dot{V}(s) &= D_{V(s)}\ubar L\\
V(0)&=v.
\end{align*} 
Since $x\in T_p\Omega\iff \langle \ubar L|_p,x\rangle = 0$ we see from the fact that $\dot{\langle V(s),\ubar L\rangle}=\langle D_{V(s)}\ubar L,\ubar L\rangle +\kappa\langle V(s),\ubar L\rangle= \frac12V(s)\langle \ubar L,\ubar L\rangle+\kappa\langle V(s),\ubar L\rangle=\kappa\langle V(s),\ubar L\rangle$ and $\langle V(0),\ubar L\rangle=0$ we can solve to get $\langle V(s),\ubar L\rangle=0$ for all $s$. As a result any section $W\in\Gamma(T\Sigma)$ is extended to all of $\Omega$ satisfying $[\ubar L,W]=0$. We also notice along each generator $0=[\ubar L,W]s = \ubar L(Ws)=\dot{Ws}$ such that $Ws|_\Sigma=0$ forces $Ws = 0$ on all of $\Omega$. We conclude that $W|_{\Sigma_s}\in\Gamma(T\Sigma_s)$ and denote by $E(\Sigma)\subset\Gamma(T\Omega)$ the set of such extensions off of $\Sigma$ along $\ubar L$.  We also note that linear independence is preserved along generators by standard uniqueness theorems allowing us to extend basis fields $\{X_1,X_2\}\subset \Gamma(T\Sigma)$ off of $\Sigma$ as well.
\subsection{The Structure Equations}
We will need to propagate the Christoffel symbols with the known result (\cite{S}):
\begin{lemma}\label{l8}
Given $U,V,W\in E(\Sigma)$,
$$\langle [\ubar L,\s\nabla_VW],U\rangle = (\s\nabla_V\ubar\chi)(W,U)+(\s\nabla_W\ubar\chi)(V,U)-(\s\nabla_U\ubar\chi)(V,W)$$
where $\s\nabla$ the induced covariant derivative on each $\Sigma_s$.
\end{lemma}
\begin{proof}
Starting from the Koszul formula
$$2\langle\s\nabla_VW,U\rangle = V\langle W,U\rangle+W\langle V,U\rangle - U\langle V,W\rangle-\langle V,[W,U]\rangle+\langle W,[U,V]\rangle+\langle U,[V,W]\rangle$$
we apply $\ubar L$ to the left hand term to get
$$\ubar L\langle \s\nabla_VW,U\rangle = \langle D_{\ubar L}\s\nabla_VW,U\rangle+\langle \s\nabla_VW,D_{\ubar L}U\rangle=\langle [\ubar L,\s\nabla_VW],U\rangle + 2\ubar\chi(\s\nabla_VW,U)$$
and to the right keeping in mind that $[V,W]\in E(\Sigma)$
\begin{align*}
\ubar L\Big(V\langle W,U\rangle&+W\langle V,U\rangle - U\langle V,W\rangle-\langle V,[W,U]\rangle+\langle W,[U,V]\rangle+\langle U,[V,W]\rangle\Big)\\
&=V\ubar L\langle W,U\rangle+W\ubar L\langle V,U\rangle-U\ubar L\langle V,W \rangle - 2\ubar\chi(V,[W,U])+2\ubar\chi(W,[U,V])+2\ubar\chi(U,[V,W])\\
&=2\Big(V\ubar\chi(W,U)+W\ubar\chi(V,U)-U\ubar\chi(V,W)- \ubar\chi(V,[W,U])+\ubar\chi(W,[U,V])+\ubar\chi(U,[V,W])\Big)\\
&=2\Big((\s\nabla_V\ubar\chi)(W,U)+(\s\nabla_W\ubar\chi)(V,U)-(\s\nabla_U\ubar\chi)(V,W)+2\ubar\chi(\s\nabla_VW,U)\Big).
\end{align*}
Equating terms according to the Koszul formula the result follows upon cancellation of the term $\ubar\chi(\s\nabla_VW,U)$.
\end{proof}
\begin{lemma}\label{l9}
For $S,U,V,W\in E(\Sigma)$,
\begin{align*}
\langle [\ubar L, \s R_{VW}U],S\rangle &= (\s\nabla_W\s\nabla_V\ubar\chi)(U,S)-(\s\nabla_V\s\nabla_W\ubar\chi)(U,S)+(\s\nabla_W\s\nabla_U\ubar\chi)(V,S)\\
&-(\s\nabla_V\s\nabla_U\ubar\chi)(W,S)+(\s\nabla_V\s\nabla_S\ubar\chi)(W,U)-(\s\nabla_W\s\nabla_S\ubar\chi)(V,U)
\end{align*}
where $\s{R}$ the induced Riemann curvature tensor on $\Sigma_s$.
\end{lemma}
\begin{proof}
We notice any $f\in\mathcal{F}(\Sigma)$ can be extended to all of $\Omega$ by imposing $\ubar L(f)=0$ along generators. As such $fV\in E(\Sigma)$ and $[\ubar L, \s{R}_{fVW}U]=[\ubar L,\s{R}_{VfW}U]=[\ubar L,\s{R}_{VW}fU]=[\ubar L, f\s{R}_{VW}U]=f[\ubar L,\s{R}_{VW}U]$. Within $E(\Sigma)$ we conclude that both $\langle[\ubar L,\s{R}_{VW}U],S\rangle$ and the right hand side of the identity restricts to 4-tensors pointwise on each $\Sigma_s$. It therefore suffices to prove the identity pointwise. In particular, for any $v,w\in T_q\Sigma_s$ we extend to vector fields $V,W\in E(\Sigma)$ such that $\s{\nabla}_VW|_q=0$. The Riemann tensor on $\Sigma_s$ reads
\begin{align*}
\langle \s R_{VW}U,S\rangle &=\langle \s\nabla_{[V,W]}U,S\rangle-\langle\s\nabla_V\s\nabla_WU,S\rangle+
\langle\s\nabla_W\s\nabla_VU,S\rangle\\
&=\langle\s\nabla_{[V,W]}U,S\rangle -V\langle\s\nabla_WU,S\rangle+\langle\s\nabla_WU,\s\nabla_VS\rangle+W\langle\s\nabla_VU,S\rangle-\langle\s\nabla_VU,\s\nabla_WS\rangle 
\end{align*}
so applying $\ubar L$ to the terms on the right assuming restriction to $q\in\Sigma_s$ we have
\begin{align*}
&\ubar L\Big(\langle\s\nabla_{[V,W]}U,S\rangle -V\langle\s\nabla_WU,S\rangle+\langle\s\nabla_WU,\s\nabla_VS\rangle+W\langle\s\nabla_VU,S\rangle-\langle\s\nabla_VU,\s\nabla_WS\rangle\Big)\\
&=\langle[\ubar L,\s\nabla_{[V,W]}U],S\rangle-V\ubar L\langle\s\nabla_WU,S\rangle+W\ubar L\langle\s\nabla_VU,S\rangle\\
&=-V\langle[\ubar L,\s\nabla_WU],S\rangle-2V\ubar\chi(\s\nabla_WU,S)+W\langle[\ubar L,\s\nabla_VU],S\rangle+2W\ubar\chi(\s\nabla_VU,S)
\end{align*}
where the first term in the second line vanishes as a result of Lemma \ref{l8} since $[V,W]\in E(\Sigma)$ and $[V,W]|_q=0$. Using Lemma \ref{l8} on the first and third terms of the third line we get
\begin{align*}
=-V\Big((\s\nabla_W\ubar\chi)(U,S)&+(\s\nabla_U\ubar\chi)(W,S)-(\s\nabla_S\ubar\chi)(W,U)\Big)-2V\ubar\chi(\s\nabla_WU,S)\\
&+W\Big((\s\nabla_V\ubar\chi)(U,S)+(\s\nabla_U\ubar\chi)(V,S)-(\s\nabla_S\ubar\chi(V,U)\Big)+2W\ubar\chi(\s\nabla_VU,S)\\
=-(\s\nabla_V\s\nabla_W\ubar\chi)(U,S)&-(\s\nabla_V\s\nabla_U\ubar\chi)(W,S)+(\s\nabla_V\s\nabla_S\ubar\chi)(W,U)-2V\ubar\chi(\s\nabla_WU,S)\\
&+(\s\nabla_W\s\nabla_V\ubar\chi)(U,S)+(\s\nabla_W\s\nabla_U\ubar\chi)(V,S)-(\s\nabla_W\s\nabla_S\ubar\chi)(V,U)+2W\ubar\chi(\s\nabla_VU,S).
\end{align*}
We also note that restriction to $q\in\Sigma_s$ gives
$$0=(\s\nabla_V\ubar\chi)(\s\nabla_WU,S)=V\ubar\chi(\s\nabla_WU,S)-\ubar\chi(\s\nabla_V\s\nabla_WU,S)$$
allowing us to simplify the remaining terms above
$$-2V\ubar\chi(\s\nabla_WU,S)+2W\ubar\chi(\s\nabla_VU,S) = 2\Big(-\ubar\chi(\s\nabla_V\s\nabla_WU,S)+\ubar\chi(\s\nabla_W\s\nabla_VU,S)\Big)
=2\ubar\chi(\s R_{VW}U,S).$$
Since
$$\ubar L\langle\s R_{VW}U,S\rangle = \langle[\ubar L,\s R_{VW}U],S\rangle + 2\ubar\chi(\s R_{VW}U,S)$$
the result follows upon cancellation of $2\ubar\chi(\s R_{VW}U,S)$ given that $q$ was arbitrarily chosen.
\end{proof}
Now we're in a position to find the structure equations that we'll need to propagate $\rho$. Recalling that the tensors $\gamma_s$, $\ubar\chi$, $\chi$ and $\zeta$ are restrictions of associated tensors on $\Omega$ we measure their propagation with the Lie derivative along $\ubar L$. The following proposition is known (\cite{S},\cite{G}), we provide proof for completeness:
\begin{proposition}[Structure Equations]\label{p5}
\begin{align}
\ubar L\mathcal{K} &= -\tr\ubar\chi\mathcal{K}-\frac12\s\Delta \tr\ubar\chi+\s\nabla\cdot(\s\nabla\cdot\hat{\ubar\chi})\\
\mathcal{L}_{\ubar L}\gamma &= 2\ubar\chi\\
\mathcal{L}_{\ubar L}\ubar\chi &= -\ubar\alpha + \frac12|\hat{\ubar\chi}|^2\gamma + \tr\ubar\chi\hat{\ubar\chi}+\frac14(\tr\ubar\chi)^2\gamma + \kappa\ubar\chi\\
\ubar L\tr\ubar\chi &= -\frac12(\tr\ubar\chi)^2 - |\hat{\ubar\chi}|^2 - G(\ubar L,\ubar L) +\kappa \tr\ubar\chi\\
\mathcal{L}_{\ubar L}\chi&=\Big(\mathcal{K}+\hat{\ubar\chi}\cdot\hat{\chi}+\frac12G(\ubar L,L)\Big)\gamma+\frac12\tr\ubar\chi\hat{\chi}+\frac12\tr\chi\hat{\ubar\chi}-\hat{G}-2S(\s\nabla\zeta)-2\zeta\otimes\zeta-\kappa\chi\\
\ubar L\tr\chi &= G(\ubar L,L)+2\mathcal{K}-2\s{\nabla}\cdot\zeta-2|\zeta|^2-\langle\vec{H},\vec{H}\rangle-\kappa\tr\chi\\
\mathcal{L}_{\ubar L}\zeta &= G_{\ubar L} - \s\nabla\cdot\hat{\ubar\chi} - \tr\ubar\chi\zeta+\frac12\s{d}\tr\ubar\chi+\s{d}\kappa
\end{align}
where $\ubar\alpha$ is the symmetric 2-tensor given by $\ubar\alpha(V,W) = \langle R_{\ubar L V}\ubar L,W\rangle$, $S(T)$ represents the symmetric part of a 2-tensor $T$, $G_{\ubar L} = G(\ubar L,\cdot)|_{\Sigma_s}$ and $\hat{G} = G|_{\Sigma_s}-\frac12(\tr_{\gamma}G)\gamma$.
\end{proposition}
\begin{proof}
We prove each equation in turn, when used we will assume $S,U,V,W\in E(\Sigma)$:
\begin{enumerate}
\item Since $\Sigma$ is of dimension two we have
$$\mathcal{K}\{\langle V,U\rangle\langle W,S\rangle-\langle V,S\rangle\langle W,U\rangle\}=\langle\s R_{VW}U,S\rangle$$
Applying $\ubar L$ to the left hand side of the equality we get
\begin{align*}
(\ubar L\mathcal{K})\{\langle V,U\rangle\langle W,S\rangle&-\langle V,S\rangle\langle W,U\rangle\}\\
&+2\mathcal{K}\{\langle W,S\rangle\ubar\chi(V,U)+\langle V,U\rangle\ubar\chi(W,S)-\langle W,U\rangle\ubar\chi(V,S)-\langle V,S\rangle\ubar\chi(W,U)\}
\end{align*}
so that a trace over $V,U$ and then $W,S$ gives
$$2\ubar L\mathcal{K}+4\tr\ubar\chi\mathcal{K}.$$
Applying $\ubar L$ to the right hand side we have
$$\ubar L\langle\s R_{VW}U,S\rangle = \langle[\ubar L,\s R_{VW}U],S\rangle+2\ubar\chi(\s R_{VW}U,S)$$
allowing us to use Lemma \ref{l9}. Taking a trace over $V,U$ and $W,S$ we get
$$2\s\nabla\cdot\s\nabla\cdot\ubar\chi-2\s\Delta tr\ubar\chi + 2tr\ubar\chi K$$
having used the fact that $\s {Ric}=K\gamma$ in obtaining the last term. Equating terms we conclude that
$$\ubar LK= \s\nabla\cdot\s\nabla\cdot\ubar\chi-\s\Delta tr\ubar\chi-tr\ubar\chi K = \s\nabla\cdot\s\nabla\cdot\hat{\ubar\chi}-\frac12\s\Delta tr\ubar\chi-tr\ubar\chi K$$
\item Coming from Lemma \ref{l7} we have already made extensive use of this identity:\\
 $\begin{aligned}[t]
(\mathcal{L}_{\ubar L}\gamma)(V,W) &=\ubar L \langle V,W\rangle = \langle D_{\ubar L}V,W\rangle+\langle V,D_{\ubar L}W\rangle\\
&=\langle D_V\ubar L,W\rangle + \langle V,D_W\ubar L\rangle\\
&=2\ubar\chi(V,W)
\end{aligned}$
\item $\begin{aligned}[t]
(\mathcal{L}_{\ubar L}\ubar\chi)(V,W) &=\ubar L\ubar\chi(V,W) = \ubar L\langle D_V\ubar L, W\rangle\\
&=\langle D_{\ubar L}D_V \ubar L, W\rangle + \langle D_V \ubar L, D_{\ubar L}W\rangle\\
&= \langle R_{V\ubar L}\ubar L + D_VD_{\ubar L}\ubar L,W\rangle + \langle \vec{\ubar\chi}(V) +\zeta(V)\ubar L, \vec{\ubar\chi}(W)+\zeta(W)\ubar L\rangle\\
&=-\langle R_{\ubar L V}\ubar L, W\rangle +\kappa\ubar\chi(V,W) + \langle \vec{\ubar\chi}(V),\vec{\ubar\chi}(W)\rangle
\end{aligned}$

having used Lemma \ref{l7} to get the third line. Since $\vec{\ubar\chi}(V) = \vec{\hat{\ubar\chi}}(V)+\frac12\tr\ubar\chi V$ we see that
\begin{align*}
\langle\vec{\ubar\chi}(V),\vec{\ubar\chi}(W)\rangle &= \langle\vec{\hat{\ubar\chi}}(V),\vec{\hat{\ubar\chi}}(W)\rangle+\tr\ubar\chi\hat{\ubar\chi}(V,W)+\frac14(\tr\ubar\chi)^2\langle V,W\rangle\\
&=\frac12|\hat{\ubar\chi}|^2\langle V,W\rangle+\tr\ubar\chi\hat{\ubar\chi}(V,W)+\frac14(\tr\ubar\chi)^2\langle V,W\rangle
\end{align*}
using the fact that $AB+BA=tr(AB)\mathbb{I}$ for traceless symmetric $2\times2$ matrices to get the second line. The result follows.
\item We will denote tensor contraction between the contravariant $a$-th and covariant $b$-th slots by $C_b^a$. Extending a local basis off of $\Sigma$ and applying Gram-Schmidt we get an orthonormal frame field $\{E_1,E_2\}$ allowing us to write $g^{-1}|_\Omega=\gamma^{-1} = E_1\otimes E_1+E_2\otimes E_2$ and $\gamma = E^{\flat}_1\otimes E^{\flat}_1+E^{\flat}_2\otimes E^{\flat}_2$. It's an easy exercise to show $C_1^2 \gamma^{-1}\otimes\gamma = \delta - \ubar L\otimes ds$ whereby $\delta(\eta,X) = \eta(X)$ for any 1-form $\eta$ and vector field $X$. Since $\delta$ and $\ubar L\otimes ds$ are Lie constant along $\ubar L$
$$0=\mathcal{L}_{\ubar L}C_1^2\gamma^{-1}\otimes\gamma = C_1^2(\mathcal{L}_{\ubar L}\gamma^{-1}\otimes\gamma+\gamma^{-1}\otimes2\ubar\chi)$$
giving
\begin{align*}
-2C_1^2C_1^2\gamma^{-1}\otimes\ubar\chi\otimes\gamma^{-1} &=-C_1^2(C_1^2\gamma^{-1}\otimes 2\ubar\chi)\otimes\gamma^{-1}\\
&=C_1^2(C_1^2\mathcal{L}_{\ubar L}\gamma^{-1}\otimes\gamma)\otimes\gamma^{-1}\\
&=C_1^2\mathcal{L}_{\ubar L}\gamma^{-1}\otimes(C_2^1\gamma\otimes\gamma^{-1})\\
&=C_1^2\mathcal{L}_{\ubar L}\gamma^{-1}\otimes(E^{\flat}_1\otimes E_1+E^{\flat}_2\otimes E_2)\\
&=\mathcal{L}_{\ubar L}\gamma^{-1}.
\end{align*}
As a result

$\begin{aligned}[t]
\ubar L\tr\ubar\chi &= \mathcal{L}_{\ubar L}C^1_1C^2_2\gamma^{-1}\otimes\ubar\chi\\
&= C^1_1C^2_2(\mathcal{L}_{\ubar L}\gamma^{-1}\otimes\ubar\chi+\gamma^{-1}\otimes\mathcal{L}_{\ubar L}\ubar\chi)\\
&=-2|\ubar\chi|^2+\tr\mathcal{L}_{\ubar L}\ubar\chi\\
&=-2|\ubar\chi|^2-Ric(\ubar L,\ubar L) + \kappa \tr\ubar\chi + |\ubar\chi|^2\\
&= -(\hat{\ubar\chi}+\frac12\tr\ubar\chi\gamma)\cdot(\hat{\ubar\chi}+\frac12\tr\ubar\chi\gamma) - G(\ubar L,\ubar L) +\kappa\tr\ubar\chi\\
&= -\frac12(\tr\ubar\chi)^2 - |\hat{\ubar\chi}|^2-G(\ubar L,\ubar L)+\kappa\tr\ubar\chi
\end{aligned}$
\item$\begin{aligned}[t]
(\mathcal{L}_{\ubar L}\chi)(V,W)&=\ubar L\langle D_VL,W\rangle=\langle D_{\ubar L}D_VL,W\rangle+\langle D_VL,D_{\ubar L}W\rangle\\
&=\langle R_{V\ubar L}L+D_VD_{\ubar L}L,W\rangle+\langle \vec{\chi}(V)-\zeta(V)L,\vec{\ubar\chi}(W)+\zeta(W)\ubar L\rangle\\
&=\langle R_{V\ubar L}L,W\rangle+V\langle -2\vec{\zeta}-\kappa L,W\rangle+\langle2\vec{\zeta}+\kappa L,D_VW\rangle+\langle\vec{\chi}(V),\vec{\ubar\chi}(W)\rangle-2\zeta(V)\zeta(W)\\
&=\langle R_{V\ubar L}L,W\rangle-2(\s\nabla_V\zeta)(W)-\kappa\chi(V,W)+\langle\vec{\chi}(V),\vec{\ubar\chi}(W)\rangle-2\zeta(V)\zeta(W).
\end{aligned}$

Having used Lemma \ref{l7} to get the second and third lines. 
\begin{lemma}\label{l10} The first term satisfies the identity\\
$\begin{aligned}[t]
\langle R_{V\ubar L}W,L\rangle &=-\Big(\mathcal{K}-\frac14\langle\vec{H},\vec{H}\rangle+\frac12\hat{\ubar\chi}\cdot\hat{\chi}+\frac12G(\ubar L,L)\Big)\langle V,W\rangle+\hat{G}(V,W)\\
&\indent-(\text{curl}\,\zeta)(V,W)+\frac12\Big(\langle\vec{\chi}(V),\vec{\ubar\chi}(W)\rangle-\langle\vec{\ubar\chi}(V),\vec{\chi}(W)\rangle\Big)
\end{aligned}$

\end{lemma}
\begin{proof}
From the first Bianchi identity followed by the Ricci equation (\cite{O}, pg125)

$\begin{aligned}[t]
\langle R_{V\ubar L}W,L\rangle+\langle R_{\ubar L W}V,L\rangle &= \langle R_{VW}\ubar L,L\rangle\\
&=\langle\s{R}^{\perp}_{VW}\ubar L,L\rangle+\langle\tilde{\II}(V,L),\tilde{\II}(W,\ubar L)\rangle-\langle\tilde{\II}(V,\ubar L),\tilde{\II}(W,L)\rangle
\end{aligned}$

where using Lemma \ref{l7}

$\begin{aligned}[t]
\langle \s{R}^{\perp}_{VW}\ubar L,L\rangle&:= \langle D^{\perp}_{[V,W]}\ubar L - [D^{\perp}_V,D^{\perp}_W]\ubar L,L\rangle=-2(\s\nabla_V\zeta)(W)+2(\s\nabla_W\zeta)(V)=-2(\text{curl}\zeta)(V,W)\\
\tilde{\II}(V,\ubar L)&:= D^{||}_V\ubar L=\vec{\ubar\chi}(V)\\
\tilde{\II}(V,L)&:= D^{||}_VL= \vec{\chi}(V).
\end{aligned}$

We conclude that the antisymmetric part satisfies
$$\frac12\Big(\langle R_{V\ubar L}W,L\rangle-\langle R_{W\ubar L}V,L\rangle\Big) = -(\text{curl}\zeta)(V,W)+\frac12\Big(\langle\vec{\chi}(V),\vec{\ubar\chi}(W)\rangle-\langle\vec{\ubar\chi}(V),\vec{\chi}(W)\rangle\Big).$$
Next we find that

$\begin{aligned}[t]
G(V,W)&=\text{Ric}(V,W)-\frac12R\langle V,W\rangle\\
&=\frac12\Big(\langle R_{\ubar L V}L,W\rangle+\langle R_{LV}\ubar L,W\rangle\Big)+\tr_\gamma\langle R_{(\cdot)V}(\cdot),W\rangle+\frac12(G(\ubar L,L)+\tr_\gamma G)\langle V,W\rangle.
\end{aligned}$

Since $\Sigma$ is of dimension two we must have that $\tr_\gamma\langle R_{(\cdot)V}(\cdot),W\rangle\propto \langle V,W\rangle$ with factor of proportionality $\mathcal{K}-\frac14\langle\vec{H},\vec{H}\rangle+\frac12\hat{\ubar\chi}\cdot\hat{\chi}$ coming from Proposition \ref{p1}. We conclude therefore that the symmetric part satisfies
$$\frac12\Big(\langle R_{V\ubar L}W,L\rangle+\langle R_{W\ubar L}V,L\rangle\Big) = \hat{G}(V,W)-(\mathcal{K}-\frac14\langle\vec{H},\vec{H}\rangle+\frac12\hat{\ubar\chi}\cdot\hat{\chi}+\frac12G(\ubar L,L))\langle V,W\rangle$$
and the result follows as soon as we sum up the antisymmetric and symmetric contributions.
\end{proof}
Combining the previous lemma with the propagation of $\chi$ we have

$\begin{aligned}[t]
(\mathcal{L}_{\ubar L}\chi)(V,W)&=\Big(\mathcal{K}-\frac14\langle\vec{H},\vec{H}\rangle+\frac12\hat{\ubar\chi}\cdot\hat{\chi}+\frac12G(\ubar L,L)\Big)\langle V,W\rangle-\hat{G}(V,W)\\
&\indent+(\text{curl}\zeta)(V,W)-\frac12\Big(\langle\vec{\chi}(V),\vec{\ubar\chi}(W)\rangle-\langle\vec{\ubar\chi}(V),\vec{\chi}(W)\rangle\Big)\\
&\indent-2(\s\nabla_V\zeta)(W)-\kappa\chi(V,W)+\langle\vec{\chi}(V),\vec{\ubar\chi}(W)\rangle-2\zeta(V)\zeta(W)\\
&=\Big(\mathcal{K}-\frac14\langle\vec{H},\vec{H}\rangle+\frac12\hat{\ubar\chi}\cdot\hat{\chi}+\frac12G(\ubar L,L)\Big)\langle V,W\rangle - \hat{G}(V,W)\\
&+\indent\frac12\Big(\langle\vec{\ubar\chi}(V),\vec{\chi}(W)\rangle+\langle\vec{\chi}(V),\vec{\ubar\chi}(W)\rangle\Big)-(\s\nabla_V\zeta)(W)-(\s\nabla_W\zeta)(V)\\
&\indent-2\zeta(V)\zeta(W)-\kappa\chi(V,W).
\end{aligned}$

Using again the fact that $AB+BA = \tr(AB)\mathbb{I}$ for symmetric, traceless $2\times2$ matrices it follows that 
$$\langle \vec{\chi}(V),\vec{\ubar\chi}(W)\rangle+\langle\vec{\ubar\chi}(V),\vec{\chi}(W)\rangle =(\hat{\chi}\cdot\hat{\ubar\chi})\langle V,W\rangle+\tr\chi\hat{\ubar\chi}(V,W)+\tr\ubar\chi\hat{\chi}(V,W)+\frac12\langle\vec{H},\vec{H}\rangle\langle V,W\rangle$$
 giving the result.

\item $\begin{aligned}[t]
\ubar L\tr\chi &=\mathcal{L}_{\ubar L}(C_1^1C_2^2\gamma^{-1}\otimes\chi)\\
&=-2\ubar\chi\cdot\chi+\tr(\mathcal{L}_{\ubar L}\chi)\\
&=-2\ubar\chi\cdot\chi+2\mathcal{K}+2\hat{\ubar\chi}\cdot\hat{\chi}+G(\ubar L,L)-2\s\nabla\cdot\zeta-2|\zeta|^2-\kappa\tr\chi\\
&=G(\ubar L,L)+2\mathcal{K}-2\s\nabla\cdot\zeta-2|\zeta|^2-\langle\vec{H},\vec{H}\rangle-\kappa\tr\chi
\end{aligned}$
\item $\begin{aligned}[t]
(\mathcal{L}_{\ubar L}\zeta)(V)&=\ubar L \zeta(V) = \frac12\ubar L\langle D_V\ubar L,L\rangle\\
&= \frac12\langle D_{\ubar L}D_V\ubar L,L\rangle + \frac12\langle D_V \ubar L,D_{\ubar L}L\rangle\\
&= \frac12\langle R_{V\ubar L}\ubar L + D_VD_{\ubar L}\ubar L,L\rangle + \frac12\langle \vec{\ubar\chi}(V)+\zeta(V)\ubar L, -2\vec{\zeta} - \kappa L\rangle\\
&=\frac12\langle R_{V\ubar L}\ubar L,L\rangle + V\kappa+ \kappa\zeta(V) - \ubar\chi(V,\vec{\zeta}) -\kappa\zeta(V)\\
&=-\s\nabla\cdot\hat{\ubar\chi}(V) + \frac12V\tr\ubar\chi - \tr\ubar\chi\zeta(V)+G(V,\ubar L)  +V\kappa
\end{aligned}$

having used Lemma \ref{l7} to obtain the third line and the Codazzi equation (4) to get the fifth. 
\end{enumerate}
\end{proof}

From Proposition \ref{p5} we have the propagation of the first two terms of $\rho$ for the third and forth we'll need
\begin{corollary}\label{c2} 
Assuming $\{\Sigma_s\}$ is expanding along $\ubar L$ we have
\begin{align*}
\ubar L(\s{\nabla}\cdot\zeta) &= -2\s{\nabla}\cdot(\hat{\ubar\chi}\cdot\zeta)-2\tr\ubar\chi\s{\nabla}\cdot\zeta
-\s{\nabla}\cdot\s{\nabla}\cdot\hat{\ubar\chi}+\frac12\s{\Delta}\tr\ubar\chi-\s{d}\tr\ubar\chi\cdot\zeta
+\s{\nabla}\cdot G_{\ubar L}+\s{\Delta}\kappa\\
\ubar L\s{\Delta}\log\tr\ubar\chi &= -2\s{\nabla}\cdot(\hat{\ubar\chi}\cdot\s{d}\log\tr\ubar\chi)
-\frac32\tr\ubar\chi\s{\Delta}\log\tr\ubar\chi-\frac12\tr\ubar\chi|\s{d}\log\tr\ubar\chi|^2
-\s{\Delta}\frac{|\hat{\ubar\chi}|^2+G(\ubar L,\ubar L)}{\tr\ubar\chi}+\s{\Delta}\kappa
\end{align*}
\end{corollary}
\begin{proof} When used we will assume $V,W\in E(\Sigma)$.\\
For any 1-form $\eta$ on $\Omega$ we have
\begin{align*}
\mathcal{L}_{\ubar L}(\s{\nabla}\eta)(V,W)&=\ubar L(\s{\nabla}_V\eta(W)) = V\ubar L\eta(W)-\ubar L\eta(\s{\nabla}_VW)\\
&=V(\mathcal{L}_{\ubar L}\eta)(W)-(\mathcal{L}_{\ubar L}\eta)(\s{\nabla}_VW)-\eta([\ubar L,\s{\nabla}_VW])\\
&=\s{\nabla}_V(\mathcal{L}_{\ubar L}\eta)(W)-\eta([\ubar L,\s{\nabla}_VW])
\end{align*}
from which we find
\begin{align*}
\ubar L(\s{\nabla}\cdot\eta) &= C_1^1C_2^2(\mathcal{L}_{\ubar L}\gamma^{-1}\otimes\s{\nabla}\eta+\gamma^{-1}\otimes\mathcal{L}_{\ubar L}(\s{\nabla}\eta))\\
&=-2\ubar\chi\cdot\s{\nabla}\eta+\tr(\mathcal{L}_{\ubar L}\s{\nabla}\eta)\\
&=-2(\hat{\ubar\chi}+\frac12\tr\ubar\chi\gamma)\cdot\s{\nabla}\eta+\s{\nabla}\cdot(\mathcal{L}_{\ubar L}\eta)-\eta(2\vec{\s{\nabla}\cdot\hat{\ubar\chi}})
\end{align*}
the last term coming from Lemma \ref{l8} after taking a trace over $V,W$. We conclude that
$$\ubar L(\s{\nabla}\cdot\eta)=-\tr\ubar\chi\s{\nabla}\cdot\eta-2\s{\nabla}\cdot(\hat{\ubar\chi}\cdot\eta)+\s{\nabla}\cdot(\mathcal{L}_{\ubar L}\eta).$$
The first part of the corollary now straight forwardly follows from Proposition \ref{p5} for $\eta = \zeta$. For the second, since $\s{\Delta}\log\tr\ubar\chi = \s{\nabla}\cdot\s{d}\log\tr\ubar\chi$ we have
$$\ubar L\s{\Delta}\log\tr\ubar\chi=-\tr\ubar\chi\s{\Delta}\log\tr\ubar\chi-2\s{\nabla}\cdot(\hat{\ubar\chi}\cdot\s{d}\log\tr\ubar\chi)+\s{\nabla}\cdot(\mathcal{L}_{\ubar L}\s{d}\log\tr\ubar\chi).$$
From the fact that
$$\s{\nabla}\cdot(\mathcal{L}_{\ubar L}\s{d}\log\tr\ubar\chi) = \s{\nabla}\cdot(\s{d}\ubar L\log\tr\ubar\chi) = \s{\Delta}\Big(-\frac12\tr\ubar\chi-\frac{|\hat{\ubar\chi}|^2+G(\ubar L,\ubar L)}{\tr\ubar\chi}+\kappa\Big)$$
the result follows as soon as we make the substitution
$$\s{\Delta}\tr\ubar\chi = \tr\ubar\chi\Big(\s{\Delta}\log\tr\ubar\chi+|\s{d}\log\tr\ubar\chi|^2\Big)$$
\end{proof}
\begin{theorem}[Propagation of $\rho$]\label{t3}
Assuming $\{\Sigma_s\}$ is expanding along the flow vector $\ubar L=\mathfrak{s}L^-$ we conclude that
\begin{align*}
\dot{\rho}+\frac32\mathfrak{s}\rho &= \frac{\mathfrak{s}}{2}\Big(\frac12\langle\vec{H},\vec{H}\rangle\Big(|\hat{\chi}^-|^2+G(L^-,L^-)\Big)+|\tau|^2-\frac12G(L^-,L^+)\Big)\\
&\quad+\s\Delta\Big(\mathfrak{s}(|\hat{\chi}^-|^2+G(L^-,L^-))\Big)-2\s\nabla\cdot(\mathfrak{s}\hat{\chi}^-\cdot\tau)+\s\nabla\cdot(\mathfrak{s} G_{L^-})
\end{align*}
\end{theorem}
\begin{proof} From Proposition \ref{p5} and Corollary \ref{c2} the proof reduces to an exercise in algebraic manipulation
\begin{align*}
\mathcal{L}_{\ubar L}\rho&=\mathcal{L}_{\ubar L}K - \frac14tr\chi\mathcal{L}_{\ubar L}\tr\ubar\chi-\frac14\tr\ubar\chi\mathcal{L}_{\ubar L}tr\chi+\ubar L\s\nabla\cdot\zeta-\ubar L\s\Delta\log \tr\ubar\chi\\
&=\Big(\s\nabla\cdot\s\nabla\cdot\hat{\ubar\chi}-\frac12\s\Delta \tr\ubar\chi-\tr\ubar\chi K\Big)-\frac14tr\chi\Big(-\frac12tr^2\ubar\chi-|\hat{\ubar\chi}|^2-G(\ubar L,\ubar L)+\kappa \tr\ubar\chi\Big)\\
&\quad-\frac14\tr\ubar\chi\Big(G(\ubar L,L)+2K-2\s\nabla\cdot\zeta-2|\zeta|^2-\langle\vec{H},\vec{H}\rangle-\kappa tr\chi\Big)\\
&\quad-2\s\nabla\cdot(\hat{\ubar\chi}\cdot\zeta)-2\tr\ubar\chi\s\nabla\cdot\zeta-\s\nabla\cdot\s\nabla\cdot\hat{\ubar\chi}
+\frac12\s\Delta \tr\ubar\chi-\s{d}\tr\ubar\chi\cdot\zeta+\s\nabla\cdot G_{\ubar L}+\s\Delta\kappa\\
&\quad+2\s\nabla\cdot(\hat{\ubar\chi}\cdot\s{d}\log \tr\ubar\chi)+\frac32\tr\ubar\chi\s\Delta\log \tr\ubar\chi+\frac12\tr\ubar\chi|\s{d}\log \tr\ubar\chi|^2+\s\Delta\frac{|\hat{\ubar\chi}|^2+G(\ubar L,\ubar L)}{\tr\ubar\chi}-\s\Delta\kappa
\end{align*}
\begin{align*}
&=-\frac32\tr\ubar\chi K+\frac18\tr\ubar\chi\langle\vec{H},\vec{H}\rangle+\frac14\langle\vec{H},\vec{H}\rangle\Big(\frac{|\hat{\ubar\chi}|^2+G(\ubar L,\ubar L)}{\tr\ubar\chi}\Big)-\frac14\tr\ubar\chi G(\ubar L,L) - \frac32\tr\ubar\chi\s\nabla\cdot\zeta\\
&\quad+\frac14\tr\ubar\chi\langle\vec{H},\vec{H}\rangle-2\s\nabla\cdot(\hat{\ubar\chi}\cdot(\zeta-\s{d}\log \tr\ubar\chi))+\frac32\tr\ubar\chi\s\Delta\log \tr\ubar\chi\\
&\quad+\frac12\tr\ubar\chi|\zeta|^2-\s{d}\tr\ubar\chi\cdot\zeta+\frac12\tr\ubar\chi|\s{d}\log \tr\ubar\chi|^2\\
&=-\frac32\tr\ubar\chi\rho+\frac14\langle\vec{H},\vec{H}\rangle\Big(\frac{|\hat{\ubar\chi}|^2+G(\ubar L,\ubar L)}{\tr\ubar\chi}\Big)+\frac12\tr\ubar\chi|\zeta-\s{d}\log \tr\ubar\chi|^2-\frac14\tr\ubar\chi G(\ubar L,L)\\
&\quad+\s\Delta\frac{|\hat{\ubar\chi}|^2+G(\ubar L,\ubar L)}{\tr\ubar\chi}-2\s\nabla\cdot(\hat{\ubar\chi}\cdot(\zeta-\s{d}\log \tr\ubar\chi))+\s\nabla\cdot G_{\ubar L}.
\end{align*}
The result therefore follows as soon as we express all terms according to the inflation basis $\{L^-,L^+\}$ where $\{\Sigma_s\}$ is a flow along $\ubar L = \mathfrak{s}L^-$ of speed $\mathfrak{s} = \tr\ubar\chi$.
\end{proof}
\begin{corollary}\label{c3} For $\{\Sigma_s\}$ expanding along the flow vector $\ubar L = \mathfrak{s}L^-$ and any $u\in\mathcal{F}(\Sigma_s)$
\begin{align*}
\int_{\Sigma_s}e^u\Big(\dot{\rho}+\frac32\mathfrak{s}\rho\Big)dA&=\int_{\Sigma_s} \mathfrak{s} e^u\Big(\Big(|\hat{\chi}^-|^2+G(L^-,L^-)\Big)\Big(\frac14\langle\vec{H},\vec{H}\rangle+\s\Delta u\Big)\\
&\indent+\frac12|2\hat{\chi}^-\cdot\s{d}u+\tau|^2+G(L^-, |\s\nabla u|^2L^--\s\nabla u-\frac14L^+)\Big)dA
\end{align*}
\end{corollary}
\begin{proof}
We start by integrating by parts on the last three terms of Theorem \ref{t3}
\begin{align*}
\int e^u&\Big(\s\Delta(\mathfrak{s}(|\hat{\chi}^-|^2+G(L^-,L^-)))-2\s\nabla\cdot(\mathfrak{s}\hat{\chi}^-\cdot\tau)+\s\nabla\cdot (\mathfrak{s} G_{L^-})\Big)dA\\
&=\int \mathfrak{s}e^u\Big(e^{-u}(\s\Delta e^u)(|\hat{\chi}^-|^2+G(L^-,L^-))+2\hat{\chi}^-(\s\nabla u,\vec{\tau})-G(L^-,\s\nabla u)\Big)dA\\
&=\int \mathfrak{s}e^{u}\Big((\s\Delta u+|\s\nabla u|^2)(|\hat{\chi}^-|^2+G(L^-,L^-))+2\hat{\chi}^-(\s\nabla u,\vec{\tau})-G(L^-,\s\nabla u)\Big)dA\\
&=\int \mathfrak{s}e^u\Big((|\hat{\chi}^-|^2+G(L^-,L^-))\s{\Delta}u+|\hat{\chi}^-|^2|\s{\nabla}u|^2+2\hat{\ubar\chi}(\s\nabla u,\vec{\tau})+G(L^-,|\s{\nabla}u|^2L^--\s\nabla u)\Big)dA.
\end{align*}
As a result
\begin{align*}
\int e^u\Big(\dot{\rho}+\frac32\mathfrak{s}\rho\Big)dA &=\int \mathfrak{s}e^u\Big((|\hat{\chi}^-|^2+G(L^-,L^-))\Big(\frac14\langle\vec{H},\vec{H}\rangle+\s{\Delta}u\Big)\\
&\quad+|\hat{\chi}^-|^2|\s{\nabla}u|^2+2\hat{\chi}^-(\s{\nabla}u,\vec{\tau})+\frac12|\tau|^2+G(L^-,|\s{\nabla}u|^2L^--\s{\nabla}u-\frac14 L^+)\Big)dA.
\end{align*}
Since $\hat{\chi}^-$ is symmetric and trace-free it follows that $|\hat{\chi}^-\cdot\s{d}u|^2=\frac12|\hat{\chi}^-|^2|\s{\nabla}u|^2$ from which the first three terms of the second line simplifies to give
$$|\hat{\chi}^-|^2|\s{\nabla}u|^2+2\hat{\chi}^-(\s{\nabla}u,\vec{\tau})+\frac12|\tau|^2=\frac12|2\hat{\chi}^-\cdot\s{d}u+\tau|^2$$
\end{proof}
\begin{remark}\label{r3}
An interesting consequence of the above corollary in spacetimes satisfying the null energy condition is the fact that any $u\in\mathcal{F}(\Sigma)$ gives
$$\int e^{u}\Big(\dot{\rho}+\frac32\mathfrak{s}\rho\Big)dA\geq\int \mathfrak{s}e^u(|\hat{\chi}^-|^2+G(L^-,L^-))\Big(\frac14\langle\vec{H},\vec{H}\rangle+\s{\Delta}u\Big)dA$$
\end{remark}
The proof of Theorem \ref{t1} is a simple consequence of the following corollary:
\begin{corollary}\label{c4} Assuming $\{\Sigma_s\}$ is expanding along the flow vector $\ubar L=\mathfrak{s}L^-$ with each $\Sigma_s$ of non-zero flux ($|\rho(s)|>0$) then
\begin{align*}
\frac{d}{ds}\int_{\Sigma_s}\rho^{\frac23}dA = \int_{\Sigma_s}\dot{(\rho^{\frac23})}+\mathfrak{s}\rho^{\frac23}dA&=\frac23\int_\Sigma\frac{\mathfrak{s}}{\rho^{\frac13}}\Big(\Big(|\hat{\chi}^-|^2+G(L^-,L^-)\Big)\Big(\frac14\langle\vec{H},\vec{H}\rangle-\frac{1}{3}\s\Delta \log|\rho|\Big)\\
&\indent+\frac12|\frac23\hat{\chi}^-\cdot\s{d}\log|\rho|-\tau|^2\\
&\indent+G(L^-, \frac19|\s\nabla\log|\rho||^2L^-+\frac13\s\nabla\log|\rho|-\frac14L^+)\Big)dA
\end{align*}
\end{corollary}
\begin{proof}
From the first variation of Area formula
$$\dot{dA} = -\langle\vec{H},\ubar L\rangle dA=-\mathfrak{s}\langle\vec{H},L^-\rangle dA=\mathfrak{s}dA$$
we get the first equality. For the second we apply Corollary \ref{c3} with $e^{u} = \frac23|\rho|^{-\frac13}$, canceling the sign in the case that $\rho<0$.
\end{proof}
\subsection{Case of Equality}
\begin{lemma}\label{l11}
For $\{\Sigma_s\}$ expanding along $\ubar L = \mathfrak{s}L^-$ we have
$$\mathcal{L}_{\ubar L}\tau+\mathfrak{s}\tau+\s\nabla\cdot(\mathfrak{s}\hat{\chi}^-) = \mathfrak{s}G_{L^-}+\s{d}(\mathfrak{s}(|\hat{\chi}^-|^2+G(L^-,L^-)))$$
\end{lemma}
\begin{proof}
By combining (8) and (11):
\begin{align*}
\mathcal{L}_{\ubar L}(\zeta-\s{d}\log\tr\ubar\chi) &= G_{\ubar L}-\s\nabla\cdot\hat{\ubar\chi}-\tr\ubar\chi\zeta+\frac12\s{d}\tr\ubar\chi+\s{d}\kappa
-\s{d}\Big(-\frac12\tr\ubar\chi-\frac{|\hat{\ubar\chi}|^2+G(\ubar L,\ubar L)}{\tr\ubar\chi}+\kappa\Big)\\
&=-\tr\ubar\chi(\zeta-\s{d}\log\tr\ubar\chi)-\s\nabla\cdot\hat{\ubar\chi}+G_{\ubar L}+\s{d}\frac{|\hat{\ubar\chi}|^2+G(\ubar L,\ubar L)}{\tr\ubar\chi}.
\end{align*}
The result follows as soon as we switch to the inflation basis $\{L^-,L^+\}$.
\end{proof}
\begin{theorem}\label{t4}
Let $\Omega$ be a null hypersurface in a spacetime satisfying the null energy condition with vector field $\ubar L$ tangent to the null generators of $\Omega$. Suppose $\{\Sigma_s\}$ is an expanding (SP)-foliation defined as the level sets of a function $s:\Omega\to\mathbb{R}$ satisfying $\ubar L(s) = 1$ and achieves the case of equality $\frac{dm}{ds}=0$. Then all foliations achieve equality, moreover, we find an affine level set function $r\in\mathcal{F}(\Omega)$ with $r_0:=r|_{\Sigma_{s_0}}\circ \pi$ such that any surface $\Sigma:=\{r=\omega\circ\pi\}$, for $\omega\in\mathcal{F}(\Sigma_{s_0})$, has data:
\begin{align*}
\gamma &= \omega^2\gamma_0\\
\ubar\chi&=\omega\gamma_0\\
\tr\ubar\chi&=\frac{2}{\omega}\\
\tr\chi&=\frac{2}{\omega}(\mathcal{K}_0-\frac{r_0}{\omega}-\omega^2\s\Delta\log\omega)\\
\zeta&=-\s{d}\log\omega\\
\rho&=\frac{r_0}{\omega^3}
\end{align*}
where $r_0^2\gamma_0$ is the metric on $\Sigma_{s_0}$ and $\mathcal{K}_0$ the Gaussian curvature associated to $\gamma_0$.\\
In the case that $\tr\chi|_{\Sigma_{s_0}}=0$ our data corresponds with the the standard null cone in Schwarzschild spacetime of mass $M=\frac{r_0}{2}$.
\end{theorem}
\begin{proof}
Without loss of generality we assume $s_0=0$. Immediately from Corollary \ref{c4} we conclude for this particular foliation that
\begin{align*}
|\hat{\chi}^-|^2+G(L^-,L^-) &= 0\\
|\frac23\hat{\chi}^-\cdot\s{d}\log\rho-\tau|^2&=0\\
G(L^-,\frac19|\s\nabla\log\rho|^2L^-+\frac13\s\nabla\log\rho-\frac14L^+) &=0.
\end{align*}
So from the first equality we have both $\hat{\chi}^- = 0$ and $G(L^-,L^-)=0$. Combined with the second equality we conclude that $\tau = 0$ for this particular foliation and therefore Lemma \ref{l11} ensures that $G_{L^-} = 0$ as well. Finally we may therefore utilize the final equality to conclude also that $G(L^+,L^-)$ = 0 so that, for any $p\in\Omega$ and any $X\in T_pM$, we have
$$G(L^-,X) = 0.$$
From this and Lemma \ref{l11} we have for any foliation off of $\Sigma_0$ generated by some $\ubar L_a$ ($a>0$) that
$$\mathcal{L}_{\ubar L_a}\tau^a+a\mathfrak{s}\tau^a=0.$$
Given that $\tau^a|_{\Sigma_0}=\tau|_{\Sigma_0}=0$ this enforces $\tau^a=0$ by standard uniqueness theorems.\\
We recognise this implies the case of equality for all foliations so without loss of generality we assume that $\ubar L$ is geodesic.
We are now in a position to show that the flux $\rho\in\mathcal{F}(\Omega)$ is independent of the foliation from which it is constructed. In particular, for any $a>0$, foliating off of $\Sigma_0$ along the generator $\ubar L_a$ will construct a $\rho_a$ which we would like to show agrees pointwise on $\Omega$ with $\rho$.\\
From Theorem \ref{t3} we have
$$\ubar L\rho=-\frac32\tr\ubar\chi\rho = 3\rho\ubar L\log\tr\ubar\chi$$
so for any $p\in\Omega$ solving this ODE along the geodesic $\gamma_{\pi(p)}^{\ubar L}(s)$ gives
$$\frac{\rho\circ s(p)}{\rho(0)} = \Big(\frac{\tr\ubar\chi(p)}{\tr\ubar\chi(0)}\Big)^3.$$
For the generator $\ubar L_a$ Theorem \ref{t3} gives
\begin{align*}
\ubar L_a\rho_a = -\frac32\tr\ubar\chi_a\rho_a &= 3\rho_a(\ubar L_a\log\tr\ubar\chi_a-\kappa_a)\\
&=3\rho_a\ubar L_a(\log\tr\ubar\chi_a-\log a) \\
&=3\rho_a\ubar L_a(\log\tr\ubar\chi)
\end{align*}
where the penultimate line comes from the fact that $\kappa_a\ubar L_a = D_{\ubar L_a}\ubar L_a = a\ubar L(a)\ubar L = \ubar L_a(\log a)\ubar L_a$ and the final line from the fact that $\tr\ubar\chi_a = a\tr\ubar\chi$. Solving this ODE along the pregeodesic $\gamma_{\pi(p)}^{\ubar L_a}(t)$ we have
$$\frac{\rho_a\circ t(p)}{\rho_a(0)}=\Big(\frac{\tr\ubar\chi(p)}{\tr\ubar\chi(0)}\Big)^3 = \frac{\rho\circ s(p)}{\rho(0)}.$$
Since we're foliating off of $\Sigma_0$ in both cases and $\rho|_{\Sigma_0}$ is independent of our choice of null basis we have $\rho(p) = \rho_a(p)$ as desired.\\
We therefore define the functions $r_0$ and $r$ according to
$$\frac{1}{r_0^2} = \rho|_{\Sigma_{s_0}},\,\,\,\frac{r_0\circ\pi}{r^3} = \rho$$
(i.e. $r|_{\Sigma_0} = r_0$) so that Theorem \ref{t3} gives $-3\frac{r_0}{r^4}\ubar L_a(r)=\ubar L_a(\rho) =-\frac32\tr\ubar\chi_a\rho =-\frac32\tr\ubar\chi_a \frac{r_0}{r^3}$ and therefore $\ubar L_a(r) = \frac12\tr\ubar\chi_a r$. It follows that if we scale $\ubar L$ such that $\tr\ubar\chi|_{\Sigma_0} = \frac{2}{r_0}$ then $\ubar L(\tr\ubar\chi r) = -\frac12(\tr\ubar\chi)^2 r+\tr\ubar\chi(\frac12\tr\ubar\chi r) = 0$ implies that $\tr\ubar\chi = \frac{2}{r}$ and $\ubar L(r) = 1$. So $r$ is in fact our level set function. For $r_0^2\gamma_0$ the metric on $\Sigma_0$, by Lie dragging $\gamma_0$ along $\ubar L$ to all of $\Omega$ we have
$$\mathcal{L}_{\ubar L}(r^2\gamma_0) = 2r\gamma_0 = \frac{2}{r}(r^2\gamma_0) = \tr\ubar\chi (r^2\gamma_0).$$
So from (6), $\mathcal{L}_{\ubar L}(r^2\gamma_0-\gamma) = \tr\ubar\chi(r^2\gamma_0-\gamma)$ and $r_0^2\gamma_0- \gamma(r_0)=0$ giving $\gamma(r) = r^2\gamma_0$ by uniqueness. We conclude that for any $0\leq\omega\in\mathcal{F}(\Sigma_0)$ the cross-section $\Sigma:=\{r=\omega\circ \pi\}$ has metric $\gamma_\omega =\gamma(r)|_{\Sigma}= \omega^2\gamma_0$ with Gaussian curvature $\mathcal{K}_\omega = \frac{1}{\omega^2}\mathcal{K}_0-\s\Delta\log\omega$. Moreover, 
\begin{align*}
\frac{r_0}{\omega^3} &= \rho_\omega = \mathcal{K}_\omega-\frac14\langle\vec{H},\vec{H}\rangle\\
&=\frac{1}{\omega^2}\mathcal{K}_0-\s\Delta\log\omega-\frac{1}{2\omega}\tr\chi_\omega
\end{align*}
having used the fact that $\rho_\omega=\rho|_{\Sigma}$ (from independence of foliation) in the first line and $\tr\ubar\chi_\omega = \tr\ubar\chi|_{\Sigma}$ in the second. We conclude that,
$$\tr\chi_\omega = \frac{2}{\omega}(\mathcal{K}_0-\frac{r_0}{\omega}-\omega^2\s\Delta\log\omega).$$
In the case that $\tr\chi|_{\Sigma_0} = 0$ property (SP) forces $\frac{1}{r_0^2}=\rho|_{\Sigma_0}$ to be constant by way of the maximum principle. From our expression for $\tr\chi_{r_0}$ we conclude that $\mathcal{K}_0=1$ and therefore $\gamma_0$ is a round metric on $\mathbb{S}^2$. 
\end{proof}
\begin{remark}\label{r4}
We bring to the attention of the reader that due the lack of information regarding the term $\hat{G}$ in (9) we are unable to conclude with any knowledge of the datum $\chi$ on $\Sigma$. In the case of vacuum this no longer poses a problem and one is able to correlate $\chi|_{\Sigma}$ with $\chi|_{\Sigma_{r_0}}$ as shown by Sauter (\cite{S}, Lemma 4.3).
\end{remark}
\newpage
\section{Foliation Comparison}
In this section we show how the flux function $\rho$ of an arbitrary cross section of $\Omega$ decomposes in terms of the flux of the background foliation. With the appropriate asymptotic decay on $\Omega$ this allows us to prove Theorem \ref{t2}.
\subsection{Additional Setup}
We follow once again the construction of \cite{MS1} starting with a background foliation as constructed in Section 3 off of an initial cross-section $\Sigma_{s_0}$. As before, each $\Sigma_s$ allows a null basis $\{\ubar L,l\}$  such that $\langle \ubar L,l\rangle=2$. Also from section 3 we have the diffeomorphism $p\mapsto (\pi(p),s(p))$ of $\Omega$ onto its image. Therefore any cross-section with associated embedding $\Phi:\mathbb{S}^2\to\Omega$ is equivalently realized with the map $\tilde{\Phi}=(\pi,s)\circ\Phi$. Expressing the component functions $\Psi:=\pi\circ\Phi$ and $\omega:=s\circ\Phi$ we recognize that $\Psi:\mathbb{S}^2\to\Sigma_{s_0}$ is a diffeomorphism and therefore the embedding $\Phi:\mathbb{S}^2\to\Omega$ is uniquely characterized as a graph over $\Sigma_{s_0}$ with graph function $\omega\circ\Psi^{-1}$. 
\begin{wrapfigure}{r}{7.5cm}
\centering
\def\svgwidth{300pt} 
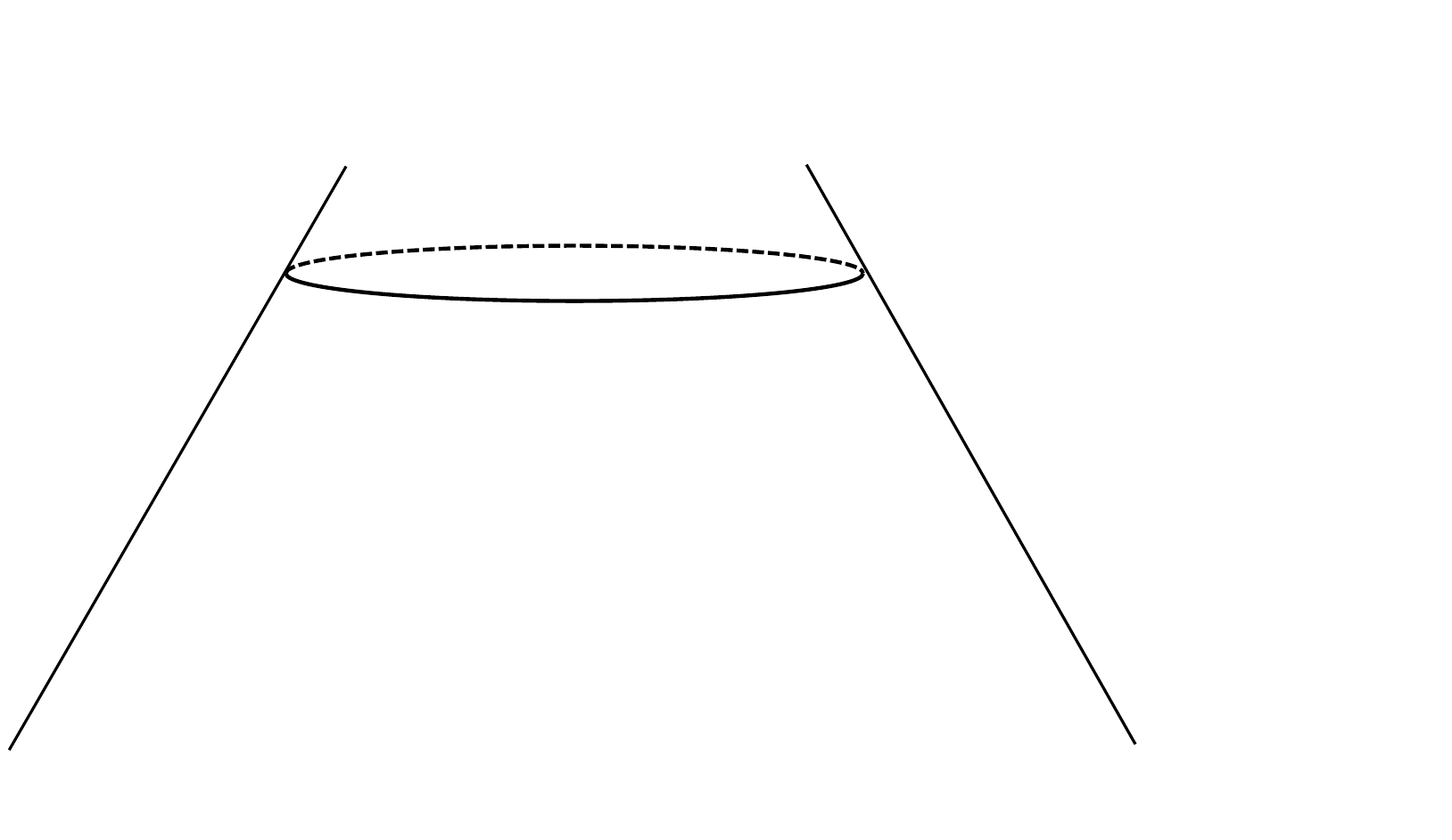
\end{wrapfigure}
Without confusion we will simply denote the graph function by $\omega$ and it's associated cross section by $\Sigma_\omega$. We wish to compare both the intrinsic and extrinsic geometry of $\Sigma_\omega$ at a point $q$ with the geometry of the surface $\Sigma_{s(q)}$. We extend $\omega$ to all of $\Omega$ in the usual way by imposing it be constant along generators of $\ubar L$, in other words, $\omega(p):=(\omega\circ\pi)(p)$. For the extrinsic geometry of $\Sigma_\omega$ we have the null-normal basis $\{\ubar L,L\}$ whereby $L$ is given by the conditions $\langle \ubar L,L\rangle = 2$ and $\langle V,L\rangle=0$ for any $V\in\Gamma(T\Sigma_\omega)$. As before $\Sigma_\omega$ has second fundamental form decomposing into the null components $\ubar\chi$ (associated to $\ubar L$) and $\chi$ (associated to $L$) with torsion $\zeta$. For each $\Sigma_s$ we equivalently decompose the second fundamental form into the components $K$ (associated to $\ubar L$) and $Q$ (associated to $l$) with torsion $t$. We will denote the induced covariant derivative on $\Sigma_s$ by $\nabla$ and on $\Sigma_\omega$ by $\s{\nabla}$. The following lemma is known (\cite{MS1},\cite{S}):
\begin{lemma}\label{l12}
Given $q\in \Sigma_\omega\cap\Sigma_{s(q)}$ the map given by
\begin{align*}
T_\omega:T_q\Sigma_{s(q)}&\to T_q\Sigma_\omega\\
v&\to \tilde v:=v+v\omega \ubar L
\end{align*}
is a well defined isomorphism with natural extension $E(\Sigma_{s_0})\to E(\Sigma_\omega)$.
 Moreover,
\begin{itemize}
\item$\gamma_\omega(\tilde V,\tilde W) = \gamma_s(V,W)$
\item$\ubar\chi(\tilde V,\tilde W) = K(V,W)$
\item$\zeta(\tilde V) = t(V)-K(V,\nabla\omega)+\kappa\langle V,\nabla\omega\rangle$
\item$\begin{aligned}[t]
\chi(\tilde V,\tilde W) &= Q(V,W)-2t(V)\langle W,\nabla\omega\rangle-2t(W)\langle V,\nabla\omega\rangle-|\nabla\omega|^2K(V,W)-2H^\omega(V,W)\\
&\quad+2K(V,\nabla\omega)\langle W,\nabla\omega\rangle+2K(W,\nabla\omega)\langle V,\nabla\omega\rangle-2\kappa\langle V,\nabla\omega\rangle\langle W,\nabla\omega\rangle
\end{aligned}$
\item$\tr\chi=\tr Q-4t(\nabla\omega)-2(\Delta\omega-2\hat{K}(\nabla\omega,\nabla\omega))+\tr K|\nabla\omega|^2-2\kappa|\nabla\omega|^2$
\end{itemize}
for $H^\omega$ the Hessian of $\omega$ on $\Sigma_s$.
\end{lemma}
\begin{proof}
For completeness we include a similar proof as in \cite{MS1} (Proposition 1). Since $T_\omega:T_q\Sigma_{s(q)}\to T_q\Sigma_\omega$ is clearly injective it suffices to show $\tilde v\in T_q\Sigma_\omega$. This follows from the fact that $\tilde v(s-\omega) = v(s-\omega)+v\omega \ubar L(s-\omega) = -v\omega+v\omega = 0$ since $\Sigma_\omega$ is locally characterised by $s|_{\Sigma_\omega}=\omega$. For the extension $\tilde V= V+V\omega\ubar L$ we note that $[\ubar L,V]=0\implies[\ubar L,\tilde V]=0$ and it follows that $\tilde V\in E(\Sigma_\omega)$ (in fact $\tilde{E}(\Sigma_{s_0}) = E(\Sigma_\omega)$). From this and the fact that $D_{\ubar L}\ubar L = \kappa \ubar L$ the first two identities follow straight forwardly. For the third identity we find that $L=l-|\nabla\omega|^2\ubar L-2\nabla\omega$ since
\begin{align*}
\langle L,\ubar L\rangle &= \langle l,\ubar L\rangle = 2\\
\langle L,\tilde V\rangle&=\langle l,V\omega \ubar L\rangle-2\langle\nabla\omega,V\rangle = 2V\omega-2V\omega = 0
\end{align*}
giving
\begin{align*}
\zeta(\tilde V) &= \frac12\langle D_{V+V\omega \ubar L}\ubar L,l-|\nabla\omega|^2\ubar L-2\nabla\omega\rangle\\
&=\frac12\langle D_V\ubar L+\kappa V\omega\ubar L,l-|\nabla\omega|^2\ubar L-2\nabla\omega\rangle\\
&=t(V)-\frac14|\nabla\omega|^2V\langle \ubar L,\ubar L\rangle-\langle D_V\ubar L,\nabla\omega\rangle+\kappa\langle V,\nabla\omega\rangle\\
&=t(V)-K(V,\nabla\omega)+\kappa\langle V,\nabla\omega\rangle.
\end{align*}
For comparison between $\chi$ and $Q$ we calculate $\chi(\tilde V,\tilde W) = \langle D_{V+V\omega \ubar L}(l-|\nabla\omega|^2\ubar L-2\nabla\omega),W+W\omega \ubar L\rangle$ in three parts
\begin{align*}
\langle D_{V+V\omega \ubar L}l,W+W\omega \ubar L\rangle &= Q(V,W)+V\omega\langle D_{\ubar L} l,W\rangle+W\omega\langle D_Vl,\ubar L \rangle+V\omega W\omega\langle D_{\ubar L}l,\ubar L\rangle\\
&=Q(V,W)-V\omega\langle l,D_{\ubar L}W\rangle-W\omega\langle l,D_V\ubar L\rangle-V\omega W\omega\langle l,D_{\ubar L}\ubar L\rangle\\
&=Q(V,W)-2V\omega t(W)-2W\omega t(V)-2\kappa V\omega W\omega\\
-|\nabla\omega|^2\langle D_{V+V\omega \ubar L}\ubar L,W+W\omega \ubar L\rangle &= -|\nabla\omega|^2\langle D_V\ubar L,W+W\omega \ubar L\rangle\\
&=-|\nabla\omega|^2K(V,W)-\frac12|\nabla\omega|^2W\omega V\langle \ubar L,\ubar L\rangle\\
&=-|\nabla\omega|^2K(V,W)\\
-2\langle D_{V+V\omega \ubar L}\nabla\omega,W+W\omega \ubar L\rangle&=-2\langle D_V\nabla\omega,W\rangle-2W\omega\langle D_V\nabla\omega,\ubar L\rangle-2V\omega \langle D_{\ubar L}\nabla\omega,W\rangle\\
&\quad\,-2V\omega W\omega \langle D_{\ubar L}\nabla\omega,\ubar L\rangle\\
&=-2H^\omega(V,W)+2W\omega K(\nabla\omega,V)-2V\omega \ubar LW\omega+2V\omega K(\nabla\omega,W)\\
&\quad\,+2V\omega W\omega \langle \nabla\omega,D_{\ubar L}\ubar L\rangle\\
&=-2H^\omega(V,W)+2W\omega K(\nabla\omega,V)+2V\omega K(\nabla\omega,W)
\end{align*}
the third to last line coming from $\langle D_{\ubar L}\nabla\omega,W\rangle = \ubar L\langle \nabla\omega, W\rangle-\langle \nabla\omega, D_{\ubar L}W\rangle=\ubar LW\omega-K(W,\nabla\omega)$. Collecting all the terms the result follows. The final identity follows upon taking a trace.
\end{proof}
We are now ready to prove our first main result of this section. On $\Sigma_\omega$ we will denote the flux function (1) by $\s\rho$ and on $\Sigma_s$ by $\rho$ the following theorem provides comparison between the two
\begin{theorem}[Flux Comparison Theorem]\label{t5}
At any $q\in\Sigma_\omega\cap\Sigma_s$ we have 
\begin{align*}
\s\rho=\rho&+\s\nabla\cdot\Big(\frac{|\hat{K}|^2+G(\ubar L,\ubar L)}{\tr K}\s\nabla\omega\Big)+\frac12\Big(|\hat{K}|^2+G(\ubar L,\ubar L)\Big)|\nabla\omega|^2\\
&+\nabla\omega\frac{|\hat{K}|^2+G(\ubar L,\ubar L)}{\tr K}+G({\ubar L},\nabla\omega)-2\hat{K}(\vec{t}-\nabla\log\tr K,\nabla\omega)
\end{align*}
\end{theorem}
\begin{remark}\label{r5}
Revisiting Theorem \ref{t4} and the case that $\hat{\ubar\chi}=G(\ubar L,\cdot) = 0$, Theorem \ref{t5} provides an alternative proof that $\s\rho$ agrees with $\rho$ point wise.
\end{remark}
\begin{proof} When used, we assume $V,W,U\in E(\Sigma_{s_0}) (\implies \tilde V,\tilde W,\tilde U\in E(\Sigma_\omega)$).
We will need to know how to relate the covariant derivatives between the two surfaces so first a lemma
\begin{lemma}\label{l13} 
$T_{\omega}\Big(\nabla_VW+V\omega \vec{K}(W)+W\omega\vec{K}(V)-K(V,W)\nabla\omega\Big)=\s{\nabla}_{\tilde V}\tilde W$
\end{lemma} 
\begin{proof}
Since $\s\nabla_{\tilde V}\tilde W|_q = (S+S\omega\ubar L)|_q=T_\omega(S|_q)$ for some $S\in\Gamma(T\Sigma_{s(q)})$ it follows that $\langle \s\nabla_{\tilde V}\tilde W,U\rangle = \langle S,U\rangle$ for any $U\in E(\Sigma_{s_0})$. We find
\begin{align*}
\langle \s\nabla_{\tilde V}\tilde W, U\rangle &= \langle D_{\tilde V}\tilde W+\frac12\ubar\chi(\tilde V,\tilde W)L+\frac12\chi(\tilde V,\tilde W)\ubar L,U\rangle\\
&=\langle D_{\tilde V}\tilde W,U\rangle + \frac12K(V,W)\langle L,U\rangle\\
&= \tilde V\langle W,U\rangle - \langle \tilde W, D_{\tilde V}U\rangle + \frac12K(V,W)\langle l - |\nabla\omega|^2\ubar L - 2\nabla\omega,U\rangle\\
&=(V+V\omega \ubar L)\langle W,U\rangle - \langle W+W\omega \ubar L,D_{V+V\omega \ubar L}U\rangle-K(V,W)U\omega\\
&=V\langle W,U\rangle +2V\omega K(W,U) - \Big(\langle W,\nabla_VU\rangle + V\omega K(W,U) -W\omega K(V,U)\Big) - K(V,W)U\omega\\
&=\Big(V\langle W,U\rangle - \langle W,\nabla_VU\rangle\Big)+K(W,U)V\omega+K(V,U)W\omega-K(V,W)U\omega\\
&=\langle \nabla_VW+V\omega \vec{K}(W)+W\omega\vec{K}(V)-K(V,W)\nabla\omega,U\rangle
\end{align*} 
so $S = \nabla_VW+V\omega \vec{K}(W)+W\omega\vec{K}(V)-K(V,W)\nabla\omega$ since $E(\Sigma_{s_0})|_{\Sigma_{s(q)}} = \Gamma(T\Sigma_{s(q)})$.
\end{proof}
Now we proceed with the proof of Theorem \ref{t5} in 3 parts:\\\\
\underline{Part 1} \textit{Comparison between $\s{\nabla}\cdot\zeta$ and $\nabla\cdot t$:}\\\\
From Lemmas \ref{l12} and \ref{l13} we have
\begin{align*}
(\s\nabla_{\tilde V}\zeta)(\tilde W)&= \tilde V(\zeta (\tilde W)) - \zeta(\s\nabla_{\tilde V}\tilde W)\\ 
&= (V+V\omega \ubar L)\Big(t(W) - K(W,\nabla\omega)+\kappa\langle W,\nabla\omega\rangle\Big)\\
&\indent-t\Big(\nabla_VW+V\omega\vec{K}(W)+W\omega\vec{K}(V)-K(V,W)\nabla\omega\Big)\\
&\indent+K\Big(\nabla_VW+V\omega\vec{K}(W)+W\omega\vec{K}(V)-K(V,W)\nabla\omega,\nabla\omega\Big)\\
&\indent-\kappa\langle\nabla_VW+V\omega\vec{K}(W)+W\omega\vec{K}(V)-K(V,W)\nabla\omega,\nabla\omega\rangle .
\end{align*}
Isolating the terms of the second line we get
\begin{align*}
(V+V\omega \ubar L)(t(W)-K(W,\nabla\omega)&+\kappa W\omega)\\
&=Vt(W)+V\omega\Big(G_{\ubar L}(W)-\nabla\cdot\hat{K}(W)-\tr Kt(W)+\frac12W\tr K+W\kappa\Big)\\
&\indent-VK(W,\nabla\omega) - V\omega (\mathcal{L}_{\ubar L}K)(W,\nabla\omega)-V\omega K(W,[\ubar L,\nabla\omega])\\
&\indent+V\kappa W\omega+\kappa VW\omega+V\omega\ubar L\kappa W\omega
\end{align*}
where (11) was used to give the first line. To continue we'll need an expression for $[\ubar L,\nabla\omega]$ and use (6) to get it:
\begin{align*}
2K(\nabla\omega,V) &= (\mathcal{L}_{\ubar L}\gamma_s)(\nabla\omega,V) = \ubar L\langle\nabla\omega,V\rangle-\langle[\ubar L,\nabla\omega],V\rangle\\
&=\ubar LV\omega-\langle[\ubar L,\nabla\omega],V\rangle\\
&=-\langle[\ubar L,\nabla\omega],V\rangle
\end{align*}
since $[\ubar L,\nabla\omega]\in\Gamma(T\Sigma_s)$ we conclude that $[\ubar L,\nabla\omega]=-2\vec{K}(\nabla\omega)$. Substitution back into our calculation and using (7) (in the form $\mathcal{L}_{\ubar L}K(V,W) = -\ubar\alpha(V,W)+\langle\vec{K}(V),\vec{K}(W)\rangle+\kappa K(V,W)$) gives
\begin{align*}
(V&+V\omega \ubar L)(t(W)-K(W,\nabla\omega)+\kappa W\omega)\\
&=Vt(W)-VK(W,\nabla\omega)\\
&\indent+V\omega\Big(G_{\ubar L}(W)-\nabla\cdot\hat{K}(W)-\tr Kt(W)+\frac12W\tr K+\ubar\alpha(W,\nabla\omega)+\langle \vec{K}(W),\vec{K}(\nabla\omega)\rangle\Big)\\
&\indent+V\omega W\kappa-\kappa V\omega K(W,\nabla\omega)+V\kappa W\omega+\kappa VW\omega+\ubar L\kappa V\omega W\omega.
\end{align*}
Collecting terms we get
\begin{align*}
(\s\nabla_{\tilde V}\zeta)(\tilde W)&=Vt(W)-t(\nabla_VW) +K(\nabla_VW,\nabla\omega)-VK(W,\nabla\omega)\\
&\indent+V\omega\Big(G_{\ubar L}(W)-\nabla\cdot\hat{K}(W)-\tr Kt(W)+\frac12W\tr K+\ubar\alpha(W,\nabla\omega)+\langle\vec{K}(W),\vec{K}(\nabla\omega)\rangle\Big)\\
&\indent-V\omega K(W,\vec{t})-W\omega K(V,\vec{t})+K(V,W)t(\nabla\omega)\\
&\indent+V\omega\langle\vec{K}(W),\vec{K}(\nabla\omega)\rangle+W\omega\langle\vec{K}(V),\vec{K}(\nabla\omega)\rangle-K(V,W)K(\nabla\omega,\nabla\omega)\\
&\indent+V\omega W\kappa-\kappa V\omega K(W,\nabla\omega)+V\kappa W\omega+\kappa VW\omega+\ubar L\kappa V\omega W\omega\\
&\indent-\kappa\nabla_VW\omega-\kappa V\omega K(W,\nabla\omega)-\kappa W\omega K(V,\nabla\omega)+\kappa K(V,W)|\nabla\omega|^2.
\end{align*}
So taking a trace over $V$ and $W$ 
\begin{align*}
\s\nabla\cdot\zeta &=\nabla\cdot t-\nabla\cdot(\vec{K}(\nabla\omega))\\
&\indent+\Big(G_{\ubar L}(\nabla\omega)-(\nabla\cdot{\hat{K}})(\nabla\omega)-\tr Kt(\nabla\omega)+\frac12\nabla\omega\tr K+\ubar\alpha(\nabla\omega,\nabla\omega)+|\vec{K}(\nabla\omega)|^2\Big)\\
&\indent-2K(\nabla\omega,\vec{t})+\tr K t(\nabla\omega)+2|\vec{K}(\nabla\omega)|^2-\tr KK(\nabla\omega,\nabla\omega)\\
&\indent+2\nabla\omega\kappa-3\kappa K(\nabla\omega,\nabla\omega)+\kappa\Delta\omega+\ubar L\kappa|\nabla\omega|^2+\kappa\tr K|\nabla\omega|^2\\
&=\nabla\cdot t-\Big(\nabla\cdot(\vec{K}(\nabla\omega))+(\nabla\cdot\hat{K})(\nabla\omega)-\frac12\nabla\omega \tr K\Big)-2\Big(K(\nabla\omega,\vec{t})-\frac12\tr Kt(\nabla\omega)\Big)\\
&\indent+3|\vec{K}(\nabla\omega)|^2-\tr KK(\nabla\omega,\nabla\omega)+G_{\ubar L}(\nabla\omega)-\tr Kt(\nabla\omega)+\ubar\alpha(\nabla\omega,\nabla\omega)\\
&\indent+2\nabla\omega\kappa-3\kappa\hat{K}(\nabla\omega,\nabla\omega)+\kappa\Delta\omega+\ubar L\kappa|\nabla\omega|^2-\frac12\kappa\tr K|\nabla\omega|^2\\
&=\nabla\cdot t-\Big(2(\nabla\cdot\hat{K})(\nabla\omega) + H^{\omega}\cdot K\Big)-2\hat{K}(\nabla\omega,\vec{t})+3|\vec{K}(\nabla\omega)|^2-\tr KK(\nabla\omega,\nabla\omega)\\
&\indent+G_{\ubar L}(\nabla\omega)-\tr Kt(\nabla\omega)+\ubar\alpha(\nabla\omega,\nabla\omega)\\
&\indent+2\nabla\omega\kappa-3\kappa\hat{K}(\nabla\omega,\nabla\omega)+\kappa\Delta\omega+\ubar L\kappa|\nabla\omega|^2-\frac12\kappa\tr K|\nabla\omega|^2\\
&=\nabla\cdot t-2(\nabla\cdot\hat{K})(\nabla\omega)-H^{\omega}\cdot\hat{K}-\frac12\tr K\Delta\omega-2\hat{K}(\nabla\omega,\vec{t})+\frac32|\hat{K}|^2|\nabla\omega|^2\\
&\indent+2\tr K\hat{K}(\nabla\omega,\nabla\omega)+\frac14(\tr K)^2|\nabla\omega|^2+G_{\ubar L}(\nabla\omega)-\tr Kt(\nabla\omega)+\ubar\alpha(\nabla\omega,\nabla\omega)\\
&\indent+2\nabla\omega\kappa-3\kappa\hat{K}(\nabla\omega,\nabla\omega)+\kappa\Delta\omega+\ubar L\kappa|\nabla\omega|^2-\frac12\kappa\tr K|\nabla\omega|^2.
\end{align*}
\underline{Step 2} \textit{Comparison between $\s{\nabla}\cdot\zeta-\s{\Delta}\log\tr\ubar\chi$ and $\nabla\cdot t-\Delta\log\tr K$}:\\\\
Since $\tr\ubar\chi=\tr K|_{\Sigma_\omega}$ we start by comparing $\s{\Delta}\log\tr K$ with $\Delta\log\tr K$
$$H^{\log \tr\ubar\chi}(\tilde V,\tilde W) =\langle \s\nabla_{\tilde V}\s\nabla\log\tr K,\tilde W\rangle = \tilde V\tilde W\log\tr K - \s\nabla_{\tilde V}\tilde W\log\tr K$$
So isolating the first term we get
\begin{align*}
\tilde V\tilde W\log\tr K &=(V+V\omega \ubar L)(W+W\omega \ubar L)\log\tr K\\
&=VW\log\tr K+(VW\omega+ V\omega W+W\omega V)\ubar L\log\tr K+V\omega W\omega \ubar L\ubar L\log\tr K
\end{align*} 
and then the second
\begin{align*}
\s\nabla_{\tilde V}\tilde W\log\tr K&=(\nabla_VW+V\omega\vec{K}(W)+W\omega\vec{K}(V)-K(V,W)\nabla\omega)\log\tr K\\
&\quad+(\nabla_VW+V\omega\vec{K}(W)+W\omega\vec{K}(V)-K(V,W)\nabla\omega)\omega \ubar L\log\tr K
\end{align*}
having used Lemma \ref{l13}. Collecting terms
\begin{align*}
H^{\log\tr K}(\tilde V,\tilde W)&=VW\log\tr K - \nabla_VW\log\tr K\\
&\indent+(VW\omega - \nabla_VW\omega)\ubar L\log\tr K+V\omega W\omega \ubar L\ubar L\log\tr K\\
&\indent-\Big(V\omega K(W,\nabla\log\tr K)+W\omega K(V,\nabla\log\tr K)-K(V,W)\langle\nabla\omega,\nabla\log\tr K\rangle\Big)\\
&\indent+\Big(K(V,W)|\nabla\omega|^2-V\omega K(W,\nabla\omega)-W\omega K(V,\nabla\omega)+V\omega W+W\omega V\Big)\ubar L\log\tr K.
\end{align*}
So that a trace over $V$ and $W$ yields
\begin{align*}
\s\Delta\log\tr K&=\Delta\log\tr K+\Delta\omega \ubar L\log\tr K+|\nabla\omega|^2\ubar L\ubar L\log\tr K-2\hat{K}(\nabla\omega,\nabla\log\tr K)\\
&\quad-2\hat{K}(\nabla\omega,\nabla\omega)\ubar L\log\tr K+2\nabla\omega \ubar L\log\tr K.
\end{align*}
We take the opportunity at this point of the calculation to bring to the attention of the reader that we have not yet used any distinguishing characteristics of the function $\log\tr K$ in comparison to an arbitrary $f\in\mathcal{F}(\Omega)$. In particular, we notice if $f\in\mathcal{F}(\Omega)$ satisfies $\ubar Lf=0$ switching with $\log\tr K$ above yields the fact $\s\Delta f = \Delta f-2\hat{K}(\nabla\omega,\nabla f)$. As a result,
\begin{lemma}\label{l14}
$$\s\Delta g = \Delta g+\s\nabla\cdot(\ubar L g\s\nabla\omega)+\nabla\omega\ubar L g-2\hat{K}(\nabla\omega,\nabla g)$$
for any $g\in\mathcal{F}(\Omega)$.
\end{lemma}
\begin{proof}
We have 
\begin{align*}
\s\Delta g &= \Delta g+\Delta\omega\ubar Lg+|\nabla\omega|^2\ubar L\ubar Lg-2\hat{K}(\nabla\omega,\nabla g)-2\hat{K}(\nabla\omega,\nabla\omega)\ubar L g+2\nabla\omega\ubar Lg\\
&=\Delta g+(\Delta\omega-2\hat{K}(\nabla\omega,\nabla\omega))\ubar Lg+(\nabla\omega+|\nabla\omega|^2\ubar L)\ubar Lg+\nabla\omega\ubar Lg-2\hat{K}(\nabla\omega,\nabla g)\\
&=\Delta g+\s\Delta\omega\ubar Lg+\s\nabla\omega\ubar Lg+\nabla\omega\ubar Lg-2\hat{K}(\nabla\omega,\nabla g)\\
&=\Delta g+\s\nabla\cdot(\ubar Lg\s\nabla\omega)+\nabla\omega\ubar Lg-2\hat{K}(\nabla\omega,\nabla g)
\end{align*}
having used the fact that $\ubar L\omega = 0$ and the comment immediately preceding the statement of Lemma \ref{l14} to get the third equality.
\end{proof}
Finishing up Step 2 we have
\begin{align*}
\s\nabla\cdot\zeta&-\s\Delta\log\tr\ubar\chi=\nabla\cdot t-\Delta\log\tr K\\
&-2(\nabla\cdot\hat{K})(\nabla\omega)-H^{\omega}\cdot\hat{K}-2\hat{K}(\nabla\omega,\vec{t}-\nabla\log\tr K)-\tr Kt(\nabla\omega)+\nabla\omega\tr K+\frac32|\hat{K}|^2|\nabla\omega|^2\\
&+G_{\ubar L}(\nabla\omega)+\ubar\alpha(\nabla\omega,\nabla\omega)\\
&-\Big(\frac12\tr K\Delta\omega + \Delta\omega \ubar L\log\tr K\Big)+\Big(\frac14(\tr K)^2-\ubar L\ubar L\log\tr K\Big)|\nabla\omega|^2-\Big(\nabla\omega\tr K+2\nabla\omega \ubar L\log\tr K\Big)\\
&+2\hat{K}(\nabla\omega,\nabla\omega)\Big(\tr K+\ubar L\log\tr K\Big)\\
&+2\nabla\omega\kappa-3\kappa\hat{K}(\nabla\omega,\nabla\omega)+\kappa\Delta\omega+\ubar L\kappa|\nabla\omega|^2-\frac12\kappa\tr K|\nabla\omega|^2\\
=\nabla\cdot t&-\Delta\log\tr K\\
&-2(\nabla\cdot\hat{K})(\nabla\omega)-H^{\omega}\cdot\hat{K}-2\hat{K}(\nabla\omega,\vec{t}-\nabla\log\tr K)-\tr Kt(\nabla\omega)+\nabla\omega\tr K+|\hat{K}|^2|\nabla\omega|^2\\
&+G_{\ubar L}(\nabla\omega)+\hat{\ubar\alpha}(\nabla\omega,\nabla\omega)+\frac12\Big(|\hat{K}|^2+G(\ubar L,\ubar L)\Big)|\nabla\omega|^2\\
&+\Big(\Delta\omega-2\hat{K}(\nabla\omega,\nabla\omega)\Big)\frac{|\hat{K}|^2+G(\ubar L,\ubar L)}{\tr K}+\Big(-\frac12(|\hat{K}|^2+G(\ubar L,\ubar L)-\kappa\tr K)+\ubar L\frac{|\hat{K}|^2+G(\ubar L,\ubar L)}{\tr K}\Big)|\nabla\omega|^2\\
&+2\nabla\omega\frac{|\hat{K}|^2+G(\ubar L,\ubar L)}{\tr K}+\tr K\hat{K}(\nabla\omega,\nabla\omega)-\kappa\hat{K}(\nabla\omega,\nabla\omega)-\frac12\kappa\tr K|\nabla\omega|^2\\
=\nabla\cdot t&-\Delta\log\tr K\\
&-2(\nabla\cdot\hat{K})(\nabla\omega)-H^{\omega}\cdot\hat{K}-2\hat{K}(\nabla\omega,\vec{t}-\nabla\log\tr K)-\tr Kt(\nabla\omega)+\nabla\omega\tr K+|\hat{K}|^2|\nabla\omega|^2\\
&+G_{\ubar L}(\nabla\omega)+\hat{\ubar\alpha}(\nabla\omega,\nabla\omega)+\s\Delta\omega\frac{|\hat{K}|^2+G(\ubar L,\ubar L)}{\tr K}+\ubar L\frac{|\hat{K}|^2+G(\ubar L,\ubar L)}{\tr K}|\nabla\omega|^2\\
&+2\nabla\omega\frac{|\hat{K}|^2+G(\ubar L,\ubar L)}{\tr K}+\tr K\hat{K}(\nabla\omega,\nabla\omega)-\kappa\hat{K}(\nabla\omega,\nabla\omega)
\end{align*}
having used (8) to get the last two lines in the second equality, Lemma \ref{l14} to get $\Delta\omega-2\hat{K}(\nabla\omega,\nabla\omega) = \s\Delta\omega$ in the second equality followed by cancellation of the terms $\frac12\Big(|\hat{K}|^2+G(\ubar L,\ubar L)\Big)|\nabla\omega|^2$ and $\frac12\kappa\tr K|\nabla\omega|^2$. \\\\
\underline{Step 3} \textit{Comparison between $\s\rho$ and $\rho$}:\\\\
Denoting the Gauss curvature on $\Sigma_s$ by $\mathcal{C}$ and the mean curvature vector $\vec{h}$ we have from the Gauss equation (3) 
\begin{align*}
\mathcal{K}&-\frac14\langle\vec{H},\vec{H}\rangle+\frac12\hat{\ubar\chi}\cdot\hat{\chi}=-\frac12R-G(\ubar L,L)-\frac14\langle R_{\ubar L L}\ubar L,L\rangle\\
&=-\frac12R-G(\ubar L,l-|\nabla\omega|^2\ubar L-2\nabla\omega)-\frac14\langle R_{\ubar L\,l-|\nabla\omega|^2\ubar L-2\nabla\omega}\ubar L,l-|\nabla\omega|^2\ubar L-2\nabla\omega\rangle\\
&=\mathcal{C}-\frac14\langle\vec{h},\vec{h}\rangle+\frac12\hat{K}\cdot\hat{Q}+|\nabla\omega|^2G(\ubar L,\ubar L)+2G(\ubar L,\nabla\omega)-\langle R_{\ubar L\,\nabla\omega}l,\ubar L\rangle-\langle R_{\ubar L\,\nabla\omega}\ubar L,\nabla\omega\rangle\\
&=\mathcal{C}-\frac14\langle\vec{h},\vec{h}\rangle+\frac12\hat{K}\cdot\hat{Q}+\frac12|\nabla\omega|^2G(\ubar L,\ubar L)+\Big(2G(\ubar L,\nabla\omega)-\langle R_{\ubar L\,\nabla\omega}l,\ubar L\rangle\Big)-\hat{\ubar\alpha}(\nabla\omega,\nabla\omega)
\end{align*}
from this we conclude
\begin{align*}
\Big(\mathcal{K}&-\frac14\langle\vec{H},\vec{H}\rangle+\s\nabla\cdot\zeta-\s\Delta\log \tr\ubar\chi\Big)-\Big(\mathcal{C}-\frac14\langle\vec{h},\vec{h}\rangle+\nabla\cdot t-\Delta\log\tr K\Big)\\
&=\frac12\Big(\hat{K}\cdot\hat{Q}-\hat{\ubar\chi}\cdot\hat{\chi}\Big)\\
&\quad+\frac12|\nabla\omega|^2G(\ubar L,\ubar L)+\Big(2G(\ubar L,\nabla\omega)-\langle R_{\ubar L\,\nabla\omega}l,\ubar L\rangle\Big)-\hat{\ubar\alpha}(\nabla\omega,\nabla\omega)\\
&\quad-2(\nabla\cdot\hat{K})(\nabla\omega)-H^{\omega}\cdot\hat{K}-2\hat{K}(\nabla\omega,\vec{t}-\nabla\log\tr K)-\tr Kt(\nabla\omega)+\nabla\omega\tr K+|\hat{K}|^2|\nabla\omega|^2\\
&\quad+G_{\ubar L}(\nabla\omega)+\hat{\ubar\alpha}(\nabla\omega,\nabla\omega)\\
&\quad+\s\Delta\omega\frac{|\hat{K}|^2+G(\ubar L,\ubar L)}{\tr K}+\ubar L\frac{|\hat{K}|^2+G(\ubar L,\ubar L)}{\tr K}|\nabla\omega|^2+2\nabla\omega\frac{|\hat{K}|^2+G(\ubar L,\ubar L)}{\tr K}+\tr K\hat{K}(\nabla\omega,\nabla\omega)\\
&\quad-\kappa \hat{K}(\nabla\omega,\nabla\omega)
\end{align*}
Isolating the first two terms and using Lemma \ref{l12} we get
\begin{align*}
\hat{K}\cdot\hat{Q}&-\hat{\ubar\chi}\cdot\hat{\chi}\\
&=\hat{K}\cdot\hat{Q}-\Big(\hat{K}\cdot\hat{Q}-|\nabla\omega|^2|\hat{K}|^2-4\hat{K}(\nabla\omega,\vec{t})+2|\hat{K}|^2|\nabla\omega|^2\\
&\,\,\quad\indent\indent\indent\indent\indent\indent\indent\indent+2\tr K\hat{K}(\nabla\omega,\nabla\omega)-2\hat{K}\cdot H^{\omega}-2\kappa\hat{K}(\nabla\omega,\nabla\omega)\Big)\\
&=-|\hat{K}|^2|\nabla\omega|^2-2\tr K\hat{K}(\nabla\omega,\nabla\omega)+2\hat{K}\cdot H^{\omega}+4\hat{K}(\nabla\omega,\vec{t})+2\kappa \hat{K}(\nabla\omega,\nabla\omega)
\end{align*}
and finally we have
\begin{align*}
\Big(\mathcal{K}&-\frac14\langle\vec{H},\vec{H}\rangle+\s\nabla\cdot\zeta-\s\Delta\log \tr\ubar\chi\Big)-\Big(\mathcal{C}-\frac14\langle\vec{h},\vec{h}\rangle+\nabla\cdot t-\Delta\log\tr K\Big)\\
&=\frac12|\nabla\omega|^2\Big(|\hat{K}|^2+G(\ubar L,\ubar L)\Big)+G_{\ubar L}(\nabla\omega)-2\hat{K}(\nabla\omega,\vec{t}-\nabla\log\tr K)\\
&\quad+\Big(2G_{\ubar L}(\nabla\omega)-\langle R_{\ubar L\,\nabla\omega}l,\ubar L\rangle-2(\nabla\cdot\hat{K})(\nabla\omega)+2\hat{K}(\nabla\omega,\vec{t})-\tr Kt(\nabla\omega)+\nabla\omega\tr K\Big)\\
&\quad+\s\Delta\omega\frac{|\hat{K}|^2+G(\ubar L,\ubar L)}{\tr K}+\ubar L\frac{|\hat{K}|^2+G(\ubar L,\ubar L)}{\tr K}|\nabla\omega|^2+2\nabla\omega\frac{|\hat{K}|^2+G(\ubar L,\ubar L)}{\tr K}.
\end{align*}
Amazingly the third line vanishes by the Codazzi equation (4) as well as all terms with a factor $\kappa$ giving
\begin{align*}
\s\rho-\rho&=\frac12\Big(|\hat{K}|^2+G(\ubar L,\ubar L)\Big)|\nabla\omega|^2+G_{\ubar L}(\nabla\omega)-2\hat{K}(\nabla\omega,\vec{t}-\nabla\log\tr K)\\
&\quad+\s\Delta\omega\frac{|\hat{K}|^2+G(\ubar L,\ubar L)}{\tr K}+\ubar L\frac{|\hat{K}|^2+G(\ubar L,\ubar L)}{\tr K}|\nabla\omega|^2+2\nabla\omega\frac{|\hat{K}|^2+G(\ubar L,\ubar L)}{\tr K}
\end{align*}
and the result then follows from the fact that $\s\nabla\omega=\nabla\omega+|\nabla\omega|^2\ubar L$ as well as
 $$\s\nabla\cdot\Big(\frac{|\hat{K}|^2+G(\ubar L,\ubar L)}{\tr K}\s\nabla\omega\Big) = \s\Delta\omega\frac{|\hat{K}|^2+G(\ubar L,\ubar L)}{\tr K}+\s\nabla\omega\frac{|\hat{K}|^2+G(\ubar L,\ubar L)}{\tr K}.$$
\end{proof}
\subsection{Asymptotic flatness}
In this section we wish to study the limiting behaviour of our mass functional in the setting of asymptotic flatness constructed by Mars and Soria \cite{MS1}. Beyond the assumption that we have a cross section $\Sigma_{s_0}$ of $\Omega$ we also assume for some (hence any) choice of past-directed geodesic null generator $\ubar L$ (i.e. $D_{\ubar L}\ubar L=0$) that $S_+=\infty$. So all geodesics $\gamma_q^{\ubar L}$ are `past complete' with domain $(s_-(q),\infty)$. We now take $s_0=0$ ignoring all points $p$ satisfying $s(p)\leq S_-$ and conclude that $\Omega\cong \mathbb{S}^2\times(S_-,\infty)$. Although the value of $S_-$ will depend on our choice of geodesic generator $\ubar L$ our interest lies only on the past of $\Sigma_{0}$ (i.e. $\mathbb{S}^2\times (0,\infty)$) so we ignore this subtlety. A null hypersurface $\Omega$ with all the above properties is called \textit{extending to past null infinity}.\\
In order to impose decay conditions of various \textit{transversal tensors} (i.e. tensors satisfying $T(\ubar L,\cdots) =\cdots= T(\cdots,\ubar L)=0$) we choose a local basis on $\Sigma_{0}$ and extend it to a basis field $\{X_i\}\subset E(\Sigma_0)$. Given a transversal k-tensor $T(s)$ we say,
\begin{itemize}
\item $T=O(1)$ iff $T_{i_1...i_k}:=T(X_{i_1},...,X_{i_k})$ is uniformly bounded and $T=O_n(s^{-m})$ iff $$s^{m+j}(\mathcal{L}_{\ubar L})^jT(s)=O(1)\quad(0\leq j\leq n)$$
\item $T=o(s^{-m})$ iff $\displaystyle{\lim_{s\to\infty}s^mT(s)_{i_1...i_k}=0}$ and $T=o_n(s^{-m})$ iff $$s^{m+j}(\mathcal{L}_{\ubar L})^jT(s)=o(1)\quad(0\leq j\leq n)$$
\item $T=o_n^X(s^{-m})$ iff
$$s^m\mathcal{L}_{X_{i_1}}\cdots\mathcal{L}_{X_{i_j}} T(s) = o(1)\quad(0\leq j\leq n).$$
Now we're ready to define asymptotic flatness for $\Omega$ as given by the authors of \cite{MS1}:
\end{itemize}
\begin{definition}\label{d7}
We say $\Omega$ is \underline{past asymptotically flat} if it extends to past null infinity and there exists a choice of cross section $\Sigma_0$ and null geodesic generator $\ubar L$ with corresponding level set function $s$ satisfying the following:\\
\begin{enumerate}
\item There exists two symmetric 2-covariant transversal and $\ubar L$ Lie constant tensor fields $\mathring\gamma$ and $\gamma_1$ such that
$$\tilde\gamma:=\gamma-s^2\mathring\gamma-s\gamma_1 = o_1(s)\cap o_2^X(s)$$
\item There exists a tansversal and $\ubar L$ Lie constant one-form $t_1$ such that
$$\tilde t :=t-\frac{t_1}{s}=o_1(s^{-1})$$
\item There exist $\ubar L$ Lie constant functions $\theta_0$ and $\theta$ such that
$$\tilde\theta:=\tr Q-\frac{\theta_0}{s}-\frac{\theta}{s^2} = o(s^{-2})$$
\item The scalar $\langle R_{X_{i_1} X_{i_2}}X_{i_3},X_{i_4}\rangle$ along $\Omega$ is such that $\displaystyle{\lim_{s\to\infty}}\frac{1}{s^2}\langle  R_{X_{i_1} X_{i_2}}X_{i_3},X_{i_4}\rangle$ exists while its double trace satisfies $-\frac12 R-G(\ubar L,l)-\frac14\langle R_{\ubar L l}\ubar L,l\rangle=o(s^{-2}).$
\end{enumerate}
\end{definition}
We will have the need to supplement the notion of asymptotic flatness of $\Omega$ with a stronger version of the \textit{energy flux decay condition} ($G_{\ubar L}=o(s^{-2}),\,\,\mathcal{L}_{\ubar L}\tilde \gamma=o_1^X(1)$ as given in \cite{MS1}) with the following:
\begin{definition}\label{d8}
Suppose $\Omega$ is past asymptotically flat. We say $\Omega$ has \underline{strong flux decay} if
$$G_{\ubar L}=o(s^{-2}),\,\,\tilde t=o_1^X(s^{-1})\,\,\text{and}\,\,\mathcal{L}_{\ubar L}^j\tilde\gamma=o_{3-j}^X(s^{1-j})\,\,\text{for}\,\,1\leq j\leq 3$$
and \underline{strong decay} if the condition on $G_{\ubar L}$ is dropped.
\end{definition}
We will also need some results from \cite{MS1} (Proposition 3, Lemma 2, Section 4) resulting directly from the asymptotically flat restriction on $\Omega$. One particularly valuable consequence is the ability to choose our geodesic generator $\ubar L$ to give any conformal change on the `metric at null infinity', which turns out to be given by the 2-tensor, $\mathring\gamma$. By the Uniformization Theorem we conclude that this covers all possible metrics on a Riemannian 2-sphere. We will denote the covariant derivative coming from $\mathring\gamma$ by $\mathring\nabla$.
\begin{proposition}\label{p6}
Suppose $\Omega$ is past asymptotically flat with a choice of affinely parametrized null generator $\ubar L$ and corresponding level set function $s$. Letting $\gamma(s)^{ij}$ denote the inverse of $\gamma(s)_{ij}$, 
\begin{align}
\gamma(s)^{ij}&=\frac{1}{s^2}\mathring\gamma^{ij}-\frac{1}{s^3}\mathring{\gamma_1}^{ij}+o(s^{-3})\\
K_{ij}&=s{\mathring\gamma}_{ij}+\frac12{\gamma_1}_{ij}+o(1)\\
\mathcal{K}_{\gamma(s)}&=\frac{\mathring{\mathcal{K}}}{s^2}+o(s^{-2})\\
\tr Q&=\frac{2\mathring{\mathcal{K}}}{s}+\frac{\theta}{s^2}+o(s^{-2})\\
\tr K&=\frac{2}{s}+\frac{\ubar\theta}{s^2}+o(s^{-2})
\end{align}
where $\mathring\gamma^{ij}$ is the inverse of ${\mathring\gamma}_{ij}$, tensors with $\mathring{\text{ring}}$ highlight the fact that indices have been raised with $\mathring\gamma$ and $\ubar\theta = -\frac12\mathring{\tr}\gamma_1$.  \\
It follows in case $\mathcal{L}_{\ubar L}\tilde\gamma = o_1^X(1)$ that
$$t_1 = \frac12\mathring\nabla\cdot \gamma_1+\s{d}\ubar\theta\iff G_{\ubar L} = o(s^{-2})$$
\end{proposition}
\begin{proof}
We refer the reader to \cite{MS1} (Proposition 3) for proof.
\end{proof}
As promised in Remark \ref{r1} we are now able to prove the following well known result:
\begin{lemma}\label{l15}
Suppose $\Omega$ extends to past null infinity with null geodesic generator $\ubar L$. Then any cross section $\Sigma\hookrightarrow\Omega$ satisfies $\tr K\geq 0$. If $\Omega$ is past asymptotically flat then $\Sigma$ is expanding along $\ubar L$.
\end{lemma}
\begin{proof}
For $\omega\in\mathcal{F}(\Omega)$ constructed by Lie dragging $s|_{\Sigma}$ along $\ubar L$ we have $\Sigma = \Sigma_1$ for the geodesic foliation $\{\Sigma_\lambda\}$ given by $s = \omega\lambda$. So it suffices to prove the result along an arbitrary geodesic foliation for $\Omega$. From (8) we have, whenever $\tr K(s_0)< 0$ for some $s_0$, that
$$\ubar L\Big(\frac{1}{\tr K}\Big)=\frac12+|\hat{\chi}^-|^2+G(L^-,L^-)\geq \frac12$$
wherever it be defined as well as
$$\frac{1}{\tr K}(s) \geq \frac{1}{\tr K}(s_0)+ \frac{s-s_0}{2}$$
for any such $s\geq s_0$. So we can find an $s_1>s_0$ such that $\tr K(s)\xrightarrow{s\to s_1^-} -\infty$. Since this contradicts smoothness we must have that $\tr K\geq 0$ on all of $\Omega$. If $\Omega$ is past asymptotically flat it follows from Proposition \ref{p6} that $\tr K(s)>0$ for sufficiently large $s$. Since (8) gives
$$\ubar L(\tr K) = -\frac12(\tr K)^2-|\hat{K}|^2-G(\ubar L,\ubar L)\leq 0$$
we have $\tr K(s_0)\geq \tr K(s_1)$ for all $s_0\leq s_1$. So we must have that $\tr K>0$ on all of $\Omega$.
\end{proof}
\begin{lemma}\label{l16}
On each $\Sigma_s$ the difference tensor 
$$\mathcal{D}(V,W):=\nabla_VW-\mathring\nabla_VW$$
admits the decomposition
$$\mathcal{D}_{ij}^k = \frac12(\mathring{\nabla}_i\mathring{\gamma_1}^k\,_j+\mathring\nabla_j{\gamma_1}^k\,_i-\mathring{\nabla}^k{\gamma_1}_{ij})\frac{1}{s}+O(s^{-2}).$$
Moreover, if $f\in\mathcal{F}(\Omega)$ is Lie constant along $\ubar L$ then
$$\Delta f = \frac{1}{s^2}\mathring\Delta f+(-\mathring\gamma_1^{ij}\mathring\nabla_i\mathring\nabla_jf-(\mathring\nabla_i\mathring\gamma_1^{ij})f,_j+(\mathring\nabla^i\ubar\theta)f,_i)\frac{1}{s^3}+o(s^{-3}).$$
\end{lemma}
\begin{proof}
The result follows from the well known fact (see, for example, \cite{W}) that 
\begin{align*}
\langle\mathcal{D}(V,W),U\rangle &= \frac12(\mathring\nabla_V\gamma(W,U)+\mathring\nabla_W\gamma(V,U)-\mathring\nabla_U\gamma(V,W))\\
&=\frac{s}{2}(\mathring\nabla_V\gamma_1(W,U)+\mathring\nabla_W\gamma_1(V,U)
-\mathring\nabla_U\gamma_1(V,W))\\
&\indent+\frac12(\mathring\nabla_V\tilde\gamma(W,U)+\mathring\nabla_W\tilde\gamma(V,U)
-\mathring\nabla_U\tilde\gamma(V,W)).
\end{align*}
The second is a simple consequence of the first, we refer the reader to \cite{MS1} (Lemma 2) for proof.
\end{proof}
In the next Proposition we show that the decomposition of the metric given in Definition \ref{d7} part 1 allows us to find $\mathcal{K}_{\gamma(s)}$ up to $O(s^{-4})$:
\begin{proposition}\label{p7}
For a decomposition of the metric $\gamma(s) = s^2\mathring\gamma+s\gamma_1+\tilde\gamma$ for some fixed $s$ we have:
\begin{align}
\mathcal{K}_{\gamma(s)}&=\frac{\mathring{\mathcal{K}}}{s^2}+\frac{1}{s^3}(\mathring{\mathcal{K}}\ubar\theta+\frac12\mathring\nabla\cdot\mathring\nabla\cdot\gamma_1+\mathring\Delta\ubar\theta)+O(s^{-4})
\end{align}
\end{proposition}
\begin{proof}
First we take the opportunity to show that $V,W\in E(\Sigma_0)$ gives $\mathring\nabla_VW\in E(\Sigma_0)$. Starting with the Koszul formula 
$$2\mathring\gamma(\mathring\nabla_VW,U) = V\mathring\gamma(W,U)+W\mathring\gamma(U,V)-U\mathring\gamma(V,W) - \mathring\gamma(V,[W,U])+\mathring\gamma(W,[U,V])+\mathring\gamma(U,[V,W])$$
and the fact that $\mathring\gamma$ is Lie constant along $\ubar L$ we conclude that $\ubar L\mathring\gamma(\mathring\nabla_VW,U) = \mathring\gamma([\ubar L,\mathring\nabla_VW],U)$ on the left, applying $\ubar L$ on the right we find everything vanishes since $V,W\in E(\Sigma_0)\implies[V,W]\in E(\Sigma_0)$. Therefore $\mathring\gamma([\ubar L,\mathring\nabla_VW],U) = 0$. Since $[\ubar L,\mathring\nabla_VW]\in \Gamma(T\Sigma_s)$ and $\mathring\gamma$ is positive definite it follows that $[\ubar L,\mathring\nabla_VW]=0$ and therefore $\mathring\nabla_VW\in E(\Sigma_0)$.
To show the decomposition of $\mathcal{K}_{\gamma(s)}$ we start by finding the decomposition of the Riemann curvature tensor on $\Sigma_s$:
\begin{align*}
\langle R^s_{X_iX_j}X_k,X_m\rangle &= \langle \nabla_{[X_i,X_j]}X_k,X_m\rangle - X_i\langle \nabla_{X_j}X_k,X_m\rangle+\langle\nabla_{X_j}X_k,\nabla_{X_i}X_m\rangle+X_j\langle \nabla_{X_i}X_k,X_m\rangle\\
&\indent-\langle\nabla_{X_i}X_k,\nabla_{X_j}X_m\rangle\\
&=\langle\mathring\nabla_{[X_i,X_j]}X_k,X_m\rangle - X_i\langle \mathring\nabla_{X_j}X_k,X_m\rangle+\langle\mathring\nabla_{X_j}X_k,\mathring\nabla_{X_i}X_m\rangle+X_j\langle\mathring\nabla_{X_i}X_k,X_m\rangle\\
&\indent-\langle\mathring\nabla_{X_i}X_k,\mathring\nabla_{X_j}X_m\rangle\\
&\indent+\langle\mathcal{D}([X_i,X_j],X_k),X_m\rangle-X_i\langle\mathcal{D}(X_j,X_k),X_m\rangle+\langle\mathcal{D}(X_j,X_k),\mathring\nabla_{X_i}X_m\rangle\\
&\indent+\langle\mathring\nabla_{X_j}X_k,\mathcal{D}(X_i,X_m)\rangle+X_j\langle\mathcal{D}(X_i,X_k),X_m\rangle-\langle\mathcal{D}(X_i,X_k),\mathring\nabla_{X_j}X_m\rangle\\
&\indent-\langle\mathring\nabla_{X_i}X_k,\mathcal{D}(X_j,X_m)\rangle\\
&\indent+\langle\mathcal{D}(X_j,X_k),\mathcal{D}(X_i,X_m)\rangle-\langle\mathcal{D}(X_i,X_k),\mathcal{D}(X_j,X_m)\rangle.
\end{align*}
Using the decomposition $\gamma_s = s^2\mathring\gamma+O(s)$ we recognize the leading order term, combining lines 3 and 4, is $s^2\mathring\gamma(\mathring{R}_{X_iX_j}X_k,X_m)$. In order to find the next to leading order term the fact that $\langle R^s_{X_iX_j}X_k,X_m\rangle-s^2\mathring\gamma(\mathring{R}_{X_iX_j}X_k,X_m)$ defines a 4-tensor on each $\Sigma_s$ allows us to search independently of our choice of basis $\{X_1,X_2\}$. In particular we may assume that $\mathring\nabla_{X_i}X_j=0$ at $q\in\Sigma_s$ (hence on all of $\gamma_q^{\ubar L}$, since $\mathring\nabla_{X_i}X_j\in E(\Sigma_0)$). So assuming restriction to the generator through $q$ and using Lemma \ref{l16} we have
\begin{align*}
\langle R^s_{X_iX_j}X_k,X_m\rangle&-s^2\mathring\gamma(\mathring{R}_{X_iX_j}X_k,X_m)\\
&=-sX_i\gamma_1(\mathring\nabla_{X_j}X_k,X_m)+sX_j\gamma_1(\mathring\nabla_{X_i}X_k,X_m)\\
&\indent-\frac{s}{2}X_i(\mathring\nabla_{X_j}\gamma_1(X_k,X_m)+\mathring\nabla_{X_k}\gamma_1(X_j,X_m)
-\mathring\nabla_{X_m}\gamma_1(X_j,X_k))\\
&\indent+\frac{s}{2}X_j(\mathring\nabla_{X_i}\gamma_1(X_k,X_m)+\mathring\nabla_{X_k}\gamma_1(X_i,X_m)
-\mathring\nabla_{X_m}\gamma_1(X_i,X_k))\\
&\indent+O(1)\\
&=-sX_i\gamma_1(\mathring\nabla_{X_j}X_k,X_m)+sX_j\gamma_1(\mathring\nabla_{X_i}X_k,X_m)\\
&\quad\frac{s}{2}\Big(\mathring\nabla_{X_j}\mathring\nabla_{X_i}\gamma_1(X_k,X_m)
+\mathring\nabla_{X_j}\mathring\nabla_{X_k}\gamma_1(X_i,X_m)
-\mathring\nabla_{X_j}\mathring\nabla_{X_m}\gamma_1(X_i,X_k)\\
&\quad-\mathring\nabla_{X_i}\mathring\nabla_{X_j}\gamma_1(X_k,X_m)
-\mathring\nabla_{X_i}\mathring\nabla_{X_k}\gamma_1(X_j,X_m)
+\mathring\nabla_{X_i}\mathring\nabla_{X_m}\gamma_1(X_j,X_k)\Big)\\
&\indent+O(1).
\end{align*}
It remains to simplify the two terms of the first line in the second equality. Since
$$X_i\gamma_1(\mathring\nabla_{X_j}X_k,X_m)=
\mathring\nabla_{X_i}\gamma_1(\mathring\nabla_{X_j}X_k,X_m)+
\gamma_1(\mathring\nabla_{X_i}\mathring\nabla_{X_j}X_k,X_m)$$
we conclude that
$$-X_i\gamma_1(\mathring\nabla_{X_j}X_k,X_m)+X_j\gamma_1(\mathring\nabla_{X_i}X_k,X_m)=
\gamma_1(\mathring{R}_{X_iX_j}X_k,X_m).$$
Moreover, it is easily shown using our choice of basis extension that
\begin{align*}
\frac12\mathring\nabla_{X_j}\mathring\nabla_{X_i}\gamma_1(X_k,X_m)&-
\frac12\mathring\nabla_{X_i}\mathring\nabla_{X_j}\gamma_1(X_k,X_m)+
\gamma_1(\mathring{R}_{X_iX_j}X_k,X_m)\\
&=
\frac12(\gamma_1(\mathring{R}_{X_iX_j}X_k,X_m)-\gamma_1(\mathring{R}_{X_iX_j}X_m,X_k)).
\end{align*}
So we finally have from the fact that $\Sigma_s$ is of dimension 2 that
\begin{align*}
\langle R^s_{X_iX_j}X_k,X_m\rangle &= s^2\mathring{\mathcal{K}}(\mathring\gamma_{ik}\mathring\gamma_{jm}-\mathring\gamma_{im}\mathring\gamma_{jk})
+\frac{s}{2}\mathring{\mathcal{K}}(\mathring{\gamma}_{ik}{\gamma_1}_{jm}-\mathring{\gamma}_{im}{\gamma_1}_{jk}+\mathring{\gamma}_{jm}{\gamma_1}_{ik}-\mathring{\gamma}_{jk}{\gamma_1}_{im})\\
&\indent+\frac{s}{2}(\mathring\nabla_j\mathring\nabla_k{\gamma_1}_{im}
-\mathring\nabla_j\mathring\nabla_m{\gamma_1}_{ik}-\mathring\nabla_i\mathring\nabla_k{\gamma_1}_{jm}
+\mathring\nabla_i\mathring\nabla_m{\gamma_1}_{jk})+O(1).
\end{align*}
Using (12) to take a trace over $i,k$:
\begin{align*}
(Ric^s)_{jm}&= \mathring{\mathcal{K}}\mathring{\gamma}_{jm}-\frac{1}{s}\mathring{\mathcal{K}}\ubar\theta\mathring\gamma_{jm}+\frac{1}{2s}(\mathring\nabla_j(\mathring\nabla\cdot\gamma_1)_m+2\mathring\nabla_j\mathring\nabla_m\ubar\theta
-(\mathring\nabla^2\gamma_1)_{jm}+(\mathring\nabla\cdot(\mathring\nabla\gamma_1))_{mj})\\
&\indent+\frac{\mathring{\mathcal{K}}}{s}(2\ubar\theta\mathring\gamma_{jm}+{\gamma_1}_{jm})+O(s^{-2})\\
&=\mathring{\mathcal{K}}\mathring\gamma_{jm}+\frac{1}{s}\Big(\mathring{\mathcal{K}}\ubar\theta\mathring\gamma_{jm}+\mathring{\mathcal{K}}{\gamma_1}_{jm}+\frac12\mathring\nabla_j(\mathring\nabla\cdot\gamma_1)_m
+\frac12(\mathring\nabla\cdot(\mathring\nabla\gamma_1))_{mj}+\mathring\nabla_j\mathring\nabla_m\ubar\theta
-\frac12(\mathring\nabla^2\gamma_1)_{jm}\Big)\\
&\indent+O(s^{-4})
\end{align*}
and then over $j,m$:
$$2\mathcal{K}_{\gamma(s)} = \frac{2}{s^2}\mathring{\mathcal{K}}+\frac{1}{s^3}\Big(2\mathring{\mathcal{K}}\ubar\theta-2\mathring{\mathcal{K}}\ubar\theta
+\mathring\nabla\cdot\mathring\nabla\cdot\gamma_1+2\mathring{\Delta}\ubar\theta\Big)+\frac{2}{s^3}\mathring{\mathcal{K}}\ubar\theta+O(s^{-4})$$
giving the result.
\end{proof}
\begin{remark}\label{r6}
Interestingly, in the case that $\Omega$ is asymptotically flat satisfying the energy flux decay condition we conclude that
$$\mathcal{K}_{\gamma(s)} = \frac{\mathring{\mathcal{K}}}{s^2}+\frac{1}{s^3}(\mathring{\mathcal{K}}\ubar\theta+\mathring\nabla\cdot t_1)+O(s^{-4})$$
according to Proposition \ref{p7}.
\end{remark}
\begin{definition}\label{d9}
For $\Omega$ past asymptotically flat with background geodesic foliation $\{\Sigma_s\}$ we say a foliation $\{\Sigma_{s_\star}\}$ is \underline{asymptotically geodesic} provided
$$s = \phi s_\star+\xi$$
with \underline{scale factor} $\phi>0$ a Lie constant function along $\ubar L$ and $\ubar L^i\xi=o_{2-i}^X(s^{1-i})$ for $0\leq i\leq 2$. In addition (similarly to \cite{MS1}), we will say $\{\Sigma_{s_\star}\}$ \underline{approaches large spheres} provided the class of geodesic foliations measuring $\phi=1$ also induce $\mathring\gamma$ to be the round metric on $\mathbb{S}^2$.
\end{definition}
\begin{remark}\label{r7}
Given a basis extension $\{X_i\}\subset E(\Sigma_{0})$ (on $\{\Sigma_s\}$) and a foliation $\{\Sigma_{s_\star}\}$ as in Definition \ref{d9}, Lie dragging $s|_{\Sigma_{s_\star}}$ along $\ubar L$ to give $\omega\in\mathcal{F}(\Omega)$ we see at $q\in\Sigma_{s_\star}$:
\begin{align*}
\omega_i &= \phi_i s_\star+\xi_s\omega_i+\xi_i\\
\omega_{ij} &= \phi_{ij}s_\star+\xi_{ss}\omega_i\omega_j+\xi_{sj}\omega_i+\xi_s\omega_{ij}+\xi_{ij}
\end{align*}
where $\omega_i:=X_i\omega$, $\omega_{ij}:=X_jX_i\omega$, $\xi_s:=\ubar L\xi$, $\xi_{ss}:=\ubar L\ubar L\xi$, $\xi_i:=X_i(\xi|_{\Sigma_s(q)})$, $\xi_{si} = X_i(\xi_s|_{\Sigma_s})$ and $\xi_{ij} = X_jX_i(\xi|_{\Sigma_s})$. The decay on $\xi$ therefore gives us that:
\begin{align*}
\omega_i&=\frac{\phi_is_\star+\xi_i}{1-\xi_s}=\phi_is_\star+o(s_\star)\\
\omega_{ij}&=\frac{1}{1-\xi_s}\Big(\phi_{ij}s_\star+\xi_{ss}\Big(\frac{\xi_i+\phi_is_\star}{1-\xi_s}\Big)\Big(\frac{\xi_j+\phi_js_\star}{1-\xi_s}\Big)+\xi_{sj}\Big(\frac{\xi_i+\phi_is_\star}{1-\xi_s}\Big)+\xi_{ij}\Big)=\phi_{ij}s_\star+o(s_\star).
\end{align*}
From (12) and Lemma \ref{l16} we conclude that
\begin{align*}
d\omega|_{\Sigma_{s_\star}}&= (s_\star \phi)^2(\frac{-1}{s_\star}d\phi^{-1}|_{\Sigma_{s_\star}}+o(s_\star^{-1}))\\
\Delta\omega|_{\Sigma_{s_\star}} &= \frac{1}{\phi^2s_\star}\mathring{\Delta}\phi+o(s_\star^{-1}).
\end{align*}
\end{remark}
Using Theorem \ref{t5} we prove a slightly weakened version of the beautiful result found by the authors of \cite{MS1} (Theorem 1):
\begin{proposition}\label{p8}
Suppose $\Omega$ is past asymptotically flat and $\{\Sigma_{s_\star}\}$ is an asymptotically geodesic foliation with scale factor $\phi>0$. Assuming $\mathcal{L}_{\ubar L}\tilde\gamma = o_1^X(1)$ we have
$$\lim_{s_\star\to\infty}E_H(\Sigma_{s_\star}) = \frac{1}{16\pi}\sqrt{\frac{\int \phi^2 \mathring{dA}}{4\pi}}\int\frac{1}{\phi}\Big(\mathring{\mathcal{K}}\ubar\theta-\theta-\mathring{\Delta}\ubar\theta+4\mathring{\nabla}\cdot t_1\Big)\mathring{dA}$$
with $\mathring\gamma$, $\mathring{\mathcal{K}}$, $\ubar\theta$, $\theta$ and $t_1$ associated with the background geodesic foliation.
\end{proposition}
\begin{proof}
Given any fixed $s_\star$ we define $\omega\in\mathcal{F}(\Omega)$ by Lie dragging along $\ubar L$:
$$s|_{\Sigma_{s_\star}} = (\phi s_\star+\xi)|_{\Sigma_{s_\star}}$$
as before. From the decomposition $\gamma_s = s^2\mathring\gamma+s\gamma_1+\tilde\gamma$ and the standard identity for any invertible matrix $M$:
$$\det(M+sB) = \det M(1+s\tr(M^{-1}B)+O(s^2))$$
we have 
$$\sqrt{\det(\gamma_s)} = s^2\sqrt{\det(\mathring\gamma)}(1-\frac{1}{s}\ubar\theta+o(s^{-1})).$$
From the first identity of Lemma \ref{l12} we therefore conclude that 
 $$dA_{s_\star} = dA_{s}|_{\Sigma_{s_\star}} = s_\star^2\phi^2f\mathring{dA}$$
 where $f=1+o(s_{\star}^0)$.\\
In Theorem \ref{t5}, denoting the sum of all but the first two terms by $\eta_\omega(\nabla\omega)$ we see
\begin{align*}
\eta_\omega(\nabla\omega) &= \frac12(|\hat{K}|^2+G(\ubar L,\ubar L))|\nabla\omega|^2 +G(\ubar L,\nabla\omega)+\nabla\omega\frac{|\hat{K}|^2+G(\ubar L,\ubar L)}{\tr K}-2\hat{K}(\vec{t}-\nabla\log\tr K,\nabla\omega)\\
&= \frac12(|\hat{K}|^2+G(\ubar L,\ubar L))|\nabla\omega|^2+G(\ubar L,\nabla\omega)-\nabla\omega(\ubar L\log\tr K+\frac12\tr K)-2\hat{K}(\vec{t}-\nabla\log\tr K,\nabla\omega)
\end{align*}
giving from Propositions 6 and 7
\begin{align*}
4\pi\frac{E_H(\Sigma_{s_\star})}{\sqrt{\frac{|\Sigma_{s_\star}|}{16\pi}}}=\int \s\rho dA_\omega &=
\int \rho +\s{\nabla}\cdot\Big(\frac{|\hat{K}|^2+G(\ubar L,\ubar L)}{\tr K}\s{\nabla}\omega\Big)+\eta_\omega(\nabla\omega)dA_{s_\star}\\
&=\int \rho +\eta_\omega(\nabla\omega)dA_{s_\star}\\
&=\int\Big( \frac{\mathring{\mathcal{K}}}{\omega^2}+\frac{1}{\omega^3}\Big(\mathring{\mathcal{K}}\ubar\theta+
\frac12\mathring\nabla\cdot\mathring\nabla\cdot\gamma_1+\mathring\Delta\ubar\theta\Big) - \frac14(\frac{2}{\omega}+\frac{\ubar\theta}{\omega^2})(\frac{2\mathring{\mathcal{K}}}{\omega}+\frac{\theta}{\omega^2})\\
&\indent+\nabla\cdot t+G_{\ubar L}(\nabla\omega)\\
&\indent+\frac12(|\hat{K}|^2+G(\ubar L,\ubar L))|\nabla\omega|^2\\
&\indent-\Delta\log\tr K-\nabla\omega(\ubar L\log\tr K+\frac12\tr K)-2\hat{K}(\vec{t}-\nabla\log\tr K,\nabla\omega)\Big)dA_{s_\star}\\
&\indent+o(s_\star^{-1})
\end{align*}
\begin{align*}
&=\int\Big(\frac{1}{\omega^3}\Big(\frac12\mathring{\mathcal{K}}\ubar\theta-\frac12\theta+
\frac12\mathring\nabla\cdot\mathring\nabla\cdot\gamma_1+\mathring\Delta\ubar\theta\Big)\\
&\indent+\nabla\cdot t+G_{\ubar L}(\nabla\omega)-\frac12\nabla\omega\tr K-2\hat{K}(\vec{t},\nabla\omega)\\
&\indent+\frac12(|\hat{K}|^2+G(\ubar L,\ubar L))|\nabla\omega|^2\\
&\indent-\Delta\log\tr K-\nabla\omega\ubar L\log\tr K+2\hat{K}(\nabla\log\tr K,\nabla\omega)\Big)dA_{s_\star}\\
&\indent+o(s_\star^{-1})
\end{align*}
having used the Divergence Theorem to get the second line. From Lemma \ref{l14} we have
$$-\Delta\log\tr K - \nabla\omega\ubar L\log\tr K +2\hat{K}(\nabla\log\tr K,\nabla\omega) = -\s{\Delta}\log\tr K+\s{\nabla}\cdot(\ubar L\log\tr K\s{\nabla}\omega)$$
and therefore integrates to zero on $\Sigma_{s_\star}$ by the Divergence Theorem. We also notice from the fact that $\ubar L$ is geodesic and $\s\nabla\omega = \nabla\omega+|\nabla\omega|^2\ubar L$ that
$$K(V,\nabla\omega) = K(\tilde V,\s\nabla\omega) = \ubar\chi(\tilde V,\s\nabla\omega)$$
for $V\in E(\Sigma_0)$. Lemma \ref{l12} and \ref{l13} therefore gives
\begin{align*}
\s\nabla_{\tilde V}\zeta(\tilde W) +(\s\nabla_{\tilde V}(\ubar\chi\cdot{\s{d}\omega}))(\tilde W) &= \tilde V(\zeta(\tilde W) + \ubar\chi(\tilde W,\s\nabla\omega))-\zeta(\s\nabla_{\tilde V}\tilde W)-\ubar\chi(\s\nabla_{\tilde V}\tilde W,\s\nabla\omega)\\
&=(V+V\omega\ubar L)t(W)\\
&\indent-\Big(t(\nabla_VW+V\omega\vec{K}(W)+W\omega\vec{K}(V)-K(V,W)\nabla\omega)\\
&\indent-K(\nabla_VW+V\omega\vec{K}(W)+W\omega\vec{K}(V)-K(V,W)\nabla\omega,\nabla\omega)\Big) - \ubar\chi(\s\nabla_{\tilde V}{\tilde W},\s\nabla\omega)\\
&=\nabla_Vt(W)+V\omega\mathcal{L}_{\ubar L}t(W)-V\omega K(\vec{t},\nabla\omega)-W\omega K(\vec{t},\nabla\omega)+K(V,W)t(\nabla\omega)
\end{align*}
where all terms in the penultimate line canceled from Lemma \ref{l13}. Taking a trace over $V,W$
\begin{align*}
\s\nabla\cdot\zeta-\s\nabla\cdot(\vec{\ubar\chi}(\s\nabla\omega)) &= \nabla\cdot t+\mathcal{L}_{\ubar L}t(\nabla\omega) - 2\hat{K}(\vec{t},\nabla\omega)\\
&=\nabla\cdot t+G_{\ubar L}(\nabla\omega)-\frac12\nabla\omega\tr K-2\hat{K}(\vec{t},\nabla\omega) - \nabla\cdot\hat{K}(\nabla\omega)+\nabla\omega\tr K-\tr Kt(\nabla\omega)
\end{align*}
having used (11) to get the last line. We conclude that
$$\int \nabla\cdot t+G_{\ubar L}(\nabla\omega)-\frac12\nabla\omega\tr K-2\hat{K}(\vec{t},\nabla\omega)dA_{s_\star} = \int \nabla\cdot\hat{K}(\nabla\omega)-\nabla\omega\tr K+\tr Kt(\nabla\omega)dA_{s_\star}$$
giving
\begin{align*}
4\pi\frac{E_H(\Sigma_{s_\star})}{\sqrt{\frac{|\Sigma_{s_\star}|}{16\pi}}}
=\int\Big(\frac{1}{\omega^3}\Big(\frac12\mathring{\mathcal{K}}\ubar\theta-\frac12\theta+
\frac12\mathring\nabla\cdot\mathring\nabla\cdot\gamma_1+\mathring\Delta\ubar\theta\Big)
&+\nabla\cdot\hat{K}(\nabla\omega)-\nabla\omega\tr K+\tr Kt(\nabla\omega)\\
&+\frac12(|\hat{K}|^2+G(\ubar L,\ubar L))|\nabla\omega|^2\Big)dA_{s_\star}+o(s_\star^{-1}).
\end{align*}
Now we turn to simplifying the final term in the integrand
$$(|\hat{K}|^2+G(\ubar L,\ubar L))|\nabla\omega|^2 = (\nabla\omega-\s{\nabla}\omega)\tr K-\frac12(\tr K)^2|\nabla\omega|^2.$$
Denoting $g:=\tr K-\frac{2}{s}-\frac{\ubar\theta}{s^2}$ we conclude from the hypothesis $\mathcal{L
}_{\ubar L}\tilde\gamma=o_1^X(1)$ that $g=o_1^X(s^{-2})$. So denoting $g_\omega:=g|_{\Sigma_{s_\star}}$ we have from Remark \ref{r7} (regarding the decay) that
\begin{align*}
\nabla\omega\tr K|_{\Sigma_{s_\star}}&=\nabla\omega g|_{\Sigma_{s_\star}}+\frac{1}{s^2}\nabla\omega\ubar\theta|_{\Sigma_{s_\star}}\\
&=\frac{1}{\omega^2}\s{\nabla}\omega\ubar\theta+o(s_\star^{-3})\\
\s{\nabla}\omega\tr K &=\s{\nabla}\omega g_\omega+\s\nabla\omega(\frac{2}{\omega}+\frac{\ubar\theta}{\omega^2})\\
&=\s\nabla\cdot(g_\omega\s\nabla\omega)-\s{\Delta}\omega g_\omega-\frac{2}{\omega^2}|\s\nabla\omega|^2+\s{\nabla}\omega(\frac{\ubar\theta}{\omega^2})\\
&=\s{\nabla}\cdot(g_\omega\s\nabla\omega)-(\Delta\omega-2\hat{K}(\nabla\omega,\nabla\omega))g_\omega-\frac{2}{\omega^2}|\s\nabla\omega|^2+\s{\nabla}\omega(\frac{\ubar\theta}{\omega^2})\\
&=\s{\nabla}\cdot(g_\omega\s\nabla\omega)-\frac{2}{\omega^2}|\s\nabla\omega|^2+\s{\nabla}\omega(\frac{\ubar\theta}{\omega^2})+o(s_\star^{-3})\\
\frac12(\tr K)^2|\nabla\omega|^2|_{\Sigma_{s_\star}} &= \frac12(\frac{2}{\omega}+\frac{\ubar\theta}{\omega^2})^2|\s\nabla\omega|^2+o(s_\star^{-3})\\
&=\frac{2}{\omega^2}|\s\nabla\omega|^2+\frac{2\ubar\theta}{\omega^3}|\s\nabla\omega|^2+o(s_\star^{-3}).
\end{align*}
Combining terms we conclude
$$(|\hat{K}|^2+G(\ubar L,\ubar L))|\nabla\omega|^2\Big|_{\Sigma_{s_\star}} = -\s{\nabla}\cdot(g_\omega\s\nabla\omega)+o(s_\star^{-3}).$$
It's a simple exercise to show $\hat{K} = -\frac12(\gamma_1+\ubar\theta\mathring\gamma)+o_1^X(1)$, so for $d\omega|_{\Sigma_{s_\star}} = (s_\star^2\phi^2)(-\frac{1}{s_\star}d\phi^{-1}|_{\Sigma_{s_\star}}+o(s_\star^{-1}))$ we have from Lemma \ref{l16}
\begin{align*}
\nabla\cdot\hat{K}(\nabla\omega)|_{\Sigma_{s_\star}}&=\frac{1}{s_\star^3}\frac{1}{2\phi^2}\mathring\nabla\cdot(\gamma_1+\ubar\theta\mathring\gamma)(\mathring\nabla{\phi^{-1}})+o(s_\star^{-3})\\
\nabla\omega\tr K|_{\Sigma_{s_\star}}&=-\frac{1}{s_\star^3}\frac{1}{\phi^2}\mathring\nabla\phi^{-1}\ubar\theta+o(s_\star^{-3})\\
\tr Kt(\nabla\omega)|_{\Sigma_{s_\star}}&=-\frac{1}{s_\star^3}\frac{2}{\phi^2}t_1(\mathring\nabla\phi^{-1})+o(s_\star^{-3}).
\end{align*}
Therefore
\begin{align*}
E_H(\Sigma_{s_\star})&=\frac{1}{8\pi}\sqrt{\frac{\int\phi^2 f\mathring{dA}}{4\pi}}\int\Big(\frac{f}{\phi}\Big(\frac12\mathring{\mathcal{K}}\ubar\theta
+\frac12\mathring\nabla\cdot\mathring\nabla\cdot\gamma_1+\mathring\Delta\ubar\theta-\frac12\theta\Big)
\\
&\indent+\frac{f}{2}\mathring\nabla\cdot(\gamma_1+\ubar\theta\mathring\gamma)(\mathring\nabla\phi^{-1})+f\mathring\nabla\phi^{-1}\ubar\theta-2ft_1(\mathring\nabla\phi^{-1})\Big)\mathring{dA}+o(s_\star^0)
\end{align*}
giving
\begin{align*}
\lim_{s_\star\to\infty}E_H(\Sigma_{s_\star}) &= \frac{1}{8\pi}\sqrt{\frac{\int\phi^2\mathring{dA}}{4\pi}}\int\Big(\frac{1}{\phi}\Big(\frac12\mathring{\mathcal{K}}\ubar\theta
+\frac12\mathring\nabla\cdot\mathring\nabla\cdot\gamma_1+\mathring\Delta\ubar\theta-\frac12\theta\Big)
\\
&\indent+\frac{1}{2}\mathring\nabla\cdot(\gamma_1+\ubar\theta\mathring\gamma)(\mathring\nabla\phi^{-1})+\mathring\nabla\phi^{-1}\ubar\theta
-2t_1(\mathring\nabla\phi^{-1})\Big)\mathring{dA}\\
&=\frac{1}{8\pi}\sqrt{\frac{\int\phi^2\mathring{dA}}{4\pi}}\int\Big(\frac{1}{\phi}\Big(\frac12\mathring{\mathcal{K}}\ubar\theta
+\frac12\mathring\nabla\cdot\mathring\nabla\cdot\gamma_1+\mathring\Delta\ubar\theta
-\frac12\theta\Big)\\
&\indent-\frac{1}{2}\phi^{-1}\mathring\nabla\cdot\mathring\nabla\cdot(\gamma_1+\ubar\theta\mathring\gamma)
-\phi^{-1}\mathring\Delta\ubar\theta+\frac{2}{\phi}\mathring\nabla\cdot t_1\Big)\mathring{dA}\\
\end{align*}
$$=\frac{1}{16\pi}\sqrt{\frac{\int\phi^2\mathring{dA}}{4\pi}}\int\frac{1}{\phi}\Big(\mathring{\mathcal{K}}\ubar\theta-\theta-\mathring\Delta\ubar\theta
+4\mathring\nabla\cdot t_1\Big)\mathring{dA}$$
having integrated by parts to get the second equality.
\end{proof}
\begin{remark}\label{r8}
Suppose $\Omega$ is a past asymptotically flat null hypersurface with a background geodesic foliation $\{\Sigma_s\}$ approaching large spheres (i.e $\mathring\gamma$ is the round metric at infinity). Then for any other geodesic foliation of scale factor $\psi$ it follows that the metric at infinity is $\psi^2\mathring\gamma$ (see \cite{MS1}, Section 4) approaching large spheres if and only if $\psi$ solves the equation 
\begin{equation}
1-\psi^2 = \mathring\Delta\log\psi.
\end{equation}
Proposition \ref{p8} shows all asymptotically geodesic foliations $\{\Sigma_{s_\star}\}$ of the same scale factor $\phi$ share the limit
$$E(\phi) = \lim_{s_\star\to\infty}E_H(\Sigma_{s_\star})$$
which measures a Bondi energy $E_B(\psi)$ if $\psi$ solves (18).
The Bondi mass is therefore given by 
$$m_B = \inf\{E_B(\psi)|1-\psi^2=\mathring\Delta\log\psi\}.$$
\end{remark}
\begin{theorem}\label{t6}
Suppose $\Omega$ is a past asymptotically flat null hypersurface inside a spacetime satisfying the null energy condition. Then given the existence of an asymptotically geodesic (P)-foliation $\{\Sigma_{s_\star}\}$ approaching large spheres we have
$$m(0)\leq E_B$$
for $E_B$ the Bondi energy of $\Omega$ associated to $\{\Sigma_{s_\star}\}$. If equality is achieved on an (SP)-foliation then $E_B = m_B$ the Bondi mass of $\Omega$. In the case that $\tr\chi|_{\Sigma_0}=0$ we conclude instead with the weak Null Penrose inequality
$$\sqrt{\frac{|\Sigma_0|}{16\pi}}\leq E_B$$
where equality along an (SP)-foliation enforces that any foliation of $\Omega$ shares its data ($\gamma$, $\ubar\chi$, $\tr\chi$ and $\zeta$) with some foliation of the standard null cone of Schwarzschild spacetime.
\end{theorem}
\begin{proof}
Since any asymptotically geodesic (P)-foliation has non-decreasing mass from Theorem \ref{t1} and $m(\Sigma_{s_\star})\leq E_H(\Sigma_{s_\star})$ from Lemma \ref{l5}, it follows from \cite{MS1} (Theorem 1) that $m(\Sigma_{s_\star})$ converges since $E_H(\Sigma_{s_\star})$ does. Moreover,  $\lim_{s_\star\to\infty}m(\Sigma_{s_\star})\leq \lim_{s_\star\to\infty}E_H(\Sigma_{s_\star})$ and from \cite{MS1} (Corollary 3) it follows that $\lim_{s_\star\to\infty}E_H(\Sigma_{s_\star})$ is the Bondi energy associated to the abstract reference frame coupled to the foliation $\{\Sigma_{s_\star}\}$. Given the case of equality, Theorem \ref{t1} enforces that $m(0) = m(\Sigma_{s_\star})$ for all $s_\star$. So Theorem \ref{t4} applies and we conclude that $m(\Sigma) = \frac12\Big(\frac{1}{4\pi}\int r_0^{\frac23}dA_0\Big)^{\frac32}$ (for some positive function $r_0$ on $\Sigma_0$ of area form $r_0^2dA_0$) irrespective of the cross-section $\Sigma\subset\Omega$. This gives, according to Remark \ref{r8} and Lemma \ref{l5},
$$\lim_{s_\star\to\infty}m(\Sigma_{s_\star}) = \frac12\Big(\frac{1}{4\pi}\int r_0^{\frac23}dA_0\Big)^{\frac32} =E_B\leq \inf_{\phi>0}E(\phi)\leq m_B.$$
Since $E_B\leq \inf E(\phi)\leq m_B\leq E_B$ all must be equal.

If $\tr\chi|_{\Sigma_0}=0$ property (P) gives
$$0\geq\s\Delta\log\s\rho|_{\Sigma_0}$$
and the maximum principle implies $\s\rho|_{\Sigma_0} = \mathcal{K}+\s\nabla\cdot\tau$ is constant. From the Gauss-Bonnet and Divergence Theorems we conclude that $\s\rho|_{\Sigma_0} = \frac{4\pi}{|\Sigma_0|}$ and therefore $m(0)=\sqrt{\frac{|\Sigma_0|}{16\pi}}$. Under this restriction Theorem \ref{t4} enforces that any foliation of $\Omega$ corresponds with a foliation of the standard null cone in Schwarzschild with respect to the data $\gamma$, $\ubar\chi$, $\tr\chi$ and $\zeta$.
\end{proof}
From Proposition \ref{p8} and Lemma \ref{l5} 
$$\inf_{\phi>0}E(\phi) = \frac14\Big(\frac{1}{4\pi}\int(\mathcal{K}\ubar\theta-\theta-\mathring{\Delta}\ubar\theta+4\mathring\nabla\cdot t_1)^\frac23\mathring{dA}\Big)^\frac32$$
provided $\mathcal{K}\ubar\theta-\theta-\mathring{\Delta}\ubar\theta+4\mathring\nabla\cdot t_1\geq0$. We show, given that $\Omega$ satisfies the strong flux decay condition, this quantity is in fact  $\lim_{s_\star\to\infty}m(\Sigma_{s_\star})$. We will need the following proposition to do so:
\begin{proposition}\label{p9}
Suppose $\Omega$ is past asymptotically flat with strong decay. Given a choice of affinely parametrized null generator $\ubar L$ and corresponding level set function s we have
\begin{align}
\nabla\cdot\nabla\cdot K &= -\frac{1}{2s^4}\mathring\nabla\cdot\mathring\nabla\cdot\gamma_1+o(s^{-4})\\
\Delta\tr K&=\frac{\mathring\Delta \ubar\theta}{s^4}+o(s^{-4})\\
\nabla\cdot t&=\frac{1}{s^3}\mathring{\nabla}\cdot t_1+o(s^{-3})
\end{align}
\end{proposition}
\begin{proof}
From Lemma \ref{l16} and (13)
\begin{align*}
\nabla_iK_{jm} &= \mathring\nabla_i(s\mathring\gamma_{jm}+\frac12{\gamma_1}_{jm})-\mathcal{D}^k_{ij}K_{km}
-\mathcal{D}^k_{im}K_{jk}+o_1^X(1)\\
&=\frac12\mathring\nabla_i{\gamma_1}_{jm}-\mathring\nabla_i{\gamma_1}_{jm}+o_1^X(1)\\
&=-\frac12\mathring\nabla_i{\gamma_1}_{jm}+o_1^X(1).
\end{align*}
where the first term of the second line comes from the fact that $\mathring\nabla\mathring\gamma=0$. Next we compute
\begin{align*}
\nabla_i\nabla_jK_{mn}&= \mathring\nabla_i\nabla_jK_{mn}-\mathcal{D}^k_{ij}\nabla_kK_{mn}-\mathcal{D}^k_{im}\nabla_jK_{kn}
-\mathcal{D}^k_{in}\nabla_jK_{mk}\\
&=-\frac12\mathring\nabla_i\mathring\nabla_j{\gamma_1}_{mn}+o(1)
\end{align*}
So contracting with (12) over $j,m$ followed by $i,n$  we get (18) and contracting instead over $m,n$ and then $i,j$ (19) follows. For (20)
\begin{align*}
\nabla_it_j &= \mathring\nabla_it_j - \mathcal{D}^k_{ij}t_k\\
&=\frac{1}{s}\mathring\nabla_i{t_1}_j+o(s^{-1})
\end{align*}
and the result follows as soon as we contract with (12) over $i,j$. 
\end{proof}
\begin{remark}\label{r9}
As soon as we impose that $\Omega$ has strong decay it follows from the fact that $[\mathcal{L}_X,\mathcal{L}_Y]=\mathcal{L}_{[X,Y]}$ that $\mathcal{L}_{X_i}\tilde\gamma,\,\mathcal{L}_{X_i}\mathcal{L}_{X_j}\tilde\gamma = o_1(s)$ (since, for example, $\mathcal{L}_{\ubar L}\mathcal{L}_{X_i}\tilde\gamma = \mathcal{L}_{[\ubar L,X_i]}\tilde\gamma+\mathcal{L}_{X_i}\mathcal{L}_{\ubar L}\tilde\gamma = o(1)$). As a result its not hard to see early in the proof of Proposition  7 that $$\mathcal{K}_{\gamma(s)}=\frac{\mathring{\mathcal{K}}}{s^2}+\frac{\mathcal{K}_1}{s^3}+O_1(s^{-4})$$
for some Lie constant function $\mathcal{K}_1$. We may therefore provide a simpler proof using Proposition \ref{p9} and the propogation equation (5) 
\begin{align*}
\ubar L\mathcal{K}&=-\tr K\mathcal{K}-\Delta\tr K+\nabla\cdot\nabla\cdot K
\end{align*}
in order to find $\mathcal{K}_1$.
\end{remark}
\begin{theorem}\label{t7}
Suppose $\Omega$ is past asymptotically flat with strong flux decay and $\{\Sigma_s\}$ is some background geodesic foliation. Then for any asymptotically geodesic foliation $\{\Sigma_{s_\star}\}$ with scale factor $\phi>0$ we have
$$s_\star^3\s\rho(s_\star) = \frac{1}{2\phi^3}\Big(\mathring{\mathcal{K}}\ubar\theta-\theta-\mathring{\Delta}\ubar\theta
+4\mathring\nabla\cdot t_1\Big)+o(s_{\star}^0)$$ 
\end{theorem}
\begin{proof}
First let us remind ourselves of Theorem \ref{t5}
\begin{align*}
\s\rho=\rho&+\s{\nabla}\cdot\Big(\frac{|\hat{K}|^2+G(\ubar L,\ubar L)}{\tr K}\s{\nabla}\omega\Big)+\frac12\Big(|\hat{K}|^2+G(\ubar L,\ubar L)\Big)|\nabla\omega|^2\\
&+\nabla\omega\frac{|\hat{K}|^2+G(\ubar L,\ubar L)}{\tr K}+G_{\ubar L}(\nabla\omega)-2\hat{K}(\vec{t}-\nabla\log\tr K,\nabla\omega).
\end{align*}
Denoting the exterior derivative on $\Sigma_s$ by $d_s$, since $\tr K = \frac{2}{s}+\frac{\ubar\theta}{s^2}+o(s^{-2})$, we conclude that $d_s\log\tr K=\frac{1}{2s}d\ubar\theta|_{\Sigma_s}+o(s^{-1})$ giving 
$$\hat{K}(\vec{t}-\nabla\log\tr K,\nabla\omega)|_{\Sigma_{s_\star}} = o(s_\star^{-3}).$$
Since $\mathcal{L}_{\ubar L}^2\tilde\gamma = o_1^X(s^{-1})\cap o_1(s^{-1})$ we also see that
\begin{align*}
|\hat{K}|^2+G(\ubar L,\ubar L)=-\ubar L\tr K-\frac12(\tr K)^2 &= -(-\frac{2}{s^2}-\frac{2}{s^3}\ubar\theta)-\frac12(\frac{2}{s}+\frac{\ubar\theta}{s^2})^2+o_1^X(s^{-3})\cap o_1(s^{-3})\\
&=o_1^X(s^{-3})\cap o_1(s^{-3})
\end{align*}
and therefore, from Remark \ref{r7}:
\begin{align*}
\s\nabla\cdot\Big(\frac{|\hat{K}|^2+G(\ubar L,\ubar L)}{\tr K}\s\nabla\omega\Big) &= \s\nabla\omega\frac{|\hat{K}|^2+G(\ubar L,\ubar L)}{\tr K}+\s\Delta\omega\frac{|\hat{K}|^2+G(\ubar L,\ubar L)}{\tr K}\\
&=\Big(\nabla\omega\frac{|\hat{K}|^2+G(\ubar L,\ubar L)}{\tr K}+|\nabla\omega|^2\ubar L\frac{|\hat{K}|^2+G(\ubar L,\ubar L)}{\tr K}\\
&\indent+(\Delta\omega-2\hat{K}(\nabla\omega,\nabla\omega))\frac{|\hat{K}|^2+G(\ubar L,\ubar L)}{\tr K}\Big)\Big|_{\Sigma_{s_\star}}\\
&=o(s_\star^{-3}).
\end{align*}
From the strong flux decay condition we have $G_{\ubar L}(\nabla\omega)|_{\Sigma_{s_\star}} = o(s_\star^{-3})$ also.
From (19) we have
\begin{align*}
\Delta\log\tr K &= \frac{\Delta\tr K}{\tr K}-\frac{|\nabla\tr K|^2}{(\tr K)^2}\\
&=\frac{\mathring{\Delta}\ubar\theta}{2s^3}+o(s^{-3})
\end{align*}
and combining this with Propositions \ref{p7} and \ref{p9}:
\begin{align*}
\s\rho &= \rho|_{\Sigma_{s_\star}} +o(s_\star^{-3})\\
&=\frac{1}{\omega^3}\Big(\frac12\mathring{\mathcal{K}}\ubar\theta
+\frac12\mathring\nabla\cdot\mathring\nabla\cdot\gamma_1+\mathring\Delta\ubar\theta
-\frac12\theta\Big)+\frac{1}{\omega^3}\mathring\nabla\cdot t_1-\frac{1}{2\omega^3}\mathring\Delta\ubar\theta+o(s_\star^{-3})\\
&=\frac{1}{2\omega^3}\Big(\mathring{\mathcal{K}}\ubar\theta-\theta-\mathring\Delta\ubar\theta+4\mathring{\nabla}\cdot t_1\Big)+o(s_\star^{-3})
\end{align*}
having used Proposition \ref{p6} in the final line to substitute $\frac12\mathring\nabla\cdot\mathring\nabla\cdot\gamma_1 +\mathring\Delta\ubar\theta= \mathring\nabla\cdot t_1$ and the result follows.
\end{proof}
\begin{remark}\label{r10}
We would like to bring to the attention of the reader our use of (15) in the second to last equality in the proof of Theorem \ref{t7}. Assuming $\{\Sigma_{s_\star}\}$ is in fact a geodesic foliation, running a parallel argument to decompose $\s\rho$ as we did for $\rho$ allows us to conclude that (15) must also hold for $\{\Sigma_{s_\star}\}$. We refer the reader to \cite{MS1} (Proposition 3) to observe that under the additional decay of Theorem \ref{t7}, part 4 from Definition \ref{d7} is no longer necessary to give (15) for an arbitrary geodesic foliation provided it holds for at least one. We will exploit this fact in Section 5.
\end{remark}
\begin{corollary}\label{c5}
With the same hypotheses as in Theorem \ref{t7} we have
$$\lim_{s_\star\to\infty}m(\Sigma_{s_\star}) = \frac14\Big(\frac{1}{4\pi}\int (\mathring{\mathcal{K}}\ubar\theta-\theta-\mathring{\Delta}\ubar\theta+4\mathring\nabla\cdot t_1)^\frac23\mathring{dA}\Big)^\frac32$$
\end{corollary}
\begin{proof}
From Theorem \ref{t7} we directly conclude
$$4\pi(4m(\Sigma_{s_\star}))^{\frac23} =\int (2\s{\rho})^\frac23 dA_\omega= \int\frac{1}{\omega^2}\Big(\mathring{\mathcal{K}}\ubar\theta-\theta-\mathring\Delta\ubar\theta+4\mathring{\nabla}\cdot t_1+o(1)\Big)^{\frac23}f\omega^2\mathring{dA}$$
giving 
$$4\pi(4\lim_{s_\star\to\infty}m(\Sigma_{s_\star}))^{\frac23} = \int\Big(\mathring{\mathcal{K}}\ubar\theta-\theta-\mathring\Delta\ubar\theta+4\mathring{\nabla}\cdot t_1\Big)^{\frac23}\mathring{dA}$$
by the Dominated Convergence Theorem.
\end{proof}
 Finally we're ready to prove Theorem \ref{t2}:
\begin{proof}(Theorem \ref{t2})
The first claim of Theorem \ref{t2} is a simple consequence of Theorem \ref{t1}. 
Property (P) and Theorem \ref{t7} enforces that 
$$0\leq \lim_{s_\star\to\infty}s_\star^3\s\rho = \frac{1}{2\phi^3}(\mathring{\mathcal{K}}\ubar\theta-\theta-\mathring{\Delta}\ubar\theta+4\mathring\nabla\cdot t_1)$$
and therefore Theorem \ref{t1}, Corollary \ref{c5}, Lemma \ref{l5} and Proposition \ref{p8} gives 
$$m(\Sigma_0)\leq \lim_{s_\star\to\infty}m(\Sigma_{s_\star}) = \frac14\Big(\frac{1}{4\pi}\int (\mathring{\mathcal{K}}\ubar\theta-\theta-\mathring{\Delta}\ubar\theta+4\mathring\nabla\cdot t_1)^\frac23\mathring{dA}\Big)^\frac32=\inf_{\phi>0}E(\phi)\leq m_B.$$
The rest of the proof is settled identically as in Theorem \ref{t6}.
\end{proof}
\newpage
\section{Spherical Symmetry}
For the known null Penrose inequality in spherical symmetry (see \cite{Ha}) we provide proof within our context in order to motivate a class of perturbations on the black hole exterior that maintain both the asymptotically flat and strong flux decay conditions. We also show the existence of an asymptotically geodesic (SP)-foliation for a subclass of these perturbations toward a proof of the null Penrose conjecture. 

\subsection{The metric}
 In polar areal coordinates \cite{P} the metric takes the form 
$$g=-a(t,r)^2dt\otimes dt+b(t,r)^2dr\otimes dr+r^2\mathring\gamma$$
for $\mathring\gamma$ the standard round metric on $\mathbb{S}^2$. From which the change in coordinates $(t,r)\to(v,r)$ given by
$$dv = dt+\frac{b}{a}dr$$
produces the metric and metric inverse given by 
\begin{align*}
g &= -\mathfrak{h}e^{2\beta}dv\otimes dv+e^\beta(dv\otimes dr+dr\otimes dv)+r^2\mathring\gamma\\
g^{-1}&=e^{-\beta}(\partial_v\otimes\partial_r+\partial_r\otimes\partial_v)+\mathfrak{h}\partial_r\otimes\partial_r+\frac{1}{r^2}\mathring\gamma^{-1}
\end{align*}
for $\mathfrak{h}=(1-\frac{2M(t,r)}{r})$ where $M(t,r):=\frac{r}{2}(1-\frac{1}{b^2})$ and $a(t,r)^2 = \mathfrak{h}e^{2\beta}$.\\
It's a well known fact that assigning $M(t,r)=m_0>0$ and $\beta(t,r) = 0$ for $m_0$ a constant the above metric covers the region given by $v>0$ in Kruskal spacetime or Schwarzschild geometry in an `Eddington-Finkelstein' coordinate chart. We will therefore refer to the null hypersurfaces $\Omega:=\{v=v_0\}$ as the \textit{standard null-cones} (of spherically symmetric spacetime) as they agree with the similarly named hypersurfaces in the Schwarzschild case.
\subsection{Calculating $\rho$}
We approach the calculation similarly to the case of Schwarzschild. Denoting the gradient of $v$ by $Dv$ we use the identity $D_{D v}D v = \frac12D |D v|^2$ to see $\ubar L := Dv=e^{-\beta}\partial_r$ satisfies $D_{\ubar L}\ubar L =0$ providing us our choice of geodesic generator for $\Omega$ and level set function $s$ (as in Section 3). For convenience we will choose our background foliation $\{\Sigma_r\}$ of $\Omega$ to be the level sets of the coordinate $r$. An arbitrary cross section $\Sigma$ of $\Omega$ is therefore given as a graph over $\Sigma_{r_0}$ (for some $r_0$) which we Lie drag along $\partial_r$ to the rest of $\Omega$ giving some $\omega\in\mathcal{F}(\Omega)$. On $\Sigma$ we therefore have the linearly independent normal vector fields 
\begin{align*}
\ubar L&=e^{-\beta}\partial_r\\
D(r-\omega)&=e^{-\beta}\partial_v+\mathfrak{h}\partial_r-\nabla\omega
\end{align*}
where in this subsection (5.2) $\nabla$ will temporarily denote the induced covariant derivative on $\Sigma_r$. We wish to find the null section $L\in\Gamma(T^\perp\Sigma)$ satisfying $\langle L,\ubar L\rangle=2$. Since $L = c_1\ubar L+c_2D(r-\omega)$ we have
\begin{align*}
2&=\langle L,\ubar L\rangle = c_2e^{-\beta}\partial_r(r-\omega)\\
&=c_2e^{-\beta}\\
0&=\langle L,L\rangle = 2c_1c_2\langle\ubar L,D(r-\omega)\rangle+c_2^2\langle D(r-\omega),D(r-\omega)\rangle\\
&=2c_1c_2e^{-\beta}+c_2^2(e^{-2\beta}\langle\partial_v,\partial_v\rangle+|\nabla\omega|^2+2e^{-\beta}\mathfrak{h}\langle\partial_v,\partial_r\rangle)\\
&=2c_1c_2e^{-\beta}+c_2^2(\mathfrak{h}+|\nabla\omega|^2)
\end{align*}
giving $c_2 = 2e^{\beta}$ and $c_1 = -e^{2\beta}(\mathfrak{h}+|\nabla\omega|^2)$ so that
\begin{align*}
L &= -e^{\beta}(\mathfrak{h}+|\nabla\omega|^2)\partial_r+2\partial_v+2\mathfrak{h}e^{\beta}\partial_r-2e^{\beta}\nabla\omega\\
&=2\partial_v +e^\beta(\mathfrak{h}-|\nabla\omega|^2)\partial_r-2e^\beta\nabla\omega\\
&= 2\partial_v+e^{\beta}(\mathfrak{h}-|\s\nabla\omega|^2)\partial_r-2e^\beta(\s\nabla\omega-|\s\nabla\omega|^2\partial_r)\\
&=2\partial_v +e^\beta(\mathfrak{h}+|\s\nabla\omega|^2)\partial_r-2e^\beta\s\nabla\omega
\end{align*}
having used the fact that $\s\nabla\omega = \nabla\omega+|\nabla\omega|^2\partial_r$ to get the third equality. We note from the warped product structure (as for Kruskal spacetime) that $E_{\partial_r}(\Sigma_{r_0}) = \mathcal{L}(\mathbb{S}^2)|_\Omega$ where $\mathcal{L}(\mathbb{S}^2)$ is the set of lifted vector fields from the $\mathbb{S}^2$ factor of the spacetime product manifold. As a result we may globally extend $V\in E_{\partial_r}(\Sigma_{r_0})$ to satisfy $[\partial_v,V]=0$. The following facts are a direct application of the Koszul formula, we refer the reader to \cite{O} (pg 206) for the details:
\begin{align}
D_{\partial_r}\partial_v &= -\frac12\partial_r(\mathfrak{h}e^{2\beta})e^{-\beta}\partial_r\\
D_V\partial_v&=0\\
D_{\partial_r}\partial_r &= \partial_r\beta\partial_r\\
D_V\partial_r&=\frac{1}{r}V.
\end{align}
\begin{lemma}\label{l17}
Suppose $\Omega=\{v=v_0\}$ is the standard null cone in a spherically symmetric spacetime of metric
$$g = -\mathfrak{h}e^{2\beta(v,r)}dv\otimes dv+e^\beta(v,r)(dv\otimes dr+dr\otimes dv)+r^2\mathring\gamma$$
where $\mathfrak{h} = (1-\frac{2M(v,r)}{r})$ and $\mathring\gamma$ is the round metric on $\mathbb{S}^2$. Then for some cross section $\Sigma_{r_0}\subset\Omega$ and $\omega\in\mathcal{F}(\Sigma_{r_0})$, $\Sigma:=\{r=\omega\circ\pi\}$ produces the data (writing $\omega\circ\pi$ as $\omega$):
\begin{align*}
\gamma &= \omega^2\mathring\gamma\\
\ubar\chi&=\frac{e^{-\beta(v_0,\omega)}}{\omega}\gamma\\
\tr\ubar\chi&=\frac{2e^{-\beta(v_0,\omega)}}{\omega}\\
\chi&=e^{\beta(v_0,\omega)}\Big((\mathfrak{h}+|\s\nabla\omega|^2)\frac{\gamma}{\omega}-2\tilde H^{\omega}-2\beta_r\s{d}\omega\otimes\s{d}\omega\Big)\\
\tr\chi&=\frac{2e^\beta(v_0,\omega)}{\omega}(\mathfrak{h}-\omega^2\s\Delta\log\omega-\omega\beta_r|\s\nabla\omega|^2)\\
\zeta&=-\s{d}\log\omega\\
\rho&=\frac{2M(v_0,\omega)}{\omega^3}+\s\Delta\omega+\frac{\beta_r}{\omega}|\s\nabla\omega|^2
\end{align*}
\end{lemma}
\begin{proof}
For any $V\in E_{\partial_r}(\Sigma_{r_0})$ we have from Lemma \ref{l12} that $\tilde V := V+V\omega\partial_r|_\Sigma\in\Gamma(T\Sigma)$ so that the first identity follows directly from the metric restriction. From (25):
$$D_{\tilde V}\ubar L = e^{-\beta}D_V(\partial_r)+e^{\beta}V\omega D_{\ubar L}\ubar L =\frac{e^{-\beta}}{r}V$$ 
so the second identity is given by
\begin{align*}
\ubar\chi(\tilde V,\tilde W)&=\langle D_{\tilde V}\ubar L,\tilde W\rangle\\
&=\frac{e^{-\beta}}{r}\langle V,W\rangle
\end{align*}
and a trace over $V,W$ gives the third so that $\s\Delta\log\tr\ubar\chi = -\s\Delta\beta-\s\Delta\log\omega$. For the forth identity:
\begin{align*}
\chi(\tilde V,\tilde W) &= 2\langle D_{\tilde V}\partial_v,\tilde W\rangle+e^\beta(\mathfrak{h}+|\s\nabla\omega|^2)\langle D_{\tilde V}\partial_r,\tilde W\rangle-2e^\beta\langle D_{\tilde V}\s\nabla\omega,\tilde W\rangle-2\beta_re^\beta\tilde V\omega\tilde W\omega\\
&=e^{\beta}(\mathfrak{h}+|\s\nabla\omega|^2)\frac{1}{\omega}\langle\tilde V,\tilde W\rangle-2e^\beta\tilde H^{\omega}(\tilde V,\tilde W)-2\beta_re^\beta(\s{d}\omega\otimes\s{d}\omega)(\tilde V,\tilde W)
\end{align*}
where $\langle D_{\tilde V}\partial_v,\tilde W\rangle = 0$ from (22) and (23) to give the second equality. Taking a trace over $\tilde V, \tilde W$ we conclude with the fifth identity:
\begin{align*}
\tr\chi|_\Sigma &= \frac{2e^{\beta(v_0,\omega)}}{\omega}(\mathfrak{h}+|\s\nabla\omega|^2-\omega\s\Delta\omega)-2\beta_re^{\beta(v_0,\omega)}|\s\nabla\omega|^2\\
&=\frac{2e^\beta}{\omega}(\mathfrak{h}-\omega^2(\frac{\s\Delta\omega}{\omega}-\frac{|\s\nabla\omega|^2}{\omega^2}))-2\beta_re^\beta|\s\nabla\omega|^2\\
&=\frac{2e^\beta}{\omega}(\mathfrak{h}-\omega^2\s\Delta\log\omega-\omega\beta_r|\s\nabla\omega|^2).
\end{align*}
As a result we have that
$$\langle \vec{H},\vec{H}\rangle = \tr\ubar\chi\tr\chi = \frac{4}{\omega^2}(\mathfrak{h}-\omega^2\s\Delta\log\omega-\omega\beta_r|\s\nabla\omega|^2).$$
Since the metric on $\Sigma$ is given by $\omega^2\mathring\gamma$ we conclude that it has Gaussian curvature 
$$\mathcal{K} = \frac{1}{\omega^2}(1-\mathring\Delta\log\omega)=\frac{1}{\omega^2}-\s\Delta\log\omega$$ and therefore
$$\mathcal{K}-\frac14\langle\vec{H},\vec{H}\rangle = \frac{2M(v_0,\omega)}{\omega^3}+\frac{\beta_r}{\omega}|\s\nabla\omega|^2.$$
Moreover, the torsion is given by
\begin{align*}
\zeta(\tilde V) &= \frac12\langle D_{\tilde V}\ubar L,L\rangle\\
&=\frac{e^{-\beta}}{2r}\langle V,L\rangle\\
&=-\frac{1}{r}V\omega\\
&=-\frac{1}{r}\tilde V\omega
\end{align*}
from which we conclude $\zeta(\tilde V)|_\Sigma = -\tilde V\log\omega$ and $\s\nabla\cdot\zeta = -\s\Delta\log\omega$, giving
$$\rho = \frac{2M(v_0,\omega)}{\omega^3}+\s\Delta\beta+\frac{\beta_r}{\omega}|\s\nabla\omega|^2.$$
\end{proof}
\begin{remark}\label{r11}
We recover the data of Lemma \ref{l4} as soon as we set $m_0 = M$, $\beta=0$ and $r_0 = 2m_0$ as expected.
\end{remark}
So in comparison to Schwarzschild spacetime we have the additional terms $\s\Delta\beta+\frac{\beta_r}{\omega}|\s\nabla\omega|^2$ in the flux function $\rho$. It turns out that a non-trivial $G(\ubar L,\ubar L)$ is responsible. Since $\s\Delta\beta=\beta_{rr}|\s\nabla\omega|^2+\beta_r\s\Delta\omega$ and
\begin{align*}
G(\ubar L,\ubar L)&= -\ubar L\tr K-\frac12(\tr K)^2-|\hat{K}|^2\\
&=-e^{-\beta}\partial_r(\frac{2e^{-\beta}}{r})-\frac12\frac{4e^{-2\beta}}{r^2}\\
&=\frac{2\beta_r}{r}e^{-2\beta}
\end{align*}
it follows, for arbitrary $\omega$, that $\s\Delta\beta(\omega) + \frac{\beta_r}{\omega}|\s\nabla\omega|^2= 0$ if and only if $\beta$ is independent of the $r$-coordinate and therefore $G(\ubar L,\ubar L) = 0$. For the function $M(v_0,r)$ we look to $G(\ubar L, L)$ along the foliation $\{\Sigma_r\}$ since:
\begin{align*}
G(\ubar L,L) &= \ubar L\tr\chi-2\mathcal{K}_s+2\nabla\cdot t+2|\vec{t}|^2+\langle\vec{H},\vec{H}\rangle\\
&=e^{-\beta}\partial_r(\frac{2e^\beta}{r}(1-\frac{2M}{r}))-\frac{2}{r^2}+\frac{4}{r^2}(1-\frac{2M}{r})\\
&=\frac{2\beta_r}{r}(1-\frac{2M}{r})-\frac{4M_r}{r^2}.
\end{align*}
It follows from Lemma \ref{l11}, on $\Sigma_r$, that
\begin{align*}
G_{\ubar L} &=0.
\end{align*}
Since these components are all that contribute to the monotonicity of (2) for the foliation $\{\Sigma_r\}$ we see that our need of the null energy condition reduces to 
$$0\leq\mathfrak{h}\beta_r\leq \frac{2M_r}{r}$$
on $\{\mathfrak{h}\geq 0\}\cap\Omega$. Next we show that $\{\Sigma_r\}$ is a re-parametrization of a geodesic (SP)-foliation:
\subsection{Asymptotic flatness}
We now wish to choose the necessary decay on $\beta$ and $M$ in order to employ Theorem \ref{t2}. For $\ubar L = e^{-\beta}\partial_r$ the geodesic foliation $\{\Sigma_s\}$ has level set function given by
$$s(r) = \int_{r_0}^re^{\beta(t)}dt$$
for which $\omega = const.\iff s = const.$ and therefore
$$\rho(s) = \frac{2M(r(s))}{r(s)^3}.$$
It follows from Lemma \ref{l17} that $\frac{1}{4}\langle\vec{H},\vec{H}\rangle-\frac13\s\Delta\log\rho = \frac{\mathfrak{h}}{r(s)^2}>0\iff r(s)>r_0=2M(v_0,r_0)$ as in Schwarzschild.
\begin{lemma}\label{l18}
Choosing $|\beta(v_0,r)| = o_2(r^{-1})$ integrable and $M(v_0,r) = m_0+o(r^0)$ for some constant $m_0$, $\Omega$ is asymptotically flat with strong flux decay.
\end{lemma}
\begin{proof}
We've already verified that $G_{\ubar L} = 0$. Since $\frac{ds}{dr} = e^{\beta(r)} = (1+\beta\frac{e^\beta-1}{\beta})$, $|\beta|$ is integrable and $\frac{e^\beta-1}{\beta}$ is bounded it follows that $\frac{ds}{dr} = 1+f$ where $|f|=o_2(r^{-1})$ is integrable. As a result
$$s = r-r_0+\int_{r_0}^\infty f(t)dt-\int_{r}^{\infty}f(t)dt=r-c_0+o_3(r^0)$$
where $\beta_0 = \int_{r_0}^\infty f(t)dt$ and $c_0 = r_0-\beta_0$. We conclude that
$r(s) = s+c_0+o_3(s^0)$ since our assumptions on $\beta$ imply that $\int_{r(s)}^\infty f(t)dt = o_3(s^0)$. From the fact that
$$\gamma_s = r^2\mathring\gamma|_{\Sigma_s} = (s+c_0+o_3(1))^2\mathring\gamma = s^2(1+\frac{c_0}{s}+o_3(s^{-1}))^2\mathring\gamma=s^2\mathring\gamma+2c_0s\mathring\gamma+o_3(s)\mathring\gamma$$
we see $\tilde \gamma = o_3(s)\mathring\gamma$ ensuring condition 1 of Definition \ref{d7} holds up to strong decay given that all dependence on tangential derivatives falls on the $\ubar L$ Lie constant tensor $\mathring\gamma$. Since $\vec{t} = 0$ for this foliation condition 2 follows trivially up to strong decay. If we assume that $M(v_0,r) = m_0+o(1)$ for some constant $m_0$ we see directly from Lemma \ref{l17}
\begin{align*}
\tr Q &= \tr\chi|_{\Sigma_s}\\
&=\frac{2}{r}(1-\frac{2M}{r})|_{\Sigma_s}+o(s^{-2})\\
&=\frac{2}{s}(1-\frac{c_0}{s})(1-\frac{2m_0}{s})+o(s^{-2})\\
&=\frac{2}{s}-2\frac{c_0+2m_0}{s^2}+o(s^{-2})
\end{align*}
giving us the third condition of Definition \ref{d7}.\\
We refer the reader to \cite{MS1} to observe the use of the forth condition of Definition \ref{d7} in proving (15) for an arbitrary geodesic foliation. As mentioned in Remark \ref{r10}, strong flux decay bypasses our need of this condition since $\tr Q = \frac{2\mathring{\mathcal{K}}}{s}+o(s^{-1})$ is verified above. 
\end{proof}
From Lemma \ref{l18}, Theorem \ref{t2}, Theorem \ref{t4} and the comments immediately proceeding Remark \ref{r11} we have the following proof of the known (see \cite{Ha}) null Penrose conjecture in spherical symmetry:
\begin{theorem}\label{t8}
Suppose $\Omega:=\{v=v_0\}$ is a standard null cone of a spherically symmetric spacetime of metric
$$ds^2 = -\Big(1-\frac{2M(v,r)}{r}\Big)e^{2\beta(v,r)}dv^2+2e^{\beta(v,r)}dvdr+r^2\Big(d\vartheta^2+\sin\vartheta^2d\varphi^2\Big)$$
where 
\begin{enumerate}
\item $|\beta(v_0,r)| = o_2(r^{-1})$ is integrable
\item $M(v_0,r) = m_0+o(r^0)$ for some constant $m_0>0$
\item $0\leq\mathfrak{h}\beta_r\leq\frac{2M_r}{r}$
\end{enumerate}
Then, 
$$\sqrt{\frac{|\Sigma|}{16\pi}}\leq m_0$$
for $m_0$ the Bondi mass of $\Omega$ and $\Sigma:=\{r_0 = 2M(v_0,r_0)\}$. In the case of equality we have $\beta = 0$ and $M=m_0$ so that $\Omega$ is a standard null cone of Schwarzschild spacetime.
\end{theorem}
\subsection{Perturbing spherical symmetry}
We wish to study perturbations off of the spherically symmetric metric given in Theorem \ref{t8} for the coordinate chart $(v,r,\vartheta,\varphi)$. We start by choosing a 1-form $\eta$ such that $\eta(\partial_r(\partial_v)) = \mathcal{L}_{\partial_v}\eta = 0$ and a 2-tensor $\gamma$ satisfying $\gamma(\partial_r(\partial_v),\cdot) = \mathcal{L}_{\partial_v}\gamma = 0$ with restriction $\gamma|_{(v,r)\times\mathbb{S}^2}$ positive definite. Finally we choose smooth functions $M$, $\beta$ and $\alpha$. Defining $\vec{\eta}$ to be the unique vector field satisfying $\gamma(\vec{\eta},X) = \eta(X)$ for arbitrary $X\in \Gamma(TM)$ and $r^2|\vec\eta|^2 := \gamma(\vec{\eta},\vec{\eta})$ the spacetime metric and its inverse are given by
\begin{align*}
g &= -(\mathfrak{h}+\alpha)e^{2\beta}dv\otimes dv+e^{\beta}(dv\otimes(dr+\eta)+(dr+\eta)\otimes dv)+r^2\gamma\\
g^{-1}&=e^{-\beta}(\partial_v\otimes \partial_r+\partial_r\otimes\partial_v)+(\mathfrak{h}+\alpha+|\vec{\eta}|^2)\partial_r\otimes\partial_r-(\vec{\eta}\otimes \partial_r+\partial_r\otimes\vec{\eta})+\frac{1}{r^2}\gamma^{-1}.
\end{align*}
We see that $\Omega:= \{v = v_0\}$ remains a null hypersurface with $\ubar L (= Dv) = e^{-\beta}\partial_r\in \Gamma(T\Omega)\cap\Gamma(T^\perp\Omega)$. Our metric resembles the perturbed metric used by Alexakis \cite{A} to successfully verify the Penrose inequality for vacuum perturbations of the standard null cone of Schwarzschild spacetime. 
We'll need the following to specify our decay conditions:
\begin{definition}\label{d10}
Suppose $\Omega$ extends to past null infinity with level set function $s$ for some null generator $\ubar L$.
For a transversal $k$-tensor $T$
\begin{itemize}
\item We say $T(s,\delta)=\delta o_n^X(s^{-m})$ if $T=o_n^X(s^{-m})$ and
$$\limsup_{\delta\to 0}\sup_\Omega\frac{1}{\delta}|s^m(\mathcal{L}_{X_{i_1}}...\mathcal{L}_{X_{i_j}}T)(s,\delta)|<\infty\,\,\text{for}\,\, 0\leq j\leq n$$
\item We define 
$$|T|^2_{\mathring{H}^m}=|T|^2_{\mathring\gamma}+|\mathring\nabla T|^2_{\mathring\gamma}+\dots+|\mathring\nabla^mT|^2_{\mathring\gamma}.$$
\end{itemize}
\end{definition}
\underline{Decay Conditions on $\Omega$:}
\begin{enumerate}
\item $r^2\gamma = r^2\mathring\gamma+r\delta\gamma_1+\tilde\gamma$ where:
\begin{enumerate}
\item $\mathring\gamma$ is the $\partial_r$-Lie constant, transversal standard round metric on $\mathbb{S}^2$ independent of $\delta$
\item $\gamma_1$ is a $\partial_r$-Lie constant, transversal 2-tensor independent of $\delta$
\item $\tilde\gamma$ is a transversal 2-tensor satisfying $(\mathcal{L}_{\partial_r})^i\tilde\gamma = \delta o_{5-i}^X(r^{1-i})$ for $0\leq i\leq 3$
\end{enumerate}
\item $\alpha=\delta\frac{\alpha_0}{r}+\tilde\alpha$ where $\alpha_0$ is a $\partial_r$-Lie constant function independent of $\delta$ and $|\tilde\alpha|_{\mathring{H}^2}\leq \delta h_1(r)$ for $h_1=o(r^{-1})$
\item $\beta$ satisfies:
\begin{enumerate}
\item $|\beta| = o_2(r^{-1})$ is r-integrable
\item $|\mathring\nabla\beta|_{\mathring{H}^3}\leq \delta h_2(r)$ for some integrable $h_2 = o(r^{-1})$
\item $|\mathring\nabla\beta_r|_{\mathring{H}^2}=O(r^{-1})$
\end{enumerate}
\item $M = m_0+\tilde m$ where $m_0>0$ is constant independent of $\delta$ and $|\tilde m|_{\mathring{H}^2}\leq \delta h_3(r)$ for $h_3 = o(1)$ 
\item $\eta$ is a transversal 1-form satisfying:
\begin{enumerate}
\item $\eta = o_2(1)$
\item $|\eta|_{\mathring{H}^3}+r|\mathcal{L}_{\partial_r}\tilde\eta|_{\mathring{H}^3}\leq\delta h_4(r)$ for $h_4=o(1)$.
\end{enumerate}
\end{enumerate}
\subsubsection{The geodesic foliation}
As in the spherically symmetric case we identify the null geodesic generator $Dv = e^{-\beta}\partial_r$. We will again for convenience take the background foliation to be level sets of the coordinate $r$. We wish therefore to relate the given decay in $r$ to the geodesic foliation given by the generator $\ubar L:=Dv$ in order to show $\Omega$ is asymptotically flat with strong flux decay. \\
Once again $\frac{ds}{dr} = e^{\beta} = 1+f$ where $f = \beta\frac{e^\beta-1}{\beta}$ is $r$-integrable due to decay condition 3. Taking local coordinates $(\vartheta,\varphi)$ on $\Sigma_{r_0}$ (for some $r_0$) we have
\begin{equation}
s = r-c_0(\vartheta,\varphi)-\beta_1(r,\vartheta,\varphi)
\end{equation}
for $\beta_0(\vartheta,\varphi): = \int_{r_0}^\infty f(t,\vartheta,\varphi)dt$, $c_0=r_0-\beta_0$ and $\beta_1(r,\vartheta,\varphi) = \int_r^\infty f(t,\vartheta,\varphi)dt$. Since each $\Sigma_r$ is compact, an $m$-th order partial derivative of $f$ is bounded by $C|\mathring\nabla f|_{\mathring{H}^{m-1}}$ for some constant $C$ independent of $r$ (from decay condition 3). From decay condition 3, provided $m\leq 4$, derivatives in $\vartheta,\varphi$ of $\beta_0$ and $\beta_1$ pass into the integral (for fixed $r$) onto $f$ and are bounded. On any $\Sigma_s$ (i.e fixed $s$) it follows from (26) that
$$\partial_{\vartheta(\varphi)} r = -\frac{\int_{r_0}^r\partial_{\vartheta(\varphi)} f(t,\vartheta,\varphi)dt}{1+f}=-e^{-\beta}\int_{r_0}^r\beta_{\vartheta(\varphi)}e^\beta dt$$
with bounded derivatives up to third order. It's a simple verification in local coordinates, from
$$r(s,\vartheta,\varphi) = s+c_0(\vartheta,\varphi)+\beta_1(r(s,\vartheta,\varphi),\vartheta,\varphi),$$ 
that $\partial_s^i\beta_1=o_{3-i}^X(s^{-i})$ for $0\leq i\leq 3$. Coupled with the fact that $\mathcal{L}_{\ubar L} = e^{-\beta}\mathcal{L}_{\partial_r}$ on transversal tensors we conclude that $(\mathcal{L}_{\ubar L})^i\tilde\gamma=o_{3-i}^X(s^{1-i})$ for $0\leq i\leq 3$ and therefore
\begin{equation}
\gamma_s = r^2\gamma|_{\Sigma_s} = s^2\mathring\gamma+s\Gamma_1
+\tilde\Gamma
\end{equation}
where
\begin{align*}
\Gamma_1 &= 2c_0\mathring\gamma+\delta\gamma_1\\
\tilde\Gamma &= \tilde\gamma+2s\beta_1\mathring\gamma+c_0^2\mathring\gamma+c_0\delta\gamma_1 +\beta_1^2\mathring\gamma+2c_0\beta_1\mathring\gamma+\beta_1\delta\gamma_1
\end{align*}
satisfies the requirements towards strong decay.
\subsubsection{Calculating $\rho$}
Since we will compare computations for the foliation $\{\Sigma_r\}$ with the geodesic foliation of 5.4.1 we will revert back to denoting the covariant derivative on $\Sigma_s$ by $\nabla$ and the covariant derivative on $\Sigma_r$ by $\s\nabla$. For the foliation $\{\Sigma_r\}$ we have the linearly independent normal vector fields
\begin{align*}
\ubar L&=e^{-\beta}\partial_r\\
Dr &= e^{-\beta}\partial_v+(\mathfrak{h}+\alpha+|\vec{\eta}|^2)\partial_r-\vec{\eta}
\end{align*}
from which similar calculations as in spherical symmetry yield the unique null normal satisfying $\langle \ubar L,L\rangle = 2$ to be given by
$$L = 2\partial_v+e^{\beta}(\mathfrak{h}+\alpha+|\vec{\eta}|^2)\partial_r-2e^{\beta}\vec{\eta}.$$
\begin{lemma}\label{l19}
We have 
\begin{align*}
\ubar\chi &= e^{-\beta}(r\mathring\gamma+\frac{\delta}{2}\gamma_1+\frac{1}{2}(\mathcal{L}_{\partial_r}\tilde\gamma))\\
\tr\ubar\chi &= e^{-\beta}(\frac{2}{r}+\frac{\delta\ubar\theta}{r^2})+\delta o_4^X(r^{-2})
\end{align*}
for $\ubar\theta:= - \frac12\mathring{\tr}\gamma_1$. Moreover,
$$\s\nabla^m\ubar\chi = -e^{-\beta}\frac{\delta}{2}\mathring\nabla^m\gamma_1+\delta o_{4-m}^X(1),\,\,\,0\leq m\leq 4.$$
\end{lemma}
\begin{proof}
First we extend $V,W\in E_{\partial_r}(\Sigma_{r_0})$ off of $\Omega$ such that $[\partial_v,V(W)] = 0$. Then for $\ubar\chi$:
\begin{align*}
\ubar\chi(V,W)&=\langle D_V(e^{-\beta}\partial_r),W\rangle\\
&=e^{-\beta}\langle D_V\partial_r,W\rangle\\
&=e^{-\beta}\frac12\partial_r\langle V,W\rangle\\
&=e^{-\beta}(r\mathring\gamma(V,W)+\frac{\delta}{2}\gamma_1(V,W)+\frac{1}{2}\mathcal{L}_{\partial_r}\tilde\gamma(V,W))
\end{align*}
having used the Koszul formula to get the third line. So using a basis extension $\{X_1,X_2\}\subset E_{\partial_r}(\Sigma_{r_0})$ Proposition \ref{p6} provides the inverse metric $\gamma(r)^{ij}=\frac{1}{r^2}\mathring\gamma^{ij}-\frac{\delta}{r^3}\mathring{\gamma_1}^{ij}+\delta o_5^X(r^{-3})$ and $\tr\ubar\chi$ follows by contracting $\gamma(r)^{-1}$ with $\ubar\chi$. For the final identity we note from Lemma \ref{l16} we have for the decomposition $\gamma_r = r^2\mathring\gamma+r\delta\gamma_1+\tilde\gamma$ the difference tensor 
\begin{align*}
\langle \mathcal{D}(V,W),U\rangle&=\langle \s\nabla_VW-\mathring\nabla_VW,U\rangle\\
&=\frac{r\delta}{2}\Big(\mathring\nabla_V{\gamma_1}(W,U)+\mathring\nabla_W{\gamma_1}(V,U)-\mathring\nabla_U{\gamma_1}(V,W)\Big)\\
&\indent+ \frac{1}{2}\Big(\mathring\nabla_V\tilde\gamma(W,U)+\mathring\nabla_W\tilde\gamma(V,U)
-\mathring\nabla_U\tilde\gamma(V,W)\Big)
\end{align*}
for $V,W,U\in E(\Sigma_{r_0})$. So proceeding as in Proposition \ref{p9}
\begin{align*}
\s\nabla_i\ubar\chi_{jk} &= \mathring\nabla_i\ubar\chi_{jk}-\mathcal{D}^m_{ij}\ubar\chi_{mk}-\mathcal{D}^m_{ik}\ubar\chi_{jm}\\
&=\mathring\nabla_i(re^{-\beta}\mathring\gamma_{jk}+e^{-\beta}\frac{\delta}{2}{\gamma_1}_{jk}+e^{-\beta}\frac{1}{2}(\mathcal{L}_{\partial_r}\tilde\gamma)_{jk})-e^{-\beta}\delta\mathring\nabla_i({\gamma_1}_{jk})+\delta o_4^X(1)\\
&=r\mathring\gamma_{jk}\mathring\nabla_i(e^{-\beta})-e^{-\beta}\frac{\delta}{2}\mathring\nabla_i{\gamma_1}_{jk}+\frac{\delta}{2}{\gamma_1}_{jk}\mathring\nabla_i(e^{-\beta})+\delta o_3^X(1)\\
&=-e^{-\beta}\frac{\delta}{2}\mathring\nabla_i{\gamma_1}_{jk}+\delta o_3^X(1).
\end{align*}
Iteration provides our result
\begin{align*}
\s\nabla^{m}\ubar\chi &= -e^{-\beta}\frac{\delta}{2}\mathring\nabla^m{\gamma_1}_{jk}+\delta o_{4-m}^X(1),\,\,\,1\leq m\leq 4
\end{align*}
from decay condition 3.
\end{proof}
 For $\chi$ we have
\begin{align*}
\chi(V,W) &= 2\langle D_V\partial_v,W\rangle+e^\beta(\mathfrak{h}+\alpha+|\vec{\eta}|^2)\langle D_V\partial_r,W\rangle-2\langle D_V{e^{\beta}\vec\eta},W\rangle\\
&=2\langle D_V\partial_v,W\rangle+e^{2\beta}(\mathfrak{h}+\alpha+|\vec\eta|^2)\ubar\chi(V,W)-2\s\nabla_V(e^\beta\eta)(W)
\end{align*}
and using the Koszul formula on the first term we see
\begin{align*}
2\langle D_V\partial_v,W\rangle &= V(e^\beta\eta(W))+\partial_v\langle V,W\rangle - W(e^\beta\eta(V))-\langle V,[\partial_v,W]\rangle +\langle \partial_v,[W,V]\rangle+\langle W,[V,\partial_v]\rangle\\
&=\s\nabla_V(e^\beta\eta)(W) - \s\nabla_W(e^\beta\eta)(V)\\
&=\text{curl}(e^\beta\eta)(V,W)
\end{align*}
so that a trace over $V,W$ yields $\tr\chi = e^{2\beta}(\mathfrak{h}+\alpha+|\vec\eta|^2)\tr\ubar\chi-2\s\nabla\cdot(e^\beta\eta)$ and therefore
\begin{align*}
\langle \vec{H},\vec{H}\rangle &= e^{2\beta}(\mathfrak{h}+\alpha+|\vec\eta|^2)(\tr\ubar\chi)^2-2\s\nabla\cdot(e^\beta\eta)\tr\ubar\chi\\
&=\Big(1-\frac{2M}{r}+\delta\frac{\alpha_0}{r}\Big)\Big(\frac{2}{r}+\frac{\delta\ubar\theta}{r^2}\Big)^2+\delta o_2^X(r^{-3})\\
&=\Big(1-\frac{2M}{r}+\delta\frac{\alpha_0}{r}\Big)\Big(\frac{4}{r^2}+\frac{4\delta\ubar\theta}{r^3}\Big)+\delta o_2^X(r^{-3})\\
&=\frac{4}{r^2}\Big(1-\frac{2m_0}{r}+\delta\frac{\ubar\theta}{r}+\delta\frac{\alpha_0}{r}\Big)+\delta o_2^X(r^{-3})
\end{align*}
from decay conditions 2-5.
For $\zeta$ we have
\begin{align*}
\zeta(V)&=\langle D_V(e^{-\beta}\partial_r),\partial_v\rangle-e^\beta\langle D_V(e^{-\beta}\partial_r),\vec\eta\rangle\\
&=-V\beta+e^{-\beta}\langle D_V\partial_r,\partial_v\rangle-\langle D_V\partial_r,\vec\eta\rangle\\
&=-V\beta+e^{-\beta}\langle D_V\partial_r,\partial_v\rangle-e^\beta\ubar\chi(V,\vec\eta).
\end{align*}
From the Koszul formula 
\begin{align*}
2\langle D_V\partial_r,\partial_v\rangle&=V\langle\partial_r,\partial_v\rangle+\partial_r\langle V,\partial_v\rangle-\partial_v\langle V,\partial_r\rangle-\langle V,[\partial_r,\partial_v]\rangle+\langle \partial_r,[\partial_v,V]\rangle+\langle \partial_v,[V,\partial_r]\rangle\\
&=e^\beta V\beta+\partial_r(e^\beta\eta(V))\\
&=e^\beta V\beta + \mathcal{L}_{\partial_r}(e^\beta\eta)(V)
\end{align*}
from which we conclude that $\zeta(V) = -\frac12V(\beta)+\frac{e^{-\beta}}{2}\mathcal{L}_{\partial_r}(e^\beta\eta)(V)-e^\beta\ubar\chi(V,\vec\eta)$ and
\begin{align*}
\s\nabla\cdot\zeta &= -\frac12\s\Delta\beta+\frac12\s\nabla\cdot(e^{-\beta}\mathcal{L}_{\partial_r}(e^\beta\eta))-\s\nabla\cdot(e^\beta\ubar\chi(\vec\eta))\\
&=-\frac12\s\Delta\beta+\frac12\s\nabla\cdot(\beta_r\eta)+\frac12\s\nabla\cdot(\mathcal{L}_{\partial_r}\eta)
-e^\beta\ubar\chi(\s\nabla\beta,\vec\eta)-e^\beta\s\nabla\cdot(\ubar\chi(\vec\eta))\\
&=\delta o_2^X(r^{-3})
\end{align*}
having used decay conditions 3, 5 and Lemma \ref{l19} for the final line.
\begin{lemma}\label{l20}
$\Omega$ satisfies conditions 1, 2 and 3 of Definition \ref{d7}. $\Omega$ additionally satisfies strong flux decay if and only if 
$$\frac12\mathring\nabla\cdot\gamma_1+d\ubar\theta=0$$
for $\ubar\theta = -\frac12\mathring{\tr}\gamma_1$ and is subsequently past asymptotically flat.
\end{lemma}
\begin{proof}
Having already verified condition 1 up to strong decay for $\gamma_s$ of our geodesic foliation $\{\Sigma_s\}$ we continue to show conditions 2 and 3.\\
Given $V\in E_{\partial_r}(\Sigma_{r_0})$ Lemma \ref{l12} ensures $V-Vs\ubar L|_{\Sigma_s}\in \Gamma(T\Sigma_s)$ and we see that 
\begin{align*}
[V-Vs\ubar L, \ubar L]&=[V,\ubar L]+\ubar LVs\ubar L\\
&=e^{\beta}V(e^{-\beta})\ubar L+e^{-\beta}V(\partial_rs) \ubar L\\
&=(e^{\beta}V(e^{-\beta})+e^{-\beta}V(e^{\beta}))\ubar L\\
&=0.
\end{align*}
So $V-Vs\ubar L\in E(\Sigma_0)$ and Lemma \ref{l12} gives 
\begin{align*}
t(V-Vs\ubar L) = t(V) &= \zeta(V)+\ubar\chi(V,\s\nabla s)= -\frac12V(\beta)+\frac12\beta_r\eta(V)+\frac12(\mathcal{L}_{\partial_r}\eta)(V) - e^\beta\ubar\chi(V,\vec\eta)+\ubar\chi(V,\s\nabla s)\\
&=(\delta o_3^X(r^{-1})\cap o_1(r^{-1}))(V)+\ubar\chi(V,\s\nabla s)\\
&=(\delta o_3^X(r^{-1})\cap o_1(r^{-1}))(V)+re^{-\beta}\mathring\gamma(V,\frac{1}{r^2}\mathring\nabla s)\\
&=\frac{e^{-\beta}}{r}V\beta_0+(\delta o_3^X(r^{-1})\cap o_1(r^{-1}))(V)
\end{align*}
having used decay conditions 3 and 5 to get the second line, Lemma \ref{l19} for the third and (26) for the last. Moreover,
\begin{align*}
(\mathcal{L}_{V-Vs\ubar L}t)(W-Ws\ubar L)&= (V-Vs\ubar L)(t(W-Ws\ubar L)) - t([V,W])\\
&=(\mathcal{L}_Vt)(W)-Vs\ubar L(t(W))\\
&=(\mathcal{L}_V-e^{-\beta}Vs\mathcal{L}_{\partial_r})(\frac{d\beta_0}{r})(W)+o(r^{-1})\\
&=\frac{1}{r}\mathcal{L}_V(d\beta_0)(W)+o(r^{-1})\\
&=\frac{1}{r}(\mathcal{L}_{V-Vs\ubar L}d\beta_0)(W-Ws\ubar L)+o(r^{-1})
\end{align*}
where the last line follows since $\beta_0$ is $\ubar L$-Lie constant. With a basis extension $\{X_i\}\subset E(\Sigma_0)$ we therefore conclude that $\mathcal{L}_{X_i}t=\frac{1}{s}\mathcal{L}_{X_i}d\beta_0 + o(s^{-1})$ so that condition 2 for asymptotic flatness is satisfied up to strong decay with $t_1 = d\beta_0$. From Proposition \ref{p6} and (27):
\begin{align*}
\tr K &=\frac{2}{s}-\frac{1}{2s^2}\mathring{\tr}\Gamma_1+o(s^{-2}) \\
&=\frac{2}{s}-\frac{1}{2s^2}\mathring{\tr}(2c_0\mathring\gamma+\delta\gamma_1)+o(s^{-2})\\
&=\frac{2}{s}+\frac{\delta\ubar\theta-2c_0}{s^2}+o(s^{-2})
\end{align*}
and
\begin{align*}
\hat{K} &=K - \frac12\tr K\gamma_s\\
&= s\mathring\gamma+\frac12\Gamma_1-\frac12\Big(\frac{2}{s}+\frac{\delta\ubar\theta-2c_0}{s^2}+o(s^{-2})\Big)\gamma_s+o(1)\\
&=s\mathring\gamma+\frac12(2c_0\mathring\gamma+\delta\gamma_1)-\frac12(\frac{2}{s}+\frac{\delta\ubar\theta-2c_0}{s^2})(s^2\mathring\gamma+s(2c_0\mathring\gamma+\delta\gamma_1))+o(1)\\
&=-\frac{\delta}{2}(\gamma_1+\ubar\theta\mathring\gamma)+o(1).
\end{align*}
For condition 3 we take $r|_{\Sigma_s}\in \mathcal{F}(\Sigma_s)$ and Lie drag it to the the rest of $\Omega$ along $\partial_r$ (hence $\ubar L$) to give $r_s\in \mathcal{F}(\Omega)$. Using Lemma \ref{l12} from the vantage point of the cross section $\Sigma_s$ amongst the background foliation $\{\Sigma_r\}$:
$$e^{-\beta}\tr Q = e^{-\beta}\tr\chi-4(\zeta+\s{d}\log e^\beta)(\s\nabla r_s) - 2\Delta r_s+|\s\nabla r_s|^2e^\beta\tr\ubar\chi-2\beta_r|\s\nabla r_s|^2$$
From the expression of $r(s)$ in 5.4.1, recalling Remark \ref{r7}, we see $dr_s = -d\beta_0+o(1)$ from which Lemma \ref{l16} implies that $\Delta r_s = -\frac{1}{s^2}\mathring{\Delta}\beta_0+o(s^{-2})$. From decay conditions 3, 5 and Lemma \ref{l19} we have
\begin{align*}
\tr Q&=\tr\chi|_{\Sigma_s}+2\frac{\mathring\Delta\beta_0}{s^2}+o(s^{-2})\\
&=\Big(e^{2\beta}(\mathfrak{h}+\alpha+|\vec\eta|^2)\tr\ubar\chi-2\s\nabla\cdot(e^{\beta}\eta)\Big)|_{\Sigma_s}+2\frac{\mathring\Delta\beta_0}{s^2}+o(s^{-2})\\
&=(\frac{2}{s}+\frac{\delta\ubar\theta-2c_0}{s^2})(1-\frac{2M}{s}+\delta\frac{\alpha_0}{s})+2\frac{\mathring\Delta\beta_0}{s^2}+o(s^{-2})\\
&=\frac{2}{s}+\frac{\delta\ubar\theta-2c_0-4M+2\delta\alpha_0}{s^2}+2\frac{\mathring\Delta\beta_0}{s^2}+o(s^{-2})\\
&=\frac{2}{s}-2\frac{c_0+2M}{s^2}+\frac{1}{s^2}(\delta\ubar\theta+2\mathring\Delta\beta_0+2\delta\alpha_0)+o(s^{-2})
\end{align*}
and condition 3 follows as soon as we set $M = m_0+\delta o_2^X(1)$. As in the spherically symmetric case the highest order term for $\tr Q$ agrees with $\frac{2\mathring{\mathcal{K}}}{s}$ where $\mathring{\mathcal{K}}=1$ is the Gaussian curvature of $\mathring\gamma$. We recall that our use of condition 4 depends on whether $\Omega$ has strong flux decay (Remark \ref{r10}). From Proposition \ref{p6} and (27) we will have strong flux decay if and only if
\begin{align*}
d\beta_0 = t_1&=\frac12\mathring\nabla\cdot\Gamma_1-\frac12\s{d}\mathring{\tr}\Gamma_1 \\
&= \frac12\mathring\nabla\cdot(2c_0\mathring\gamma+\delta\gamma_1)+\s{d}(\delta\ubar\theta-2c_0)\\
&=\frac12{d}(-2\beta_0)+\frac{\delta}{2}\mathring\nabla\cdot\gamma_1+\delta d\ubar\theta+2d\beta_0\\
&=d\beta_0+\delta(\frac12\mathring\nabla\cdot\gamma_1+d\ubar\theta)
\end{align*}
which in turn holds if and only if $\frac12\mathring\nabla\cdot\gamma_1 +d\ubar\theta = 0$.
\end{proof}
Henceforth we will adopt the conditions of Lemma \ref{l20} for $\Omega$. From Proposition \ref{p7} 
\begin{align*}
\mathcal{K}_{r^2\gamma}&=\frac{1}{r^2}+\frac{\delta}{r^3}\Big(\ubar\theta
+\frac12\mathring\nabla\cdot\mathring\nabla\cdot\gamma_1+\mathring\Delta\ubar\theta\Big)
+\delta o_4^X(r^{-3})\\
&=\frac{1}{r^2}+\frac{\delta}{r^3}\ubar\theta+\delta o_4^X(r^{-3}).
\end{align*}
From Lemma \ref{l19} we have
$$\s\nabla_i\s\nabla_j\ubar\chi_{mn} = -\frac{\delta}{2}\mathring\nabla_i\mathring\nabla_j{\gamma_1}_{mn}+\delta o_2^X(1)$$
so that contraction with $\gamma(r)^{-1}$ first in $mn$ then $ij$ gives
$$\s\Delta\tr\ubar\chi = \frac{\delta}{r^4}\mathring\Delta\ubar\theta+\delta o_2^X(r^{-4})$$ which we use in $\s\Delta\log\tr\ubar\chi = \frac{\s\Delta\tr\ubar\chi}{\tr\ubar\chi} - \frac{|\s\nabla\tr\ubar\chi|^2}{(\tr\ubar\chi)^2}$ to conclude
$$\s\Delta\log\tr\ubar\chi = \frac{\delta}{2r^3}\mathring\Delta\ubar\theta+\delta o_2^X(r^{-3}).$$
Finally we have $\rho$
\begin{align*}
\rho&=\mathcal{K}_{r^2\gamma}-\frac{1}{4}\langle\vec{H},\vec{H}\rangle+\s\nabla\cdot\zeta-\s\Delta\log\tr\ubar\chi\\
&=\frac{1}{r^2}+\frac{\delta}{r^3}\ubar\theta-\frac{1}{r^2}+\frac{2m_0}{r^3}-\delta\frac{\ubar\theta}{r^3}-\delta\frac{\alpha_0}{r^3}-\frac{\delta}{2r^3}\mathring\Delta\ubar\theta+\delta o_2^X(r^{-3})\\
&=\frac{2m_0}{r^3}-\frac{\delta}{r^3}(\frac12\mathring\Delta\ubar\theta+\alpha_0)+\delta o_2^X(r^{-3})\\
&=\frac{2m_0}{r^3}-\frac{\delta}{r^3}(\frac12\mathring\Delta\ubar\theta+\alpha_0)+\delta o_2^X(r^{-3})
\end{align*}
and
\begin{align*}
\frac14\langle\vec{H},\vec{H}\rangle-\frac13\s\Delta\log\rho&=\frac{1}{r^2}\Big(1-\frac{2m_0}{r}+\delta\frac{\ubar\theta}{r}+\delta\frac{\alpha_0}{r}\Big)\\
&\indent-\frac13\s\Delta\log\Big(\frac{2m_0}{r^3}-\frac{\delta}{r^3}(\frac12 \mathring\Delta\ubar\theta+\alpha_0)+\delta o_2^X(r^{-3})\Big)+\delta o_2^X(r^{-3}).
\end{align*}
We may now use Lemma \ref{l16} to decompose the last term 
\begin{align*}
\s\Delta\log\Big(\frac{2m_0}{r^3}-\frac{\delta}{r^3}(\frac52 \mathring\Delta\ubar\theta+\alpha_0)+\delta o_2^X(r^{-3})\Big)&=\frac{1}{r^2}\mathring\Delta\log\Big(1-\frac{\delta}{2m_0}(\frac12\mathring\Delta\ubar\theta+\alpha_0)+\delta o_2^X(1)\Big)+\delta o(r^{-2})\\
&=\frac{1}{r^2}\mathring\Delta\log\Big(1-\frac{\delta}{2m_0}(\frac12\mathring\Delta\ubar\theta+\alpha_0)\Big)+\delta o(r^{-2})
\end{align*}
giving
$$\frac14\langle\vec{H},\vec{H}\rangle-\frac13\s\Delta\log\rho=\frac{1}{r^2}\Big(1-\frac{2m_0}{r}-\frac13\mathring\Delta\log\Big(1-\frac{\delta}{2m_0}(\frac12\mathring\Delta\ubar\theta+\alpha_0)\Big)\Big)
+\delta o(r^{-2}).$$
Since $m_0>0$ we notice for sufficiently small $\delta$ our perturbation ensures $\rho>0$ for all $r>0$. However, from our construction so far it's not yet possible to conclude that some $\delta>0$ will enforce $\frac14\langle\vec{H},\vec{H}\rangle\geq \frac13\s{\Delta}\log\rho$ along the foliation. Moreover, the existence of a horizon ($\tr\chi = 0$) is equally questionable. 
\subsubsection{Smoothing to Spherical Symmetry}
We will solve this difficulty by `smoothing' away all perturbations in a neighborhood of the (desired) horizon in order to obtain spherical symmetry on $r<r_1$ for some $r_1>0$ yet to be chosen. The resulting spherical symmetry will uncover the horizon at $r=r_0<r_1$ and will also provide a choice of $\delta>0$ so that $\frac14\langle\vec{H},\vec{H}\rangle>\frac13\s{\Delta}\log\rho$ away from it, causing the foliation $\{\Sigma_r\}$ to be an (SP)-foliation.\\\\
\indent We will use a smooth step function $0\leq S_\delta(r)\leq 1$ such that $S_\delta(r)=0$ for $r<r_1$ and $S_\delta(r)=1$ for $r>r_2$ for some finite $r_2(\delta)$ chosen to ensure $|S_\delta'(r)|\leq \delta$. 
By first choosing \textit{parameter functions} for the desired spherically symmetric region; $\tilde\beta(v,r)$ and $0<\tilde{M}(v,r)=m_0+o(1)$ such that $r_0 = 2\tilde{M}(v_0,r_0)$ and $2\tilde{M}(v_0,r)<r$ for $r>r_0$ we induce spherical symmetry on $r<r_1$ with the following substitutions:
\begin{align*}
\tilde\gamma&\to \delta r(S_\delta(r)-1)\gamma_1+S_\delta(r)\tilde\gamma\\
\beta(r,\vartheta,\varphi)&\to S_\delta(r)\beta(r,\vartheta,\varphi)+(1-S_\delta(r))\tilde\beta(v_0,r)\\
M(r,\vartheta,\varphi)&\to S_\delta(r) M(r,\vartheta,\varphi)+(1-S_\delta(r))\tilde M(v_0,r)\\
\tilde\alpha&\to S_\delta(r)\tilde\alpha-(1-S_\delta(r))\frac{\delta\alpha_0}{r}\\
\eta&\to S_\delta(r)\eta.
\end{align*}
We leave the reader the simple verification that these changes to our perturbation tensors $\tilde\gamma$, $\beta$, $M$, $\tilde\alpha$ and $\eta$ maintain the decay conditions 1-5. Clearly for $r>r_2$ our substitutions leave the metric unchanged while inducing spherical symmetry on $r<r_1$ with the spherical parameter functions $\tilde\beta\,,\tilde M$:

\begin{figure}[h]
\centering
\def\svgwidth{200pt} 
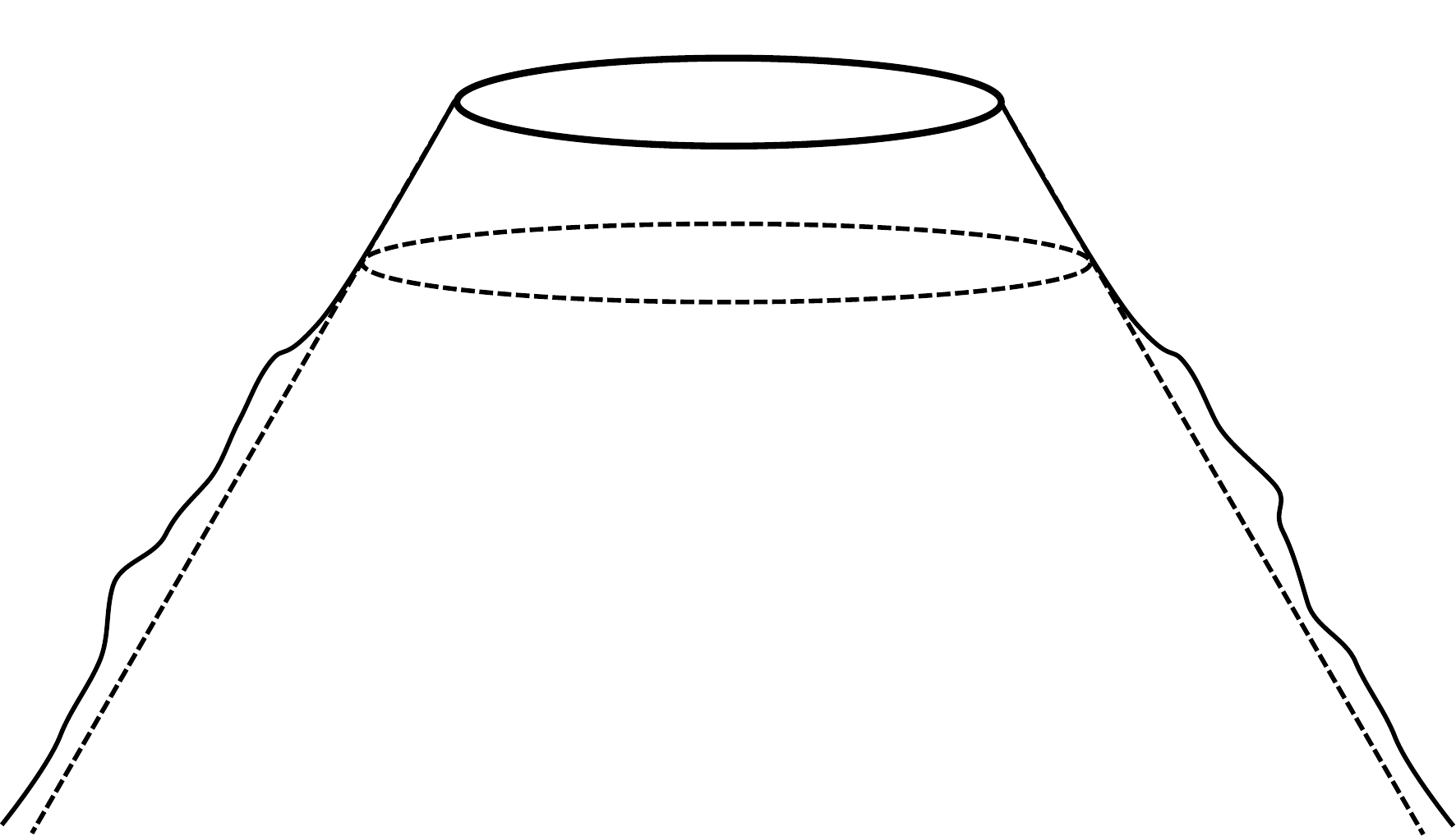
\end{figure}

 An example $S_\delta(r)$ is given by the function
\[S_\delta(r)= \begin{cases} 
      0 & r\leq r_1 \\
      \frac{e^{\frac{k}{r_1-r}}}{e^{\frac{k}{r_1-r}}+e^{\frac{k}{r-r_2}}} & r_1< r< r_2 \\
      1 & r_2\leq r 
   \end{cases}
\]
where $k=\frac{4e^4}{\delta}$ and $r_2(\delta) = r_1+k$. Since $S_\delta(r) = P(\frac{1}{r_1-r}+\frac{1}{r_2-r})$ for $P(r) = \frac{e^{kr}}{1+e^{kr}}$ satisfying the logistic equation
$$P'(r) = kP(1-P)$$
we have
\begin{align*}
S_\delta'(r) &= kS_\delta(r)(1-S_\delta(r))(\frac{1}{(r-r_1)^2}+\frac{1}{(r-r_2)^2})\\
&=k\frac{S_\delta(r)}{(r-r_1)^2}(1-S_\delta(r))+kS_\delta(r)\frac{1-S_\delta(r)}{(r-r_2)^2}\\
&\leq k\Big(\frac{S_\delta(r)}{(r-r_1)^2}+\frac{1-S_\delta(r)}{(r-r_2)^2}\Big).
\end{align*}
Elementary analysis reveals on the interval $r_1<r<r_2$ that 
\begin{align*}
0&\leq \frac{e^{\frac{k}{r_1-r}}}{(r-r_1)^2}\leq\frac{4e^2}{k^2}\\
0&\leq \frac{1}{e^{\frac{k}{r_1-r}}+e^{\frac{k}{r-r_2}}}\leq \frac{1}{2}e^{2}
\end{align*} 
yielding from simple symmetry arguments that both $\frac{S_\delta(r)}{(r-r_1)^2},\,\,\frac{1-S_\delta(r)}{(r-r_2)^2}\leq\frac{2e^4}{k^2}$ and therefore
$$0\leq S'_\delta(r)\leq k\frac{4e^4}{k^2} = \delta$$
as desired.
Denoting $m(r,\delta):= S_\delta(r)m_0+(1-S_\delta(r))\tilde M(r)$ the new metric gives
\[\rho= \begin{cases} 
      \frac{2\tilde M(v_0,r)}{r^3},& r < r_1 \\
       \frac{2m(r,\delta)}{r^3}-\frac{\delta}{r^3}(\frac12\mathring\Delta\ubar\theta+\alpha_0)+\delta o_2^X(r^{-3}), &r_1\leq r\leq r_2 \\
      \frac{2m_0}{r^3}-\frac{\delta}{r^3}(\frac12\mathring\Delta\ubar\theta+\alpha_0)+\delta o_2^X(r^{-3}), & r_2< r
   \end{cases}
\]
and
\[\frac14\langle\vec{H},\vec{H}\rangle-\frac13\s\Delta\log\rho= \begin{cases} 
      \frac{1}{r^2}(1-\frac{2\tilde M(v_0,r)}{r}), & r< r_1 \\
      \frac{1}{r^2}\Big(1-\frac{2m(r,\delta)}{r}-\frac13\mathring\Delta\log\Big(1-\frac{\delta}{2m(r,\delta)}(\frac12\mathring\Delta\ubar\theta+\alpha_0)\Big)\Big) 
+\delta o(r^{-2}),&r_1\leq r\leq r_2\\
      \frac{1}{r^2}\Big(1-\frac{2m_0}{r}-\frac13\mathring\Delta\log\Big(1-\frac{\delta}{2m_0}(\frac12\mathring\Delta\ubar\theta+\alpha_0)\Big)\Big)
+\delta o(r^{-2}), & r_2< r.
   \end{cases}
\]
Since $C(r_1)\geq m(r,\delta)\geq m_0$ for $C(r_1):=\sup_{r_0<r<r_1}\tilde M$ we see for any choice of $r_1>C(r_1)$ (which is possible since $\tilde M = m_0+o(1)$) and sufficiently small $\delta$ the foliation $\{\Sigma_r\}$ satisfies property (SP). If we therefore restrict to perturbations satisfying the null energy condition on $\Omega$ then Theorem \ref{t2} implies the following:
\begin{theorem}\label{t9}
Let $g_\delta$ be a metric perturbation off of spherical symmetry given by
$$g_\delta = -(\mathfrak{h}+\alpha)e^{2\beta}dv\otimes dv+e^{\beta}(dv\otimes(dr+\eta)+(dr+\eta)\otimes dv)+r^2\gamma$$
where
\begin{enumerate}
\item $r^2\gamma = r^2\mathring\gamma+r\delta\gamma_1+\tilde\gamma$ is trasversal with $\mathring\gamma$ the transversal $\partial_r$-Lie constant round metric on $\mathbb{S}^2$ independent of $\delta$, $\gamma_1$ a transversal $\partial_r$-Lie constant 2-tensor independent of $\delta$ satisfying $\mathring\nabla\cdot\gamma_1=d(\mathring{\tr}\gamma_1)$ and $(\mathcal{L}_{\partial_r})^i\tilde\gamma = \delta o_{5-i}^X(r^{1-i})$ for $0\leq i\leq 3$.
\item $\alpha=\delta\frac{\alpha_0}{r}+\tilde\alpha$ where $\alpha_0$ is $\partial_r$-constant, independent of $\delta$ and $|\tilde\alpha|_{\mathring{H}^2}\leq \delta h_1(r)$ for $h_1=o(r^{-1})$
\item $\beta$ satisfies:
\begin{enumerate}
\item $|\beta| = o_2(r^{-1})$ is  r-integrable
\item $|\mathring\nabla\beta|_{\mathring{H}^3}\leq \delta h_2(r)$ for some integrable $h_2 = o(r^{-1})$
\item $|\mathring\nabla\beta_r|_{\mathring{H}^2}=O(r^{-1})$
\end{enumerate}
\item $M=m_0+\tilde m$ where $m_0>0$ is constant, independent of $\delta$ and $|\tilde m|_{\mathring{H}^2}\leq \delta h_3(r)$ for $h_3 = o(1)$
\item $\eta$ is a transversal 1-form satisfying:
\begin{enumerate}
\item $\eta = o_2(1)$
\item $|\eta|_{\mathring{H}^3}+r|\mathcal{L}_{\partial_r}\tilde\eta|_{\mathring{H}^3}\leq\delta h_4(r)$ for $h_4=o(1)$.
\end{enumerate}\end{enumerate}
Then for sufficiently small $\delta$, $\Omega:=\{v=v_0\}$ is past asymptotically flat with strong flux decay. 
In addition, for any choice of spherical parameters $\tilde\beta(v,r)$ and $\tilde{M}(v,r)$ such that $0<\tilde{M}(v_0,r)=m_0+o(1)$, $r_0 = 2\tilde{M}(v_0,r_0)$ and $2\tilde{M}(v_0,r)<r$ for $r>r_0$, smoothing to spherical symmetry with the step function $S_\delta(r)$ (as above) according to:
\begin{align*}
\tilde\gamma&\to \delta r(S_\delta(r)-1)\gamma_1+S_\delta(r)\tilde\gamma\\
\beta(r,\vartheta,\varphi)&\to S_\delta(r)\beta(r,\vartheta,\varphi)+(1-S_\delta(r))\tilde\beta(r)\\
M(r,\vartheta,\varphi)&\to S_\delta(r) M(r,\vartheta,\varphi)+(1-S_\delta(r))\tilde M(r)\\
\tilde\alpha&\to S_\delta(r)\tilde\alpha-(1-S_\delta(r))\frac{\delta\alpha_0}{r}\\
\eta&\to S_\delta(r)\eta
\end{align*}
we have that $\Sigma:=\{r_0 = 2\tilde{M}(v_0,r_0)\}$ is marginally outer trapped and the coordinate spheres $\{\Sigma_r\}_{r\geq r_0}$ form an (SP)-foliation. Moreover, if $g_\delta$ respects the null energy condition on $\Omega$ we have the Penrose inequality:
$$\sqrt{\frac{|\Sigma|}{16\pi}}\leq m_B$$
where $m_B$ is the Bondi mass of $\Omega$. 
\end{theorem}
\section*{Acknowledgments}
The author would like to deeply thank Hubert L. Bray for his supervision and support during the development of these ideas as well as Marc Mars for insightful conversation and his invaluable comments upon a close reading of this paper. The author is also very appreciative for the financial support for his last year of graduate school and for travel to conferences provided by NSF grant DMS-1406396.

\end{document}

%% file: bundle.pdf_tex
\begingroup%
  \makeatletter%
  \providecommand\color[2][]{%
    \errmessage{(Inkscape) Color is used for the text in Inkscape, but the package 'color.sty' is not loaded}%
    \renewcommand\color[2][]{}%
  }%
  \providecommand\transparent[1]{%
    \errmessage{(Inkscape) Transparency is used (non-zero) for the text in Inkscape, but the package 'transparent.sty' is not loaded}%
    \renewcommand\transparent[1]{}%
  }%
  \providecommand\rotatebox[2]{#2}%
  \ifx\svgwidth\undefined%
    \setlength{\unitlength}{710.57987243bp}%
    \ifx\svgscale\undefined%
      \relax%
    \else%
      \setlength{\unitlength}{\unitlength * \real{\svgscale}}%
    \fi%
  \else%
    \setlength{\unitlength}{\svgwidth}%
  \fi%
  \global\let\svgwidth\undefined%
  \global\let\svgscale\undefined%
  \makeatother%
  \begin{picture}(1,0.37076228)%
    \put(0,0){\includegraphics[width=\unitlength,page=1]{bundle.pdf}}%
    \put(0.63175416,0.09156221){\color[rgb]{0,0,0}\makebox(0,0)[lb]{\smash{}}}%
    \put(0.63573466,0.14671983){\color[rgb]{0,0,0}\makebox(0,0)[lb]{\smash{$\vec{H}$}}}%
    \put(0.88365975,0.26385876){\color[rgb]{0,0,0}\makebox(0,0)[lt]{\begin{minipage}{0.07051074\unitlength}\raggedright $L$\end{minipage}}}%
    \put(0.88024795,0.06483629){\color[rgb]{0,0,0}\makebox(0,0)[lt]{\begin{minipage}{0.12964887\unitlength}\raggedright $\ubar L$\end{minipage}}}%
    \put(0.6209501,0.26499606){\color[rgb]{0,0,0}\makebox(0,0)[lb]{\smash{$T^{\perp}\Sigma$}}}%
    \put(0,0){\includegraphics[width=\unitlength,page=2]{bundle.pdf}}%
    \put(0.02843177,0.31844777){\color[rgb]{0,0,0}\makebox(0,0)[lb]{\smash{$\large{M}$}}}%
    \put(0.31161231,0.06142452){\color[rgb]{0,0,0}\makebox(0,0)[lb]{\smash{$V,W\in \Gamma(T\Sigma)$}}}%
    \put(0.08188354,0.17968882){\color[rgb]{0,0,0}\makebox(0,0)[lb]{\smash{$\large{\Sigma}$}}}%
  \end{picture}%
\endgroup%

%% file: cone1.pdf_tex
\begingroup%
  \makeatletter%
  \providecommand\color[2][]{%
    \errmessage{(Inkscape) Color is used for the text in Inkscape, but the package 'color.sty' is not loaded}%
    \renewcommand\color[2][]{}%
  }%
  \providecommand\transparent[1]{%
    \errmessage{(Inkscape) Transparency is used (non-zero) for the text in Inkscape, but the package 'transparent.sty' is not loaded}%
    \renewcommand\transparent[1]{}%
  }%
  \providecommand\rotatebox[2]{#2}%
  \ifx\svgwidth\undefined%
    \setlength{\unitlength}{488.40206116bp}%
    \ifx\svgscale\undefined%
      \relax%
    \else%
      \setlength{\unitlength}{\unitlength * \real{\svgscale}}%
    \fi%
  \else%
    \setlength{\unitlength}{\svgwidth}%
  \fi%
  \global\let\svgwidth\undefined%
  \global\let\svgscale\undefined%
  \makeatother%
  \begin{picture}(1,0.43799821)%
    \put(0,0){\includegraphics[width=\unitlength,page=1]{cone1.pdf}}%
    \put(0.68379887,0.39316251){\color[rgb]{0,0,0}\makebox(0,0)[lb]{\smash{$\Sigma_0$}}}%
    \put(0.82536825,0.14629327){\color[rgb]{0,0,0}\makebox(0,0)[lb]{\smash{$\Sigma_s$}}}%
    \put(0.03679087,0.36742251){\color[rgb]{0,0,0}\makebox(0,0)[lb]{\smash{$\Omega$}}}%
    \put(0,0){\includegraphics[width=\unitlength,page=2]{cone1.pdf}}%
    \put(0.30589008,0.30307283){\color[rgb]{0,0,0}\makebox(0,0)[lb]{\smash{$\ubar L$}}}%
    \put(0.59604908,0.31711276){\color[rgb]{0,0,0}\makebox(0,0)[lb]{\smash{$\ubar L$}}}%
    \put(0,0){\includegraphics[width=\unitlength,page=3]{cone1.pdf}}%
  \end{picture}%
\endgroup%

%% file: figures.pdf_tex
\begingroup%
  \makeatletter%
  \providecommand\color[2][]{%
    \errmessage{(Inkscape) Color is used for the text in Inkscape, but the package 'color.sty' is not loaded}%
    \renewcommand\color[2][]{}%
  }%
  \providecommand\transparent[1]{%
    \errmessage{(Inkscape) Transparency is used (non-zero) for the text in Inkscape, but the package 'transparent.sty' is not loaded}%
    \renewcommand\transparent[1]{}%
  }%
  \providecommand\rotatebox[2]{#2}%
  \ifx\svgwidth\undefined%
    \setlength{\unitlength}{833.91109437bp}%
    \ifx\svgscale\undefined%
      \relax%
    \else%
      \setlength{\unitlength}{\unitlength * \real{\svgscale}}%
    \fi%
  \else%
    \setlength{\unitlength}{\svgwidth}%
  \fi%
  \global\let\svgwidth\undefined%
  \global\let\svgscale\undefined%
  \makeatother%
  \begin{picture}(1,0.42532268)%
    \put(0,0){\includegraphics[width=\unitlength,page=1]{figures.pdf}}%
    \put(0.04830255,0.34016583){\color[rgb]{0,0,0}\makebox(0,0)[lb]{\smash{$P_2=(E_2,\vec{p}_2)$}}}%
    \put(-0.02205028,0.06212693){\color[rgb]{0,0,0}\makebox(0,0)[lb]{\smash{$P_1=(E_1,\vec{p}_1)$}}}%
    \put(0,0){\includegraphics[width=\unitlength,page=2]{figures.pdf}}%
    \put(0.75575038,0.14145446){\color[rgb]{0,0,0}\makebox(0,0)[lb]{\smash{$P$}}}%
    \put(0.57852935,0.15731377){\color[rgb]{0,0,0}\makebox(0,0)[lb]{\smash{$P'$}}}%
    \put(0.17379527,0.30092069){\color[rgb]{0,0,0}\makebox(0,0)[lb]{\smash{}}}%
    \put(0.30592822,0.18719635){\color[rgb]{0,0,0}\makebox(0,0)[lb]{\smash{$P_3=(E_3,\vec{p}_3)$}}}%
  \end{picture}%
\endgroup%

%% file: hyperc.pdf_tex
\begingroup%
  \makeatletter%
  \providecommand\color[2][]{%
    \errmessage{(Inkscape) Color is used for the text in Inkscape, but the package 'color.sty' is not loaded}%
    \renewcommand\color[2][]{}%
  }%
  \providecommand\transparent[1]{%
    \errmessage{(Inkscape) Transparency is used (non-zero) for the text in Inkscape, but the package 'transparent.sty' is not loaded}%
    \renewcommand\transparent[1]{}%
  }%
  \providecommand\rotatebox[2]{#2}%
  \ifx\svgwidth\undefined%
    \setlength{\unitlength}{749.61974617bp}%
    \ifx\svgscale\undefined%
      \relax%
    \else%
      \setlength{\unitlength}{\unitlength * \real{\svgscale}}%
    \fi%
  \else%
    \setlength{\unitlength}{\svgwidth}%
  \fi%
  \global\let\svgwidth\undefined%
  \global\let\svgscale\undefined%
  \makeatother%
  \begin{picture}(1,0.59377211)%
    \put(0,0){\includegraphics[width=\unitlength,page=1]{hyperc.pdf}}%
    \put(0.38507463,0.41077421){\color[rgb]{0,0,0}\makebox(0,0)[lb]{\smash{$\vec{q}$}}}%
    \put(0.13247446,0.00732478){\color[rgb]{0,0,0}\makebox(0,0)[lb]{\smash{$\vec{o}$}}}%
    \put(0.65714299,0.15070443){\color[rgb]{0,0,0}\makebox(0,0)[lb]{\smash{$o$}}}%
    \put(0.8984396,0.3840847){\color[rgb]{0,0,0}\makebox(0,0)[lb]{\smash{$q$}}}%
    \put(0.82899827,0.44095936){\color[rgb]{0,0,0}\makebox(0,0)[lb]{\smash{$U_q$}}}%
    \put(0.29830805,0.40743169){\color[rgb]{0,0,0}\makebox(0,0)[lb]{\smash{$V_{\vec{q}}$}}}%
    \put(0.15306561,0.55307577){\color[rgb]{0,0,0}\makebox(0,0)[lb]{\smash{$H_C$}}}%
    \put(0.55830128,0.2188826){\color[rgb]{0,0,0}\makebox(0,0)[lb]{\smash{$M$}}}%
    \put(0,0){\includegraphics[width=\unitlength,page=2]{hyperc.pdf}}%
    \put(0.02433201,0.42742308){\color[rgb]{0,0,0}\makebox(0,0)[lb]{\smash{$\vec{N}$}}}%
    \put(0,0){\includegraphics[width=\unitlength,page=3]{hyperc.pdf}}%
    \put(0.89256031,0.23844662){\color[rgb]{0,0,0}\makebox(0,0)[lb]{\smash{$\frac{k\ubar L}{\tr\ubar\chi}$}}}%
  \end{picture}%
\endgroup%

%% file: cone1a.pdf_tex
\begingroup%
  \makeatletter%
  \providecommand\color[2][]{%
    \errmessage{(Inkscape) Color is used for the text in Inkscape, but the package 'color.sty' is not loaded}%
    \renewcommand\color[2][]{}%
  }%
  \providecommand\transparent[1]{%
    \errmessage{(Inkscape) Transparency is used (non-zero) for the text in Inkscape, but the package 'transparent.sty' is not loaded}%
    \renewcommand\transparent[1]{}%
  }%
  \providecommand\rotatebox[2]{#2}%
  \ifx\svgwidth\undefined%
    \setlength{\unitlength}{549.51381413bp}%
    \ifx\svgscale\undefined%
      \relax%
    \else%
      \setlength{\unitlength}{\unitlength * \real{\svgscale}}%
    \fi%
  \else%
    \setlength{\unitlength}{\svgwidth}%
  \fi%
  \global\let\svgwidth\undefined%
  \global\let\svgscale\undefined%
  \makeatother%
  \begin{picture}(1,0.57667607)%
    \put(0,0){\includegraphics[width=\unitlength,page=1]{cone1a.pdf}}%
    \put(0.69305375,0.54252673){\color[rgb]{0,0,0}\makebox(0,0)[lb]{\smash{$r=0$}}}%
    \put(0.62232496,0.44459619){\color[rgb]{0,0,0}\makebox(0,0)[lb]{\smash{\tiny{$r=2M$}}}}%
    \put(0.6294262,0.31935675){\color[rgb]{0,0,0}\makebox(0,0)[lb]{\smash{$r=4M$}}}%
    \put(0.60758295,0.23580801){\color[rgb]{0,0,0}\makebox(0,0)[lb]{\smash{$r=5M$}}}%
    \put(0.55109191,0.10311274){\color[rgb]{0,0,0}\makebox(0,0)[lb]{\smash{$\omega=r|_\Sigma$}}}%
    \put(0.76669478,0.47249654){\color[rgb]{0,0,0}\makebox(0,0)[lb]{\smash{$2\partial_v$}}}%
    \put(0.77494734,0.42694067){\color[rgb]{0,0,0}\makebox(0,0)[lb]{\smash{$\frac{\partial_u}{F}$}}}%
    \put(0.82677994,0.16713379){\color[rgb]{0,0,0}\makebox(0,0)[lb]{\smash{$L$}}}%
    \put(0.82311799,0.09006465){\color[rgb]{0,0,0}\makebox(0,0)[lb]{\smash{$\ubar L$}}}%
    \put(0,0){\includegraphics[width=\unitlength,page=2]{cone1a.pdf}}%
    \put(0.14651983,0.18509873){\color[rgb]{0,0,0}\makebox(0,0)[lb]{\smash{}}}%
    \put(0.2848239,0.22253444){\color[rgb]{0,0,0}\makebox(0,0)[lb]{\smash{}}}%
    \put(0.1583737,0.34828593){\color[rgb]{0,0,0}\makebox(0,0)[lb]{\smash{\tiny{$r=0$}}}}%
    \put(0,0){\includegraphics[width=\unitlength,page=3]{cone1a.pdf}}%
    \put(0.15702622,0.20605438){\color[rgb]{0,0,0}\makebox(0,0)[lb]{\smash{\tiny{$r=0$}}}}%
    \put(0,0){\includegraphics[width=\unitlength,page=4]{cone1a.pdf}}%
    \put(0.35154308,0.43151535){\color[rgb]{0,0,0}\makebox(0,0)[lb]{\smash{\tiny{$r=4M$}}}}%
    \put(0.35889611,0.40725019){\color[rgb]{0,0,0}\makebox(0,0)[lb]{\smash{\tiny{$r=5M$}}}}%
    \put(0.34773996,0.45515471){\color[rgb]{0,0,0}\makebox(0,0)[lb]{\smash{$v$}}}%
    \put(0.35093348,0.10134616){\color[rgb]{0,0,0}\makebox(0,0)[lb]{\smash{$u$}}}%
    \put(0,0){\includegraphics[width=\unitlength,page=5]{cone1a.pdf}}%
    \put(0.3703019,0.26964257){\color[rgb]{0,0,0}\makebox(0,0)[lb]{\smash{$\Omega:=\{v=v_0\}$}}}%
    \put(0.02599692,0.53969786){\color[rgb]{0,0,0}\makebox(0,0)[lb]{\smash{$\mathbb{P}$}}}%
    \put(0,0){\includegraphics[width=\unitlength,page=6]{cone1a.pdf}}%
    \put(0.52009598,0.5338225){\color[rgb]{0,0,0}\makebox(0,0)[lb]{\smash{$\Omega$}}}%
    \put(0,0){\includegraphics[width=\unitlength,page=7]{cone1a.pdf}}%
    \put(0.7470713,0.57500959){\color[rgb]{0,0,0}\makebox(0,0)[lb]{\smash{$u=0$}}}%
    \put(0.8161901,0.50294957){\color[rgb]{0,0,0}\makebox(0,0)[lb]{\smash{$u=\frac{g(4M)}{v_0}$}}}%
    \put(0.86173231,0.36932613){\color[rgb]{0,0,0}\makebox(0,0)[lb]{\smash{$u=\frac{g(5M)}{v_0}$}}}%
    \put(0.88677949,0.16470859){\color[rgb]{0,0,0}\makebox(0,0)[lb]{\smash{$\Sigma$}}}%
  \end{picture}%
\endgroup%

%% file: cone3.pdf_tex
\begingroup%
  \makeatletter%
  \providecommand\color[2][]{%
    \errmessage{(Inkscape) Color is used for the text in Inkscape, but the package 'color.sty' is not loaded}%
    \renewcommand\color[2][]{}%
  }%
  \providecommand\transparent[1]{%
    \errmessage{(Inkscape) Transparency is used (non-zero) for the text in Inkscape, but the package 'transparent.sty' is not loaded}%
    \renewcommand\transparent[1]{}%
  }%
  \providecommand\rotatebox[2]{#2}%
  \ifx\svgwidth\undefined%
    \setlength{\unitlength}{448.21113166bp}%
    \ifx\svgscale\undefined%
      \relax%
    \else%
      \setlength{\unitlength}{\unitlength * \real{\svgscale}}%
    \fi%
  \else%
    \setlength{\unitlength}{\svgwidth}%
  \fi%
  \global\let\svgwidth\undefined%
  \global\let\svgscale\undefined%
  \makeatother%
  \begin{picture}(1,0.61213462)%
    \put(0,0){\includegraphics[width=\unitlength,page=1]{cone3.pdf}}%
    \put(0.653827,0.31407925){\color[rgb]{0,0,0}\makebox(0,0)[lb]{\smash{$\Sigma_{s_0}$}}}%
    \put(0.7436939,0.15656256){\color[rgb]{0,0,0}\makebox(0,0)[lb]{\smash{$\Sigma_{s_2}$}}}%
    \put(0.73523351,0.04611591){\color[rgb]{0,0,0}\makebox(0,0)[lb]{\smash{$\Sigma_{s_3}$}}}%
    \put(0,0){\includegraphics[width=\unitlength,page=2]{cone3.pdf}}%
    \put(0.56906721,0.46042043){\color[rgb]{0,0,0}\makebox(0,0)[lb]{\smash{$\Sigma_{s_1}$}}}%
    \put(0.11519952,0.54711424){\color[rgb]{0,0,0}\makebox(0,0)[lb]{\smash{$\Omega$}}}%
    \put(0,0){\includegraphics[width=\unitlength,page=3]{cone3.pdf}}%
    \put(0.05864584,0.43856881){\color[rgb]{0,0,0}\makebox(0,0)[lb]{\smash{$L$}}}%
    \put(0.08541891,0.31872733){\color[rgb]{0,0,0}\makebox(0,0)[lb]{\smash{$\ubar L$}}}%
  \end{picture}%
\endgroup%

%% file: cone4.pdf_tex
\begingroup%
  \makeatletter%
  \providecommand\color[2][]{%
    \errmessage{(Inkscape) Color is used for the text in Inkscape, but the package 'color.sty' is not loaded}%
    \renewcommand\color[2][]{}%
  }%
  \providecommand\transparent[1]{%
    \errmessage{(Inkscape) Transparency is used (non-zero) for the text in Inkscape, but the package 'transparent.sty' is not loaded}%
    \renewcommand\transparent[1]{}%
  }%
  \providecommand\rotatebox[2]{#2}%
  \ifx\svgwidth\undefined%
    \setlength{\unitlength}{468.59812339bp}%
    \ifx\svgscale\undefined%
      \relax%
    \else%
      \setlength{\unitlength}{\unitlength * \real{\svgscale}}%
    \fi%
  \else%
    \setlength{\unitlength}{\svgwidth}%
  \fi%
  \global\let\svgwidth\undefined%
  \global\let\svgscale\undefined%
  \makeatother%
  \begin{picture}(1,0.57862364)%
    \put(0,0){\includegraphics[width=\unitlength,page=1]{cone4.pdf}}%
    \put(0.14225007,0.51643205){\color[rgb]{0,0,0}\makebox(0,0)[lb]{\smash{$\Omega$}}}%
    \put(0,0){\includegraphics[width=\unitlength,page=2]{cone4.pdf}}%
    \put(0.39376429,0.26354215){\color[rgb]{0,0,0}\makebox(0,0)[lb]{\smash{$\omega$}}}%
    \put(0.6039811,0.38822015){\color[rgb]{0,0,0}\makebox(0,0)[lb]{\smash{$\Sigma_{s_0}$}}}%
    \put(0,0){\includegraphics[width=\unitlength,page=3]{cone4.pdf}}%
    \put(0.12230118,0.42602297){\color[rgb]{0,0,0}\makebox(0,0)[lb]{\smash{$l$}}}%
    \put(0.12961787,0.31139525){\color[rgb]{0,0,0}\makebox(0,0)[lb]{\smash{$\ubar L$}}}%
    \put(0.00523471,0.25164267){\color[rgb]{0,0,0}\makebox(0,0)[lb]{\smash{$L$}}}%
    \put(0.700317,0.2138399){\color[rgb]{0,0,0}\makebox(0,0)[lb]{\smash{$\Sigma_\omega$}}}%
    \put(0,0){\includegraphics[width=\unitlength,page=4]{cone4.pdf}}%
    \put(0.4308202,0.13457614){\color[rgb]{0,0,0}\makebox(0,0)[lb]{\smash{$q$}}}%
    \put(0.73202259,0.14920942){\color[rgb]{0,0,0}\makebox(0,0)[lb]{\smash{$\Sigma_{s(q)}$}}}%
    \put(0,0){\includegraphics[width=\unitlength,page=5]{cone4.pdf}}%
    \put(0.03198215,0.1511586){\color[rgb]{0,0,0}\makebox(0,0)[lb]{\smash{$\ubar L$}}}%
  \end{picture}%
\endgroup%

%% file: spherical.pdf_tex
\begingroup%
  \makeatletter%
  \providecommand\color[2][]{%
    \errmessage{(Inkscape) Color is used for the text in Inkscape, but the package 'color.sty' is not loaded}%
    \renewcommand\color[2][]{}%
  }%
  \providecommand\transparent[1]{%
    \errmessage{(Inkscape) Transparency is used (non-zero) for the text in Inkscape, but the package 'transparent.sty' is not loaded}%
    \renewcommand\transparent[1]{}%
  }%
  \providecommand\rotatebox[2]{#2}%
  \ifx\svgwidth\undefined%
    \setlength{\unitlength}{510.06995494bp}%
    \ifx\svgscale\undefined%
      \relax%
    \else%
      \setlength{\unitlength}{\unitlength * \real{\svgscale}}%
    \fi%
  \else%
    \setlength{\unitlength}{\svgwidth}%
  \fi%
  \global\let\svgwidth\undefined%
  \global\let\svgscale\undefined%
  \makeatother%
  \begin{picture}(1,0.57446719)%
    \put(0,0){\includegraphics[width=\unitlength,page=1]{spherical.pdf}}%
    \put(0.43100055,0.55212421){\color[rgb]{0,0,0}\makebox(0,0)[lb]{\smash{$r_0 = 2\tilde M(v_0,r_0)$}}}%
    \put(0.4373379,0.37784727){\color[rgb]{0,0,0}\makebox(0,0)[lb]{\smash{$r=r_1$}}}%
    \put(0.44367526,0.10692587){\color[rgb]{0,0,0}\makebox(0,0)[lb]{\smash{$r>r_1$}}}%
  \end{picture}%
\endgroup%

%% file: quasilocalmass.bbl
\begin{thebibliography}{9}
\bibitem{A}
S. Alexakis, ``The Penrose inequality on perturbations of the Schwarzchild exterior",\\
arXiv:1506.06400 [gr-qc] (2015)
\bibitem{BJM}
H. Bray, J. L. Jaurequi, M. Mars, ``Timeflat surfaces and the Monotonicity of the spacetime Hawking Mass II",\\
arXiv:1402.3287v1[math.DG] (2014)
\bibitem{BHMS}
H. Bray, S. Hayward, M. Mars, W. Simon, ``Generalized inverse mean curvature flows in spacetime",\\
Commun.Math.Phys. 272, 119-138 (2007)
\bibitem{C}
D. Christodoulou, ``Mathematical Problems of General Relativity II",\\
Zurich Lectures in Advanced Mathematics.
\bibitem{G}
E. Gourgoulhon, J.L. Jaramillo, ``A 3+1 perspective on null hypersurfaces and isolated horizons",\\
Phys. Rep. 423, 159-294 (2006)
\bibitem{H}
S.W. Hawking, ``Gravitational radiation in an expanding universe",\\
J. Math. Phys. 9 598-604 (1968)
\bibitem{Ha}
S.A. Hayward, ``Gravitational energy in spherical symmetry",\\
Phys. Rev. D 53, 1938-1949 (1996)
\bibitem{MS1}
M. Mars, A. Soria, ``The asymptotic behaviour of the Hawking energy along null asymptotically flat hypersurfaces",\\
arXiv:1506.01545[gr-qc] (2015)
\bibitem{MS2}
M. Mars, A. Soria, ``On the Penrose inequality along null hypersurfaces",\\
arXiv:1511.06242v1 [gr-qc] (2015)
\bibitem{O}
B. O'Neill, ``Semi-Riemannian Geometry: with application to relativity",\\
Academic Press (1983)
\bibitem{P}
A. Parry, ``A survey of spherically symmetric spacetimes",\\
arXiv:1210.5269 [gr-qc] (2012)
\bibitem{P1}
R. Penrose, ``Gravitational collapse: the Role of General Relativity",\\
Riv. del Nuovo Cim. (numero speciale) 1 (1969) 252-276
\bibitem{P2}
R. Penrose: ``Naked Singularities",\\
Ann. of the N.Y. Acad. of Sci. 224 (1973) 125-134
\bibitem{S}
J. Sauter, Ph.D thesis, ETH Z\"{u}rich (2008)
\bibitem{W}
R.M. Wald, ``General Relativity",\\
The University of Chicago Press (1984)
\bibitem{CWZ}
M. Wang, Y. Wang, X. Zhang, ``Minkowski formulae and Alexandrov theorems in spacetime",\\
arXiv:1409.2190v1 [math.DG] (2014)
\end{thebibliography}
